\documentclass[aps,pre,twocolumn,superscriptaddress,longbibliography,10pt]{revtex4-2}
\usepackage{hyperref}
\usepackage{textcomp}
\usepackage[autostyle]{csquotes}
\usepackage{amssymb}
\usepackage{amsthm}
\usepackage{amsfonts}
\usepackage{amssymb}
\usepackage{footnote}
\usepackage{algorithm}
\usepackage{float}
\usepackage{algpseudocode}
\usepackage{graphicx}
\usepackage{notations}
\usepackage{subfiles}
\usepackage{subfig}
\usepackage{mathtools}
\usepackage{pgfplots}
\usepackage{lipsum}
\usepackage{enumitem}
\pgfplotsset{/pgf/number format/use comma,compat=newest}
\usepackage{url}
\usepackage{tikz}
\usepackage{bbm}
\usepackage{floatflt}
\usepackage{soul} 
\usetikzlibrary{arrows,decorations.pathmorphing,backgrounds,positioning,fit,petri}
\usepackage{comment}
\usepgfplotslibrary{fillbetween}

\newcommand{\E}{\mathbb{E}}
\newcommand{\R}{\mathbb{R}}

\newtheorem{proposition}{Proposition}

\theoremstyle{remark}

\newcommand\iid{\mathrel{\stackrel{\makebox[0pt]{\mbox{\normalfont\tiny iid}}}{\sim}}}

\newcommand\distreq{\mathrel{\stackrel{\makebox[0pt]{\mbox{\normalfont\tiny D}}}{=}}}

\begin{document}
\title{On the phase diagram of extensive-rank symmetric matrix denoising\texorpdfstring{\\}{} beyond rotational invariance}

\author{Jean Barbier}
\email[]{jbarbier@ictp.it}

\author{Francesco Camilli}
\email[]{fcamilli@ictp.it}

\affiliation{The Abdus Salam International Centre for Theoretical Physics\\ 
 Strada Costiera 11, 34151 Trieste, Italy}

\author{Justin Ko}
\email[]{justin.ko@uwaterloo.ca}
\affiliation{University of Waterloo \\ 
Department of Statistics and Actuarial Science, 200 University Ave W, Waterloo, ON, N2L 3G1, Canada}

\author{Koki Okajima}
\email[]{darjeeling@g.ecc.u-tokyo.ac.jp}
\affiliation{The University of Tokyo \\ Department of Physics, Graduate School of Science, 7-3-1 Hongo, Tokyo 113-0033, Japan }


\begin{abstract}
Matrix denoising is central to signal processing and machine learning. Its statistical analysis when the matrix to infer has a factorised structure with a rank growing proportionally to its dimension remains a challenge, except when it is rotationally invariant. In this case the information theoretic limits and an efficient Bayes-optimal denoising algorithm, called rotational invariant estimator \cite{ledoit2011eigenvectors,BABP}, are known. Beyond this setting few results can be found. The reason is that the model is not a usual spin system because of the growing rank dimension, nor a matrix model (as appearing in high-energy physics) due to the lack of rotation symmetry, but rather a hybrid between the two.

In this paper we make progress towards the understanding of Bayesian matrix denoising when the hidden signal is a factored matrix $\bX\bX^\intercal$ that is not rotationally invariant. Monte Carlo simulations suggest the existence of a \emph{denoising-factorisation transition} separating a phase where denoising using the rotational invariant estimator remains Bayes-optimal due to universality properties of the same nature as in random matrix theory, from one where universality breaks down and better denoising is possible, though algorithmically hard. We also argue that it is only beyond the transition that factorisation, i.e., estimating $\bX$ itself, becomes possible up to irresolvable ambiguities. On the theoretical side, we combine mean-field techniques in an interpretable multiscale fashion in order to access the minimum mean-square error and mutual information. Interestingly, our alternative method yields equations reproducible by the replica approach of \cite{SK13EPL}. Using numerical insights, we delimit the portion of phase diagram where we conjecture the mean-field theory to be exact, and correct it using universality when it is not. Our complete ansatz matches well the numerics in the whole phase diagram when considering finite size effects. 

\end{abstract}

\maketitle


\section{Introduction}

{Statistical inference is the powerhouse of modern information processing systems as appearing for instance in signal processing, error correction for communication, data cleaning in neuroscience, biology or finance. Ultimately, the goal is always to separate the signal from the noise in the data at hand, where ``signal'' simply means what is considered the relevant part of the data for a specific application. Given the ubiquity of arrays to represent modern data, a particularly general inference task is the one of matrix denoising. E.g., in recommender systems, rows represent users and columns the products they have rated, or in medicine with patient versus measured physiological features. A basic reflex to make sense of a large data array is to write down its empirical covariance matrix, followed by a preprocessing denoising step.}

{However, denoising a large covariance matrix, modeled as a product of two factors, is a hard problem, both algorithmically, and from the viewpoint of theoretical analysis, especially if its rank is proportional to its dimension as in the present contribution. The applications of matrix denoising range from dictionary learning \cite{dictionaryLearning}, representation learning \cite{bengio2013representation} or sparse coding \cite{olshausen1996emergence,olshausen1997sparse,kreutz2003dictionary,mairal2009} in machine learning and computational neuroscience, to robust principal components analysis \cite{candes2011robust,perry2018optimality}, sub-matrix localization \cite{hajek2017information}, blind source separation \cite{belouchrani1997blind}, matrix completion and recommender systems \cite{candes2009exact,candes2010power,PMF_salakhutdinov} or video processing \cite{mairal2009,Mairal_colorrestoration} in signal processing, but also appear in the analysis of community detection in large graphs \cite{abbe2017community,lelarge2019fundamental,caltagirone2017recovering}. Matrix denoising problems can also be interpreted as ``planted spin glasses'', and are therefore naturally linked to the statistical physics of disordered systems \cite{nishimori2001statistical}, a connection we will exploit extensively here.}

{In this paper we focus on the denoising of symmetric positive definite signal matrices, which can therefore be interpreted as empirical covariance matrices, corrupted by a symmetric Gaussian matrix, which is a standard modeling assumption for noise. In order to study the \emph{fundamental} limits of inference, i.e., the information theoretically optimal performance achievable independently of the computational cost to do so, the correct framework is the Bayesian one. Denoising is thus formulated as a Bayesian inference problem over an empirical covariance matrix $\bX\bX^\intercal$ given a noisy data matrix $\bY\in\RR^{N\times N}$ generated as
\begin{align}\label{eq:channel0}
    \bY=\sqrt{\frac{\lambda}{N}}\bX\bX^\intercal+\bZ.
\end{align}
The rectangular matrix $\bX\in\RR^{N\times M}$ is the hidden factor (also called signal, a terminology we will also use for the square matrix $\bX\bX^\intercal\in\RR^{N\times N}$ of rank at most $\min(N,M)$ if there is no ambiguity) and $\bZ\in\RR^{N\times N}$ is a noise drawn from the Gaussian orthogonal ensemble (GOE), i.e., a symmetric matrix with independent Gaussian entries on its upper triangular part $Z_{ij}=Z_{ji}\sim\mathcal{N}(0,1+\delta_{ij})$. We consider the ideal setting where the signal $\bX$ has independent and identically distributed (i.i.d.) components $(X_{i\mu})_{i\leq N,\mu\leq M}$ drawn from a sub-Gaussian prior distribution $P_X$ with zero mean and unit variance.} 

{When the prior is different than a centered Gaussian, the signal $\bX\bX^\intercal$ (and thus the data) is \emph{not} rotational invariant. Rotational invariance means in this context that the probability density $P(\bX\bX^\intercal)$ would be equal to $P(\bO\bX\bX^\intercal\bO^\intercal)$ for any $N\times N$ orthogonal matrix $\bO$. The setting where the prior $P_X$ breaks this invariance, as well as the differences this induces w.r.t. the rotational invariant Gaussian prior $P_X=\mathcal{N}(0,1)$, will be of particular importance in this paper.}

\vspace{5pt}
\noindent {\textbf{The well known: denoising a rank-one matrix.} Model \eqref{eq:channel0} may look rather special at first sight. However, a long line of work in the statistical physics of high-dimensional inference has proven that the analysis of prototypical models allow us to understand general mechanisms, similarly to how dissecting the Ising model in physics has paved the way to understand phase transitions and critical phenomena way beyond it. The connection between the Ising model and matrix denoising is actually more than an analogy. Indeed, the rank-one version $M=1$ of matrix denoising (also called the spiked Wigner model) with signal entries taking binary values $X_i=\pm1$ is precisely the planted version of the celebrated Sherrington-Kirkpatrick mean-field spin glass \cite{sherrington1975solvable} (see \cite{nishimori2001statistical,zdeborova2016statistical} for the mapping), whose analysis is at the core of a fruitful forty-year research field involving physicists and mathematicians \cite{mezard1990spin,Panchenko2013}.}

{The special features of this spin model inherent to the presence of a planted signal (or equivalently, ``on the Nishimori line''), first analysed using spin glass techniques by Nishimori \cite{nishimori2001statistical}, have triggered numerous works by various communities. Spiked matrix models (i.e., low-rank perturbations of much larger rank matrices, in the present case the noise) are naturally a whole subject in random matrix theory. It was first understood in this context that the presence of the signal may yield a phase transition phenomenon now known as the Ben Arous-Baik-Péché (BBP) transition \cite{baik2005phase,baik2006eigenvalues,BENAYCHGEORGES2011} at a critical value of the signal-to-noise ratio (SNR) $\lambda$. This is one of the most studied examples of phase transition in information processing. The top part of \figurename~\ref{fig:lowRank} illustrates the BBP transition: for SNR greater than a critical value, here denoted $\lambda_{\rm Algo}$, an outlier eigenvalue of the data $\bY$ escapes the semicircular bulk of eigenvalues associated mostly with the noise component. A simple spectral estimator corresponding to the associated eigenvector then has a non-trivial overlap with $\bX$. Before this point, however, no eigenvector of the data aligns with the signal, and spectral estimation is thus doomed. In the statistics literature, the spiked Wigner model was introduced as a statistical model for sparse principal components analysis (PCA) \cite{johnstone2001distribution,johnstone2004sparse,zou2006sparse,johnstone2012consistency}.} 

\begin{figure}
    \centering    \includegraphics[width=1\linewidth]{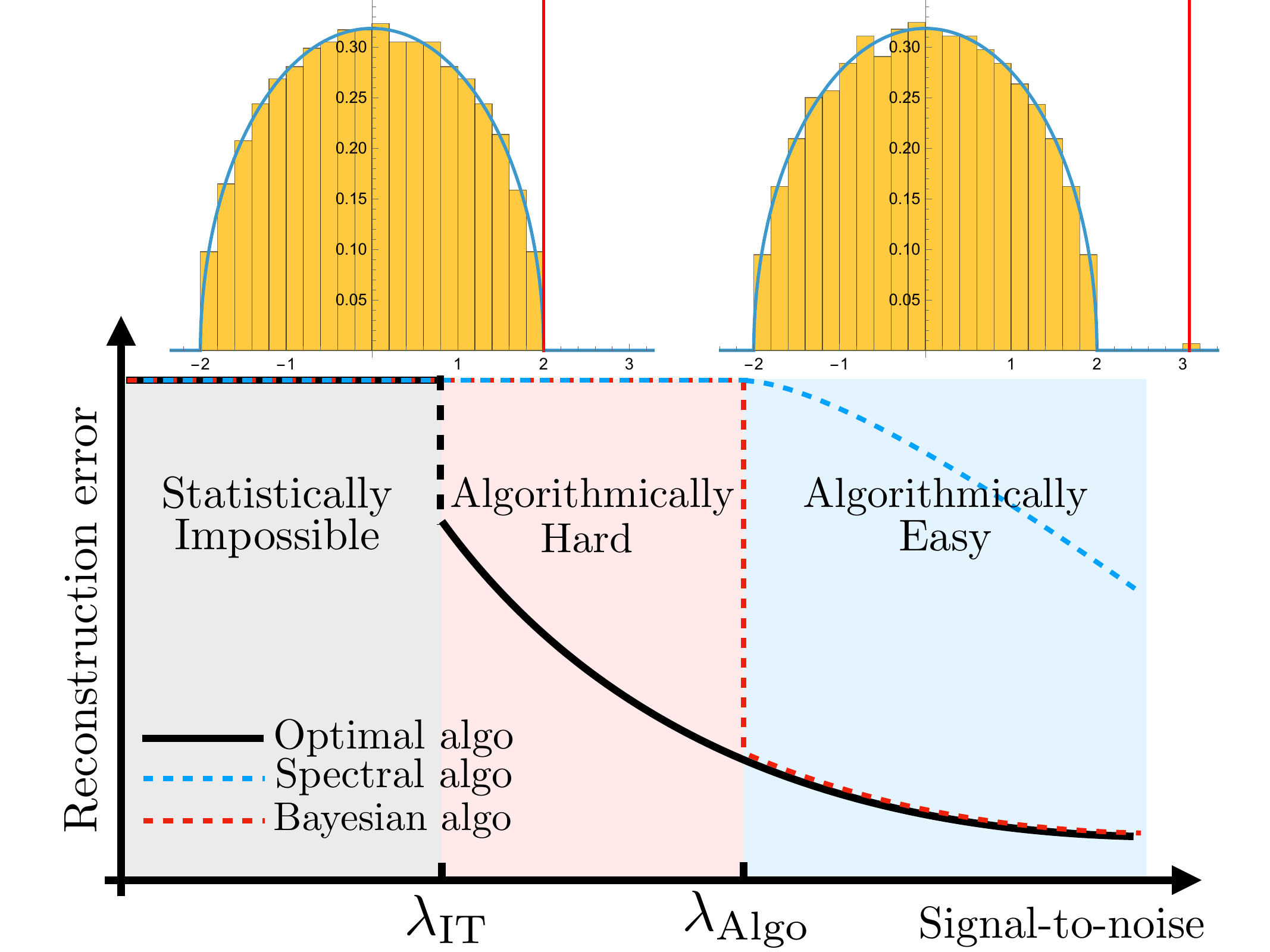}
    \caption{{The top illustrates the Ben Arous-Baik-Péché transition occurring in the rank-one spiked matrix model. In orange is the histogram of eigenvalues of $\bY$ when $M=1$. Before the transition it closely matches Wigner's semicircular law (blue) and its top eigenvalue (red line) sticks to the endpoint of the bulk of eigenvalues. All eigenvectors of $\bY$ have a vanishing $o_N(1)$ overlap with the hidden signal $\bX$. The BBP transition is marked by an outlier eigenvalue detaching from the bulk. The associated eigenvector aligns non-trivially with $\bX$ and can thus serve as spectral estimator. At the bottom is the phase diagram of the model considered from the Bayesian (information theoretic) perspective, with sparse prior $P_X= \rho \delta_0+\tau \delta_{-a}+(1-\tau-\rho) \delta_{b}$ for a proper choice of parameters (see, e.g., \cite{XXT,miolane2019fundamental}). An impossible, hard, and easy inference phases appear for typical realisations of \eqref{eq:channel0} as the signal strength increases, delimited by information theoretic and algorithmic transitions. The behavior of the Bayes-optimal, spectral, and Bayesian approximate message-passing algorithm are shown. For this example of prior $P_X$, the transitions for the message-passing and spectral algorithms match, but the Bayesian algorithm outperforms the spectral one as it exploits the prior, while the latter does not.}}
    \label{fig:lowRank}
\end{figure}

{Concerning the Bayesian viewpoint on the model, as considered in the present paper, multiple approaches are now available for the computation of the main information theoretic quantities, in particular the mutual information between signal and data and the minimum achievable mean square error. We mention the replica method \cite{mezard1990spin,lesieur2015mmse,lesieur2017constrained} among the non-rigorous ones, and proofs based on the analysis of message-passing algorithms \cite{6875223,deshpande2017asymptotic}, spatial coupling \cite{XXT,XX_long}, the interpolation method and its adaptive version \cite{koradamacris,BarbierM17a,Barbier_2019,reeves2020information}, the cavity method \cite{ASS,lelarge2019fundamental}, large deviation arguments \cite{el2018estimation} or partial differential equations approaches \cite{mourrat_Hamilton_spike,chen2022statistical}. This influential series of works has led to a rich and general picture summarized by the bottom part of \figurename~\ref{fig:lowRank}.}

{The phase diagram shown in the figure possesses three distinct regions. In the ``impossible phase'' at low $\lambda$, for a typical large realisation of the probabilistic model \eqref{eq:channel0} with $M=1$, the data contains a vanishing information about the signal and thus inference, whether computationally efficient or not, is impossible. Consequently, even the Bayes-optimal estimator (corresponding to the Bayes posterior mean, black solid curve) yields trivially bad reconstruction error. Beyond the information theoretic transition $\lambda_{\rm IT}$ the picture changes: here, the data typically contain information and inference becomes in theory possible. However, no practical algorithm running with a sub-exponential complexity in $N$ is known. Those we have at hand, including the aforementioned spectral estimator and a Bayesian, prior-aware, approximate message-passing (AMP) algorithm \cite{kabashima2003cdma,donoho2009message,rangan2012iterative,lesieur2017constrained} are all failing. Finally, beyond the algorithmic threshold $\lambda_{\rm Algo}$, the signal is strong enough to be efficiently extracted from the data and, in particular, the Bayesian AMP algorithm suddenly becomes optimal in spite of its low computational cost, following a phase transition in its performance. In this region inference is thus ``easy''.} 

{Such a phase diagram with a hard region (also called computational-to-statistical gap) and associated discontinuous transitions in algorithms performance, is not specific to matrix denoising but is actually typical in high-dimensional inference \cite{zdeborova2016statistical,BarbierGLM-PNAS}. Their study is the central question in average-case complexity theory in modern computer science. And web-of-reductions exist that prove that if the hard phase is fundamental (in the sense that there exists no polynomial-in-$N$ complexity algorithm for signal extraction in this phase), then a plethora of other interconnected inference problems must also be hard \cite{brennan2018reducibility,brennan2020reducibility}. The analyses of prototypical matrix denoising models therefore have a conceptual impact that goes way beyond their specific instance.}

{We mention that numerous generalisations of the spike Wigner model have recently been considered \cite{PNAS_structured_PCA,barbier2025information,barbier2022price,barbier2024information,alberici2021MSKNL,albericiDBM,zhang-marco2024matrixdenoising}.}

\vspace{5pt}
\noindent {\textbf{The challenge: extensive-rank matrix denoising, beyond rotational invariance.}\ \ In this paper we aim at going beyond low-rank matrix inference (i.e., $M=O(1)$ while $N\gg 1$), by considering the much more challenging extensive-rank setting, that is $N,M=M_N$ both large and verifying
\begin{align*}
    \lim_{N\to+\infty}\frac MN=\alpha \in (0,+\infty).
\end{align*}
In this scaling regime much less is known; we refer to Table 1 in \cite{Jean_Farzad_matrixinference_2024} for a compact account of the current status of the problem. 
This scaling is natural in random matrix theory.
It is indeed the one where the signal is not anymore a ``perturbation'' (due to its low-rank nature) of the noise, but rather a ``macroscopic deformation'' of it. In this case, the BBP phenomenology is not valid anymore, because for any $\lambda>0$, the spectrum of the data deviates macroscopically from the semicircular law. Many questions can then be answered using random matrix theory. For instance, the data asymptotic spectral density can be obtained using free probability \cite{potters2020first}.} 

{In this paper, however, we are not interested in typical random matrix theory questions, but rather in information theoretic ones, usually more akin to spin glass approaches. A major difficulty is that when interpreted in statistical physics terms, the model \eqref{eq:channel0} is \emph{not} a planted version of a well-known spin model as when $M=1$, because we have to treat a whole matrix rather than a vector of variables as usual. This induces novel difficulties.}

{The standard strategy to access the information theoretic quantities in statistical mechanics of inference is to cast the problem as the computation of a free entropy (i.e., log-partition function) associated with the Gibbs-Boltzmann distribution of the model, which is the posterior distribution given the data in the context of inference, see \eqref{eq:posterior}. In doing so, we will see in Section~\ref{sec:setting} that the partition function of extensive-rank matrix denoising resembles a \emph{matrix model}. A prototypical matrix model is a partition function of the type
\begin{align}\label{eq:one-matrix-integral}
    \mathcal{Z}_{\rm MM}(\bA):=\int_{\mathbb{H}_N} d\bF\exp\big(-\Tr\bF^2+\Tr\bA\bF\big),
\end{align}
where $\mathbb{H}_N$ is the space of Hermitian $N\times N$ matrices, $d\bF$ denotes the Lebesgue measure over it, and the source $\bA\in \mathbb{H}_N$. This one-matrix integral can be solved in the large $N$ limit by the character expansion method \cite{Migdal_Character_expansion,KAZAKOV_character_Gravity} and by means of the Harish Chandra-Itzykson-Zuber (HCIZ) integral \cite{harish1957differential,itzykson1980planar} defined later in \eqref{eq:HCIZ-def}. Matrix models appear in high energy physics, in particular in Yang-Mills theories with internal symmetry groups of diverging dimension \cite{KAZAKOV-Migdal_Gauge}, in random matrix theory approaches to quantum gravity \cite{Ambjorn1997_quantum_geometry_RMT,Francesco_1995_RMT_gravity}, quantum chromo-dynamics \cite{Verbaarschot_RMT_QCD}, statistical mechanics of planar random lattices \cite{KAZAKOV_Ising_Planar,Kazakov_almostflat_planar,BOULATOV1987379,KAZAKOV_Potts_Planar_lattice} or when counting planar diagrams \cite{Brezin_Parisi_PlanarDiagrams,ZVONKIN_map_enumeration}. We refer the interested reader to the review \cite{kazakov2000sreview}.}

{The known approaches for computing partition functions of matrix models require a Lebesgue measure over the matrix elements or a rotational invariant one. This allows integration of the $\Theta(N^2)$ rotational degrees of freedom using the HCIZ integral, leaving only $N$ variables for the eigenvalues. If the exponent of the integrated matrix model stays $\Theta(N^2)$, then a final saddle point over eigenvalues (that are only $N$) is justified to finish the computation.}

{In the present paper, however, we can no longer rely on rotational symmetry, because it is in general broken by the prior measure $P_X$ as we shall see. Therefore, the challenge that we tackle here is to provide a framework for analysing information theoretic limits of matrix denoising, cast as a statistical mechanics model looking like a matrix model (or vice versa). In doing so, we will not be able to rely on standard spin glass approaches, as it is a \emph{matrix} rather than a vector spin system, nor on those from matrix models due to the breaking of rotational invariance. Nevertheless, we will see that both types of techniques are key when properly adapted. This is the main contribution of the paper, with the ultimate goal of drawing the equivalent of the phase diagram in \figurename~\ref{fig:lowRank} but in the extensive-rank case $M\approx \alpha N$.} 

{Finally, this scaling regime is not only interesting due its challenging nature, but is also practically relevant. Many applications mentioned at the beginning in signal processing and learning naturally belong to this regime. Importantly, like it was realised that low-rank spiked matrix models play an important role in modern neural networks \cite{ba2022high,JMLR:v25:23-1543,wang2024nonlinear}, so does their extensive-rank counterpart, see the recent works \cite{maillard2024bayes,2025_nn_quadratic_Barbier}, which exploit a deep connection to extensive-rank matrix denoising in order to analyse modern scaling regimes of neural networks.}

\section{Bayesian setting, related literature and main contributions}\label{sec:setting}

Before properly defining the setting let us introduce few notations and information theoretic basic definitions.

Symbol $\sim$ expresses that a random variable is drawn from a certain law. $X_{i\mu}\iid P$ means that $(X_{i\mu})_{i,\mu}$ are independent and identically distributed according to $P$. For a matrix $\bx=(x_{i\mu})_{i\leq N,\mu\leq M}$ we denote  $dP_X(\bx)=\prod_{i,\mu}dP_X(x_{i\mu})$. Vectors are columns, their transpose are rows. { $\bx_i=(x_{i\mu})_{\mu\leq M}\in\mathbb{R}^M$ is the vector of the elements of $\bx$ on row $i$. Consistently, $\bx_i$ will still be referred to as a column vector, even if it contains row elements. Analogously $\bx_\mu=(x_{i\mu})_{i\leq N}\in\mathbb{R}^N$ is the vector of the elements of $\bx$ on column $\mu$.} We denote the inner product between vectors equivalently as $\bu\cdot\bv =\bu^\intercal \bv=\sum_{i} u_iv_i$. Symbol $\distreq$ means equality in law. The norm $\|\bA\|=(\sum_{ij}A_{ij}^2)^{1/2}$ is the Frobenius norm. The limit $\lim_{N\to \infty}$ means that both $N,M$ diverge with $M=\alpha N +o_N(1)$ and is called \emph{thermodynamic limit}. $\mathcal{N}(m,\sigma^2)$ is the density of a Gaussian with mean $m$ and variance $\sigma^2$.

{
The mutual information between two random variables $X,Y$ with joint distribution $P_{X,Y}$ and marginals $P_X,P_Y$ is the Kullback-Leibler divergence $I(X;Y)=I(Y;X):=D(P_{X,Y}\|P_X\otimes P_Y)$. For our purposes, $Y=\bY$ will be a continuous r.v.\ whereas $X=\bX$ can be discrete or continuous, but the conditional density $P_{Y\mid X}(y\mid x)$ will always be accessible and well defined. In this regards, the MI can also be recast as 
\begin{align*}
    I(X;Y)&=\int dP_X(x)dy\,P_{Y\mid X}(y\mid x)\ln\frac{P_{Y\mid X}(y\mid x)}{P_Y(y)}\\
    &=\mathcal{H}(Y)-\mathcal{H}(Y\mid X),
\end{align*}
where $\mathcal{H}(Y)$ is the differential entropy and $\mathcal{H}(Y\mid X)$ the conditional differential entropy. 
The same symbol $\mathcal{H}$ will also be used for discrete random variables, in which case it denotes the Shannon entropy. We refer the reader not familiar with these notions to \cite{cover1999elements,polyanskiy2025information} for more details.}

\subsection{Bayesian matrix denoising: setting}

Consider the data generating model \eqref{eq:channel0}. In order to perform Bayesian inference, the first step is to write down the posterior distribution of the signal given the data using Bayes' law. As the signal has a known factorised structure, the posterior is a probability distribution over the factor $\bX$ alone. Given observation \eqref{eq:channel0} the posterior $dP(\bx\mid\bY):=\mathbb{P}(\bX\in [\bx,\bx +d\bx)\mid\bY)$ reads
\begin{align}\label{eq:posterior}
    dP(\bx\mid\bY)=\frac{\exp\big(-\frac{1}{4}\|\bY-\sqrt{\frac{\lambda}{N}}\bx\bx^\intercal\|^2\big)}{C(\bY)}dP_X(\bx),
\end{align}with $C(\bY)$ a normalization.
We denote the expectation w.r.t. the posterior using the bracket notation
$$\langle\,\cdot\,\rangle= \langle\,\cdot\,\rangle_\bY:=\EE[\,\cdot\mid\bY].$$
Expectation $\EE$ is w.r.t. to the data $\bY$ generated as \eqref{eq:channel0}, or equivalently, w.r.t. the independent variables $(\bX,\bZ)$.
The posterior \eqref{eq:posterior} gathers all information on the generating process of the signal and data: the factorised form of $\bX\bX^\intercal$, the size $N\times M$ of the unknown factor $\bX$, its prior law $P_X$, the Gaussian nature of the symmetric noise, as well as the SNR $\lambda$. This setting where the posterior is known exactly to the statistician is said to be \emph{Bayes-optimal}. This ideal case is the appropriate one for studying the fundamental performance limits in inference.

For the information theoretical analysis, one of the main goals is the computation of the mutual information (MI) density, i.e., per signal element:
\begin{align}    \frac{I(\bX\bX^\intercal;\bY)}{MN}=\frac{\mathcal{H}(\bY)}{MN}-\frac{\mathcal{H}(\bY\mid \bX\bX^\intercal)}{MN}=\frac{\mathcal{H}(\bY)}{MN}-\frac{\mathcal{H}(\bZ)}{MN}\nonumber
\end{align}
where, given that the normalised distribution of the data is $P(\bY)=C(\bY)/\prod_{i\le j\le N}\sqrt{2\pi(1+\delta_{ij})}$, the differential entropy of the data is
\begin{align*}
    \mathcal{H}(\bY)&
    =\frac{1}{2}\sum_{i\leq j}\ln (2\pi(1+\delta_{ij}))+\frac{1}{4}\EE\|\bY\|^2\\
    &\quad -\EE\ln \int_{\mathbb{R}^{N\times M}} dP_X(\bx)e^{\frac{1}{2}\sqrt{\frac{\lambda}{N}}\Tr\bY\bx\bx^\intercal-\frac{\lambda}{4N}\Tr(\bx\bx^\intercal)^2}
\end{align*}
and the noise' one is $\mathcal{H}(\bZ)=\frac{1}{2}\sum_{i\leq j}\ln (2\pi e(1+\delta_{ij}))$.
Noticing that $\frac{1}{MN^2}\EE\|\bX\bX^{\intercal}\|^2
=1+\alpha+O(1/N)$ for a centered unit variance prior, the MI density therefore reads 
\begin{align}\label{eq:MI_definition_complete}
    &\frac{I(\bX\bX^\intercal;\bY)}{MN}=\frac{\lambda}{4}(1+\alpha)+O(1/N)\\&\quad-\frac{1}{MN}\EE\ln \int_{\mathbb{R}^{N\times M}} dP_X(\bx)e^{\frac{1}{2}\sqrt{\frac{\lambda}{N}}\Tr\bY\bx\bx^\intercal-\frac{\lambda}{4N}\Tr(\bx\bx^\intercal)^2}.\nonumber
\end{align}
{Observe that the MI between $\bX\bX^\intercal$ and $\bY$ is the same as that between $\bX$ and $\bY$ because $\bY$ is a stochastic function of $\bX$ only through $\bX\bX^\intercal$.} In statistical mechanics of inference the MI plays the same role as the free energy, and is thus the thermodynamic potential of the problem. In the thermodynamic limit, its non-analytic point(s) possibly encountered when tuning the control parameters $(\alpha,\lambda)$ correspond(s) to the phase transition(s), whose characterisation is a central goal in this paper.

The integral in the MI \eqref{eq:MI_definition_complete} is the logarithm of the partition function of a matrix model like \eqref{eq:one-matrix-integral} where $\bF$ is replaced by $\bx\bx^\intercal$ which is symmetric positive semi-definite and where the data $\bY$ plays the role of source $\bA$. A major difference is that the reference measure of the integral is not rotational invariant in general.

The noise being additive Gaussian, the MI is linked to the MMSE on $\bX\bX^\intercal$ via a thermodynamic identity called \enquote{I-MMSE relation} \cite{Verdu_I-MMSE}:
\begin{align}\label{eq:I-MMSE}
    \frac{d}{d\lambda}\frac{I(\bX\bX^\intercal;\bY)}{MN}=\frac{1}{4MN^2}\EE\|\bX\bX^\intercal-\langle\bx\bx^\intercal\rangle\|^2.
\end{align}
The r.h.s. is the ${\rm MMSE}={\rm MMSE}_N$ divided by four, with
\begin{align}
    \label{MMSEuseful}
    {\rm MMSE}&:=\frac{1}{MN^2}\EE\|\bX\bX^\intercal-\langle\bx\bx^\intercal\rangle\|^2\nonumber\\
    &\ =\frac{1}{M}\EE\Big\|\frac{\bX\bX^\intercal}{N}\Big\|^2-\frac{1}{M}\EE\Big\|\frac{\langle\bx\bx^\intercal\rangle}{N}\Big\|^2,
\end{align}
using Proposition~\ref{prop:Nishi_ID} in the appendix for the second line.

\subsection{Related works}\label{sec:related_works} 

In the past fifteen years, an interest in the statistical analysis of the challenging proportional scaling regime, $M/N\to\alpha$ as $N\to+\infty$, has grown. Given the numerous conjectures revolving around extensive-rank matrix denoising, we now review relevant results.

\vspace{5pt}
\noindent \textbf{Rotational invariant setting.} \ \ A major contribution came from random matrix theory. Ledoit and Péché \cite{ledoit2011eigenvectors} and Bun et al. \cite{BABP} introduced a rotational invariant estimator
(which could be more appropriately called rotational equivariant estimator as noted in \cite{semerjian2024matrix}) for optimal denoising of rotationally invariant signals. Given data as in equation \eqref{eq:channel0}, the RIE is the best estimator diagonalizing on the eigenbasis of the data, collected in an orthogonal matrix $\bO_\bY$. That given, it just remains to clean the eigenvalues by means of the \emph{shrinkage formula}
\begin{align} \xi_i(\bY):=\frac{\gamma_{\bY,i}-2\pi\mathsf{H}[\rho_Y](\gamma_{\bY,i})}{\sqrt{\lambda}}\label{RIEshrinkage}
\end{align}
where $(\gamma_{\bY,i})$ are the eigenvalues of $\bY/\sqrt{N}$, the spectral density $\rho_Y$ is the corresponding weak limit of $\frac1N\sum_{i}\delta_{\gamma_{\bY,i}}$ as $N\to+\infty$,  and $\mathsf{H}$ is the Hilbert transform $\mathsf{H}[\rho_Y](x):=\frac{1}{\pi}\text{P.V.}\int d\rho_Y(y)/(x-y)$. The RIE is then composed as
\begin{align*}
\Xi(\bY):=\bO_\bY\,\text{diag}\big((\xi_i(\bY))_{i\leq N}\big)\,\bO_\bY^\intercal.
\end{align*}

{The RIE exploits the factorised structure of the signal and its rank in \eqref{RIEshrinkage} through its Marchenko-Pastur (MP) asymptotic spectral law $\rho_{XX^\intercal}$, needed to compute $\rho_Y=\rho_{XX^\intercal}\boxplus \rho_Z$ by free convolution $\boxplus$ with the noise semicircular spectral law $\rho_Z$ (see \cite{potters2020first}), and which depends on $\alpha$ and the first and second moments of $P_X$; higher order moments are irrelevant to $\rho_{XX^\intercal}$ by spectral universality \cite{10.1214/13-AAP939}. This is, however, apriori partial compared to the complete prior knowledge that the signal is of the form $\bX\bX^\intercal$ with $X_{i\mu}\iid P_X$, as exploited by a Bayesian estimator. E.g., this washes out the dependence on local statistics of the eigenvalues or on the moments of $P_X$ beyond the second. Yet, the RIE is optimal when the signal $\bX\bX^\intercal$ is rotational invariant \cite{ledoit2011eigenvectors,BABP}.
}

The special case of Gaussian prior $P_X=\mathcal{N}(0,1)$ of the present model is the only one inducing rotational invariance of the signal and data. In this case the log-partition function and MI of the model are exactly computable by means of the HCIZ integral defined as
\begin{align}\label{eq:HCIZ-def}
    \mathcal{Z}_{\rm HCIZ}(\bA,\bB):=\int d\mu(\bO)\exp\Big(\frac{\beta N}{2}\Tr\bO\bA\bO^\intercal\bB\Big)
\end{align}where $\beta=2$ if $\bA,\bB$ are $N\times N$ Hermitian matrices, and $\beta=1$ if they are real symmetric. Respectively the integral is over the unitary group $\mathbb{U}(N)$ when $\beta=2$ or orthogonal one $\mathbb{O}(N)$ when $\beta=1$, w.r.t. the corresponding uniform Haar measure $\mu$. Remarkably, for $\beta=2$ this integral is solved in closed form \cite{itzykson1980planar} for any $N$, and its asymptotics was also computed first by Matytsin \cite{Matytsin_1994} and then rigorously proved in \cite{guionnet2002large,guionnet2004first}. The asymptotic of this integral when $\beta =1$ is also known \cite{zuber2008large}.

Starting from the HCIZ integral, the simplifications in \cite{Jean_Farzad_matrixinference_2024,perturbative_Maillard21} show that the MMSE with Gaussian prior reads
\begin{align}
    \lim_{N\to\infty}\text{MMSE}&=\frac{1}{\lambda\alpha}\Big(1-\frac{4\pi^2}{3}\int \rho_Y(y)^3\,dy\Big),\label{MMSEspherical}
\end{align}
Integrating the I-MMSE relation \eqref{eq:I-MMSE} in the limit, one gets the asymptotic MI density $\iota^{\rm G}$ for matrix denoising when $\bX$ has i.i.d. Gaussian $\mathcal{N}(0,1)$ entries:
\begin{align}
    \iota^{\rm G}(\alpha,\lambda)&:=\lim_{N\to+\infty}\frac{I(\bX\bX^\intercal;\bY)}{MN}\label{MIspherical}\\&\ =\frac{1}{8\alpha}+\frac{1}{2\alpha}\int dx dy
    \,\rho_Y(x)\rho_Y(y)\ln|x-y|.\nonumber
\end{align}
For its numerical computation at SNR $\lambda$, a convenient approach is to use the I-MMSE relation, i.e., to numerically integrate $4$ times the MMSE \eqref{MMSEspherical} w.r.t. $\lambda'\in[0,\lambda]$ using $\rho_Y$ given explicitly in \cite{Jean_Farzad_matrixinference_2024} for a Wishart signal $\bX\bX^\intercal/\sqrt N$.

Optimal denoising and the RIE for the rectangular, non-symmetric, version of the problem was proposed recently in \cite{troiani2022optimal,pourkamali2023rectangular,pourkamali2024rectangular,landau2023singular} based on the rectangular HCIZ integral \cite{guionnet2021large}. The work \cite{pourkamali2024bayesian} provided an interesting approach to the factorisation problem discussed in Section~\ref{sec:findings_overlap} based on the RIE. The authors of \cite{bodin2023gradient} managed to exploit rotational invariance in order to analyse the dynamics of gradient flow for the denoising of a Wishart matrix.

\vspace{5pt}
\noindent \textbf{Statistical mechanics approaches and beyond.} \ \ Moving away from the rotationally invariant setting, among the first works based on statistical physics we mention \cite{SK13EPL} which proposed a replica approach with an original \enquote{Gaussian ansatz} (see App.~\ref{appendix:Kabashima_derivation}) for dictionary learning in an optimisation setting. It was followed by the concurrent works \cite{SK13ISIT,KMZ_DL-2013,Marc-Kabashima} applying the same method in a Bayesian setting similar to ours. These works, however, seem not to take into account some correlations between variables that turn out to be relevant \emph{in general} in the high-dimensional limit, as understood in \cite{comments,schmidt:tel-03227132} and further detailed by some of the previous authors \cite{perturbative_Maillard21}. Although it is true that in \cite{SK13EPL,SK13ISIT,KMZ_DL-2013,Marc-Kabashima} relevant correlations are neglected when considering the Gaussian prior, our numerical evidence strongly suggests that this is \emph{not} the case in the factorisation phase we identify. The latter cannot appear for the Gaussian prior, and has thus eluded subsequent studies \cite{schmidt:tel-03227132} that aimed at matching the formulas in \cite{SK13EPL,SK13ISIT,KMZ_DL-2013,Marc-Kabashima} with the theoretical predictions for Gaussian prior by means of the asymptotic expansion of Matytsin \cite{Matytsin_1994}. However, \cite{schmidt:tel-03227132} was missing a justification of their final result, that was later provided in the concurrent works \cite{barbier2022DL,perturbative_Maillard21}, and rigorised in \cite{Jean_Farzad_matrixinference_2024}.

The authors of \cite{perturbative_Maillard21} provide a perturbative attempt to the computation of the MI based on Plefka-Georges-Yedidia (PGY) expansion \cite{Plefka1982ConvergenceCO,georges1991expand}. The PGY formalism yields $M N$ equations that should approximate to a certain perturbation order the TAP equations \cite{thouless1977solution}. When comparing the PGY estimator, truncated at order 2 and 3 of the perturbative expansion, with the exact solution for Gaussian prior given by the HCIZ integral, the authors still measure a significant discrepancy, a sign that higher perturbative orders may be non-negligible. The authors of \cite{barbier2022DL} attempted an innovative spectral replica symmetric ansatz, yet leading to an unsatisfactory solution as it is not possible to evaluate it and it is thus hard to extract predictions from it. 

A step towards the phase diagram for generic prior was done in the series of works \cite{camilli2023new,camilli2023matrix,camilli2024decimation}, where the authors map the model into a sequence of neural networks for associative memory. They introduce a \emph{decimation scheme}, greatly inspired by the replica computation for the Hopfield model \cite{Hopfield82} carried out by Amit, Gutfreund and Sompolinsky in \cite{Amit1}, and by the unlearning mechanism \cite{Hopfield_unlearning,vanHemmen_unlearning,Benedetti_unlearning}. \cite{camilli2023matrix} provides a phase diagram for the decimation scheme and numerical evidence that when the SNR is high, $\alpha$ is low and the prior is informative enough (such as Rademacher), decimation is able to outperform the RIE, despite requiring a computational time that is exponential in the size $N$ of the system. The idea of decimation is that of retrieving the columns $(\bX_\mu)$ (i.e., \emph{patterns} in an associative memory interpretation) one at a time. In order to do this one needs to simplify the term $-\frac{\lambda}{2N}\Tr(\bx\bx^\intercal)^2$ in the Hamiltonian \eqref{eq:Hamiltonian}. The latter favours posterior samples whose columns $(\bx_\mu)$ remain quasi-orthogonal to each other, a sort of \enquote{caging effect} preventing patterns estimators to collapse onto the same ground truth patterns, thus boosting the joint recovery of all of them. This rigidity, which helps in the estimation, is completely given up on by decimation in favour of an analytically tractable model, paying the price of a brittle estimation: it works only for rather high SNR. On the contrary, in extensive-rank matrix denoising it is possible to give a non-trivial estimate of $\bX\bX^\intercal$ as soon as $\lambda>0$ via the RIE. Therefore, a comprehensive theory should be able to take into account this caging effect.

Recently, \cite{semerjian2024matrix} put forward a conjecture of \emph{universality}, intended as independence of the MMSE and MI on the prior $P_X$ (as long as it is centered and with unit variance), which can be formulated in a weak and a strong version. The weak one states that the RIE is optimal among finite-degree polynomial estimators in the data matrix elements, whereas the strong one would posit the universality to hold for all estimators. The strong version would imply that, under Gaussian additive noise, the MI is the one corresponding to the Gaussian prior case of \cite{guionnet2002large,barbier2022DL,
perturbative_Maillard21,Jean_Farzad_matrixinference_2024}. However, this is in contrast with the evidence produced with decimation algorithms. From our analysis a richer picture emerges. We disprove strong universality rigorously for any $\alpha>0$. More precisely, we can prove that there exists an SNR beyond which the RIE cannot be optimal. Nevertheless, the conjecture is not to be completely discarded. In fact, we argue that universality of the MI and MMSE holds for low enough SNR, as supported by next section's numerics. A possible explanation for the breaking of validity of the analysis in \cite{semerjian2024matrix} arguing strong universality is the following. The argument is based on the analysis of low-degree polynomials \cite{kunisky2019notes}. The author is able to show that no algorithm based on degree-$D$ polynomials of the data entries with $D\le 3$ can outperform the RIE when $N\to+\infty$. He also suggests that his argument extends to larger $D$.
Then, \emph{assuming that the $N\to+\infty$ and $D\to+\infty$ can be exchanged}, he concludes that no function of the data (approximated by a polynomial of arbitrary degree) can thus serve as better estimator. However, commutation of limits can break at a phase transition, as mentioned by the author. Therefore, it may be that the argument is valid up to a critical SNR, a point beyond which more care should be taken with the large degree and size limits.

The universality-breaking threshold $\lambda_c$ we identify in Section~\ref{sec:completeansatz} depends on $\alpha$. In particular, recent literature supports the fact that when $\alpha\to 0$ the universality phase should shrink below $\lambda_c(\alpha=0)=:\bar\lambda_c>0$ (see \figurename~\ref{fig:phase_diag}). In \cite{barbier2024information} the authors prove that with a factorised prior, under some suitable hypothesis, the MI and MMSE reduce to those of the standard rank-1 formulas when $M=M_N$ grows sufficiently slowly, which is prior dependent and thus clearly not universal above the phase transition. Note that, for SNRs lower than the critical one, universality holds trivially at low rank, because the MMSE is that of the null estimator, i.e. $1$. On top of it, here we show in Appendix \ref{appendix:alpha-to-0-RIE-MSE} that the MSE of the RIE when $\alpha\to 0$ is that of naive PCA, which is worse than the Bayes-optimal performance for non rotational invariant priors beyond the transition. Furthermore, still for $\alpha\to 0$, decimation's performance reduces to the one found via the low-rank formula for the MI \cite{camilli2024decimation}, and is thus likely optimal in said limit. All this evidence supports that universality cannot hold for all $\lambda$ when $\alpha\to 0$.

{\subsection{Organization and main contributions} 
We now outline the organisation of the remaining, while presenting our main results in order of appearance.

In Section~\ref{sec:Motivation} we present the observations obtained from numerical experiments based on parallel tempering \citep{Marinari92SimulatedTempering, Hukushima96REMC} and annealed importance sampling \citep{neal2001annealed} used to compute the MI and MMSE. 

$(i)$ Our numerical evidence in Section~\ref{sec:partial_universality}, see \figurename~\ref{fig:ExactMIfiniteSize}, \ref{fig:MCMCvsRIE} and \ref{fig:MCMC_NonSym}, suggests the existence of a phase transition occurring for structured enough priors, coined \emph{denoising-factorisation transition}, separating two regions of the diagram of different nature: a low SNR region where prior universality holds, and its complement where it does not.

$(ii)$ These numerics also support the fact that the rotational invariant estimator discussed earlier is optimal for SNRs smaller than the critical value $\lambda_c(\alpha)$ (later identified theoretically by \eqref{transTheory}), and sub-optimal for larger values. They also suggest in \figurename~\ref{fig:multiscalevsMCMC} that despite outperforming the RIE is information theoretically possible beyond $\lambda_c(\alpha)$, this is algorithmically hard for Monte Carlo without a sufficiently informative initialisation. 

$(iii)$ Importantly, the sub-optimality of the RIE and the breaking of prior universality for discrete priors is proven in Section~\ref{sec:breakingUniv}, see \eqref{lambdac}, thus discarding the strong version of the universality conjecture in \cite{semerjian2024matrix} in a portion of the diagram. Yet, as previously mentioned, our numerics suggest that universality does hold before the transition, which we exploit for the analytical solution in that region, see result $(iv)$.

All these findings are summarized by a qualitative phase diagram in $(\alpha,\lambda)$ plane, see \figurename~\ref{fig:phase_diag} which generalises the low-rank picture of \figurename~\ref{fig:lowRank}.

Section~\ref{sec:multiscale_cavity} develops a \emph{multiscale mean-field theory} able to match the numerical observations at high enough SNR. It relies on a combination of methods from spin glasses. 
Our approach yields equations characterizing the MMSE through \eqref{eq:MMSE_Matrix_facto} and the MI \eqref{eq:MI_Matrix_facto}, also obtainable by an earlier replica calculation by Sakata and Kabashima \cite{SK13EPL} which was deemed incorrect in the whole phase diagram $(\alpha,\lambda)$ for reasons discussed in Section~\ref{sec:related_works}. With our independent multiscale mean-field approach, we redeem their results \emph{but only in a restricted region of the diagram}.

$(iv)$ In Section~\ref{sec:completeansatz_A}, we identify this sub-region of the phase diagram where these two, a-posteriori equivalent, mean-field approaches are conjecturally exact: the \emph{factorisation} or \emph{non-universal phase} after the transition \eqref{transTheory}. Using our numerical findings $(i,iii)$ and the work \cite{semerjian2024matrix} while assuming the existence of a single transition point, we supersede the mean-field equations by the known exact ones for Gaussian prior in the complement of the phase diagram based on the universality of the solution in $P_X$: this is the \emph{denoising} or \emph{universal phase}. The resulting complete ansatz therefore takes into account these two scenarios (effective presence or lack of prior universality), and provides expressions for the MI between data and signal and for the MMSE, for all $(\alpha,\lambda)$: it corresponds to formula \eqref{conjecture} paired with \eqref{eq:MI_Matrix_facto} for the MI after the identified transition \eqref{transTheory}, and \eqref{eq:MMSE_Matrix_facto} for the MMSE. Instead, for lower SNR than the transition, the HCIZ integral \eqref{MIspherical} is used for the MI and \eqref{MMSEspherical} for the MMSE. This complete ansatz together with $(iii)$ are the main theoretical results, whose consequences are discussed in Section~\ref{sec:sp_equations}.

Our complete ansatz \eqref{conjecture} also brings forward the major conceptual result of this work: in extensive-rank matrix denoising, the model behaves in the two identified phases in fundamentally different ways, which require treatments based on very different techniques. The universal phase is akin to a matrix model, where universality properties hold and the HCIZ integral is the proper tool. In the factorisation phase, strong concentration-of-measure effects take place and the model is instead akin to a planted spin glass treatable by mean-field methods. This mixed behavior in a single model is an interesting feature, in contrast with the low-rank regime which remains always mean-field in nature.

Section~\ref{sec:findings_overlap} elaborates around the possibility to carry out factorisation, namely finding $\bX$ alone (rather than the covariance $\bX\bX^\intercal$), up to unresolvable symmetries. 

$(v)$ There we propose Algo.~\ref{alg:sync}, which is empirically shown to estimate $\bX$ in the factorisation phase (thus its name) when given posterior samples obtained by Monte Carlo. We also argue that, in contrast, its estimation is asymptotically impossible in the denoising phase where only non-trivial denoising of $\bX\bX^\intercal$ is statistically possible.

Finally, the reader can find appendices after the Conclusions in the following order. In Appendix \ref{appendix:diagonals} we discuss the irrelevance (from an information theoretical perspective) of the diagonal of $\bY$. Appendix \ref{app:MCMC} gives the details about our numerical experiments carried out using Monte Carlo sampling. In Appendix \ref{app:multiscale_MFT} we report the derivations of our multiscale mean-field theory of Section \ref{sec:multiscale_cavity}. In Appendix \ref{app:MI_reconstruction_potential} we discuss the main properties of the multiscale mean-field ansatz for the MI. In Appendix \ref{appendix:alpha-to-0-RIE-MSE} we compute the $\alpha\to0$ limit of the estimation performance of the RIE. Appendix \ref{appendix:Kabashima_derivation} contains the reduction of the approach of \cite{SK13EPL} to our setting. Appendix~\ref{sec:improv} explores potentially richer ansatzes for the MI, which end up yielding the same solution, thus supporting the validity of the proposed mean-field solution in the factorisation phase.}


\vspace{5pt}
\noindent \textbf{Code repository.} The Monte Carlo simulations used in this paper are accessible at the following link \cite{github_code}.

\section{Exploring the phase diagram}\label{sec:Motivation}

\begin{figure*}[t!!]
    {
        \includegraphics[width=0.95\linewidth]{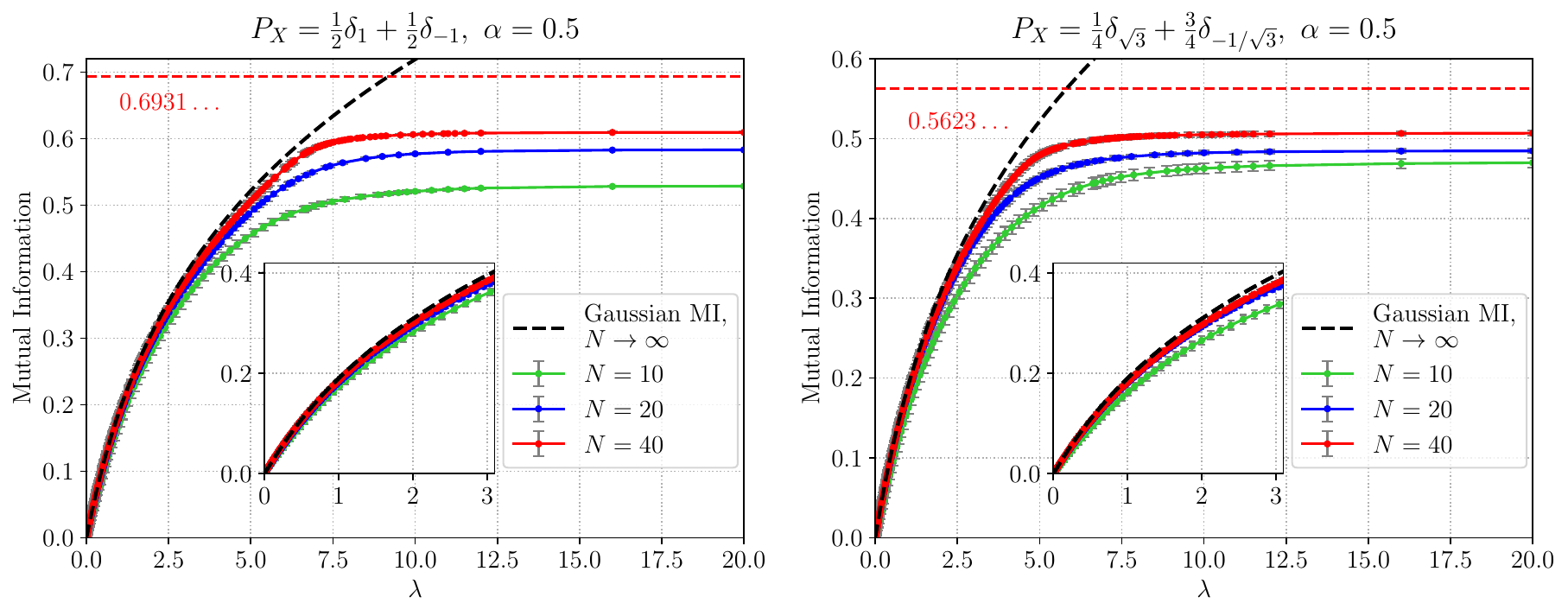}
    }
    \caption{Monte Carlo results for the mutual information with Rademacher prior (left panel) and discrete prior $P_X=\frac14 \delta_{\sqrt{3}}+\frac34 \delta_{-1/\sqrt{3}}$ (right panel) for various sizes with $\alpha=0.5$, i.e., $M=N/2$. Error bars represent the standard error of the mean. The curves are compared to the exact infinite size limit \eqref{MIspherical} in the case of standard Gaussian prior computed using the HCIZ integral. The red dashed curve correspond to Shannon entropy of the prior considered in each panel, which bounds the MI for all sizes from above.}
\label{fig:ExactMIfiniteSize}
\end{figure*}

We start by illustrating some important findings about the nature of extensive-rank matrix denoising that can be probed using numerics and simple rigorous arguments. To do so we conducted replica-exchange Monte Carlo simulations, the details of which can be found in Appendix~\ref{app:MCMC}. The left panel of \figurename~\ref{fig:ExactMIfiniteSize} are the results concerning the finite size value of the MI with $\alpha=0.5$ 
for the Rademacher prior $P_X=\frac12 \delta_1+\frac12 \delta_{-1}$, which will be our main running example of discrete prior with same mean and variance as the standard Gaussian prior throughout this paper. {However, most of the conclusions we draw are not specific to the Rademacher prior; 
in fact, we also present the simulation results for the non-symmetric distribution
$P_X=\frac14 \delta_{\sqrt{3}}+\frac34 \delta_{-1/\sqrt{3}}$ (again having the same mean and variance) in the right panel of \figurename~\ref{fig:ExactMIfiniteSize}, and demonstrate that quantitatively similar conclusions can be drawn.}
Because in the Gaussian case we know the exact asymptotic solution, we use as comparison the curve computed using the HCIZ integral corresponding to the $N,M\to\infty$ limit of the MI, see \eqref{MMSEspherical}, \eqref{MIspherical} and text below.

According to the I-MMSE relation \eqref{eq:I-MMSE}, the derivative of the MI w.r.t. to the SNR gives access to the MMSE. We plot it on \figurename~\ref{fig:MCMCvsRIE} (where it was computed by Monte Carlo, rather than by numerical differentiation of the MI) for the Rademacher prior. In addition, we plot the performance of the RIE for the same sizes. We recall that the RIE denoises the data $\bY$ through eigenvalues shrinkage (keeping its eigenvectors untouched) in order to estimate $\bX\bX^\intercal$, {while exploiting its factorised structure and prior only through its asymptotic spectral law}. This procedure is provably Bayes-optimal when the signal is rotationally invariant, which corresponds to Gaussian prior $P_X=\mathcal{N}(0,1)$ in the present setting, as the MSE of the RIE matches the MMSE \eqref{MMSEspherical}, see \cite{BUN20171,Jean_Farzad_matrixinference_2024}.

\begin{figure*}[t!]
        
    \includegraphics[width=\linewidth]{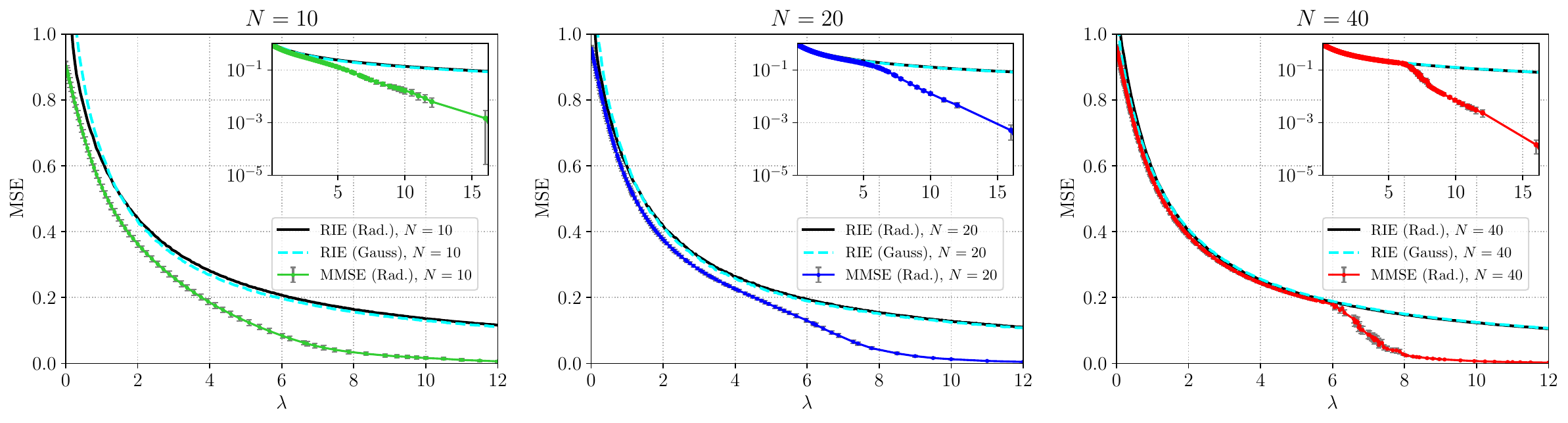}
    \caption{Monte Carlo results for the MMSE for $N=10$ (leftmost figure), $20$ (middle) and $40$ (rightmost) with Rademacher prior (solid colored lines) and $\alpha=0.5$. They are compared to the MSE of the RIE for both Rademacher and Gaussian priors with corresponding $N$. Error bars for the RIE are omitted since they are too small to be visible. Error bars for the MMSE represent standard errors of the mean. Insets are in log-scale.}
\label{fig:MCMCvsRIE}
\end{figure*}

\begin{figure}
    \centering
    \includegraphics[width=0.95\linewidth]{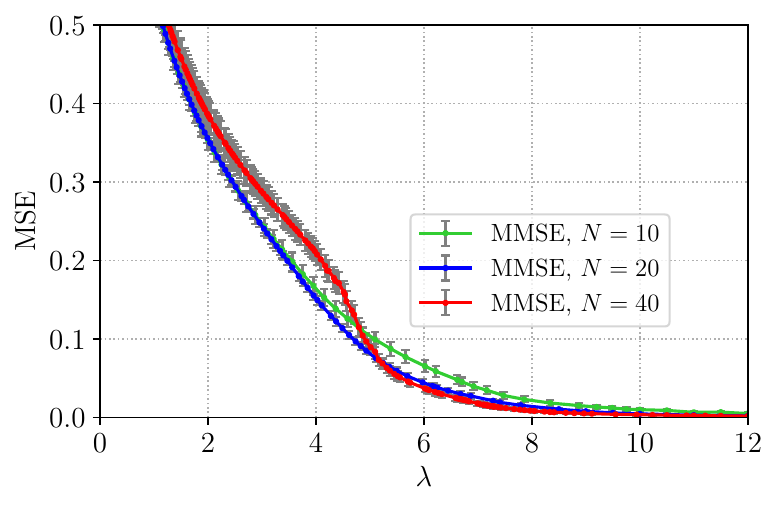}
    \caption{Monte Carlo results for the MMSE for $N=10$, $20$ and $40$ with prior $P_X=\frac14 \delta_{\sqrt{3}}+\frac34 \delta_{-1/\sqrt{3}}$ and $\alpha=0.5$. Error bars represent standard errors of the mean.}
    \label{fig:MCMC_NonSym}
\end{figure}

\subsection{Partial universality, a phase transition and \texorpdfstring{\\}{}a statistical-to-computational gap}\label{sec:partial_universality}

As mentioned before, the recent strong conjecture from Semerjian \cite{semerjian2024matrix} states that for any prior measure $P_X$ for the i.i.d. entries of $\bX$ which is centered with unit variance, the MI and MMSE are universal in the sense that for a large system, their values match for all $\lambda\ge 0$ and $\alpha>0$ the ones in the Gaussian case $P_X=\mathcal{N}(0,1)$ (up to vanishing, non-universal corrections). Moreover, this MMSE performance should be reachable efficiently using the RIE. That this conjecture is wrong \emph{in general} is clear from our finite size numerical experiments. Indeed, a striking observation is that beyond a certain SNR threshold, the MI curves for the Gaussian and discrete priors display an evident discrepancy in the behavior of their respective slopes, a point beyond which the MI for the discrete priors saturate quickly. This effect becoming sharper as the system size grows suggests a first order discontinuous phase transition in the $\lambda$-derivative of the MI, and thus the MMSE, which suddenly decreases beyond it in \figurename~\ref{fig:MCMCvsRIE}. 
{The same effects suggesting a first-order transition are also observed for the non-symmetric discrete prior, for which the MMSE is provided in \figurename~\ref{fig:MCMC_NonSym} where the sharpening of the transition with increased system size becomes even more apparent}. 
Based on the current numerics and later on our theory (see Section~\ref{sec:sp_equations}), we believe that in the thermodynamic limit this \emph{denoising-factorisation transition} $\lambda_c(\alpha)$ is first order for $\alpha>0$; instead for $\alpha\to0$ which recovers the low-rank setting, the transition can be second order (continuous in the MMSE), like for the Rademacher prior.

\begin{figure}[t!!]
    \includegraphics[width=1\linewidth]{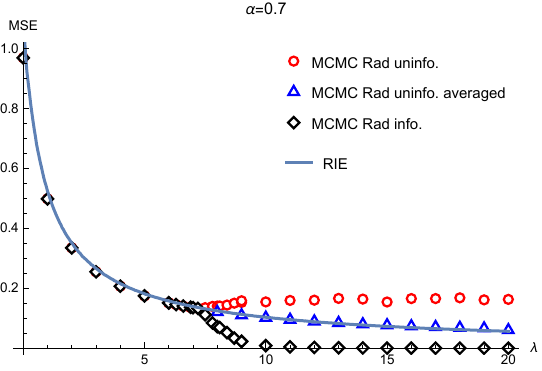} 
    \caption{MSE of the Metropolis-Hastings algorithm for estimating $\bX\bX^\intercal$ with $(N,M)=(40,28)$, so $\alpha=0.7$, with each 
    estimator (red {\color{red}$\circ$}, blue {\tiny\color{blue}$\triangle$} and black {\color{black}$\diamond$} markers) designated in the main text under ``Experimental setup''. 
    Before the transition all markers overlap. These are compared to the RIE performance in the large system limit (blue line). Experimental markers are averaged over $36$ i.i.d. instances of the problem $(\bX,\bZ)$.}    \label{fig:multiscalevsMCMC}
\end{figure}

Another important observation is the good agreement before the transition between the asymptotic MI for Gaussian prior and the finite size MI for both discrete priors, which improves with increasing $N$ and $M$. 
The numerics thus strongly suggest that despite the universality is not holding for all SNRs, it does before the transition in the thermodynamic limit $N\to +\infty$. Also at the level of the MSE we observe an almost perfect matching between the MMSE for Rademacher prior computed by Monte Carlo and the performance of the RIE already at $N=40$, the latter being evaluated for both Rademacher and Gaussian priors (the same holds for the non-symmetric prior of \figurename~\ref{fig:MCMC_NonSym}). As it should, the RIE performance is independent of the prior details as long as mean and variance match. But while the RIE is Bayes-optimal for all SNR in the Gaussian case, beyond $\lambda_c$ 
it becomes sub-optimal for the tested structured priors, as expected if universality breaks down given that it cannot exploit the discrete nature of the priors. 

A natural question is then whether in the region $\lambda>\lambda_c$ there is a statistical-to-computational gap that not only prevents the RIE to reach the MMSE performance, but also a larger class of efficient algorithms including Bayesian ones like Monte Carlo. {
To address this, we also run standard Monte Carlo Metropolis-Hastings simulations (without replica exchange) with different initialization schemes and estimators summarized here:
\newpage
\begin{itemize}[leftmargin=*]
\item[] \underline{Experimental setup in \figurename~\ref{fig:multiscalevsMCMC}:}
\vspace*{-3pt}
\begin{itemize}[leftmargin=*]
    \item[\color{red}$\circ$] A single Monte Carlo simulation with a random initialization, in which the sample average over the Markov chain is given as an estimator. 
    \item[{\tiny\color{blue}$\triangle$}] Results from 32 Monte Carlo simulations, each given a different random initialization. The sample average over all 32 Markov chains is the estimator. 
    \item[\color{black}$\diamond$] A single Monte Carlo simulation initialized from the ground truth. 
\end{itemize}
\end{itemize}}
Further details are given in Appendix \ref{app:MCMC}. 
The numerics indicate that only with a strongly informative initialisation about the signal ({\color{black}$\diamond$}), outperforming the RIE becomes possible. 
Moreover, we observe that, beyond the transition where the ``informative factorisation state'' appears ($\lambda$ larger than $\sim 7.5$), the performance of an estimator constructed from a single Monte Carlo chain with random initialization ({\color{red}$\circ$}) is worse than the RIE (the name ``factorisation state'' is linked to the possibility of solving the factorisation task using samples associated with this thermodynamic state, see Section~\ref{sec:findings_overlap}). The poor performance of Monte Carlo is probably due to a dynamical glass phase transition \cite{antenucci2019glassy} where exponentially many metastable states appear \cite{mezard2009information}, preventing it to properly sample the equilibrium corresponding to the factorisation state beyond $\approx 7.5$, recall \figurename~\ref{fig:MCMCvsRIE}. The metastable states are, however, correlated due to the data. Therefore, an additional average over random initialisations (and thus over associated metastable states) should somehow \enquote{clean} the effect of glassiness. 
This hypothesis can be tested by running several Monte Carlo chains, each with an independent random initialization, and 
adopting the sample average over all of these chains as an estimator ({\tiny\color{blue}$\boldsymbol{\triangle}$}). The performance reached in this way matches precisely that of the RIE, so the aforementioned cleaning takes place, suggesting that the dynamical glass corresponds to a shattering of the RIE state. Yet, the equilibrium factorisation state is not reached, which would require instead averaging over an exponentially (in $N$) large number of Monte Carlo samples (or a smarter sampling procedure, like the replica exchange Monte Carlo algorithm used for \figurename~\ref{fig:MCMCvsRIE}, which is \emph{not} an inference algorithm). We conclude that the dynamical glass phase can be cleaned in polynomial time but the final outcome cannot outperform the RIE, which supports the existence of a hard phase. Algorithmic hardness was hypothesized in \cite{semerjian2024matrix} as a possible consequence of the validity of the weak universality conjecture but not of the strong one.

To summarize, the picture that emerges is that of universality of the information theoretic quantities effectively holding at low SNR, which suggests calling this region the \emph{universal phase}. As proven in the next sub-section, universality cannot hold for a discrete prior beyond a certain SNR. We claim that prior universality breaks down at larger SNR through a phase transition marking the starting point of the \emph{non-universal phase}, which is also the \emph{hard phase} for Monte Carlo algorithms.

\begin{figure}[t]
\includegraphics[width=\linewidth,trim={0.8cm 0cm 0cm 0cm},clip]{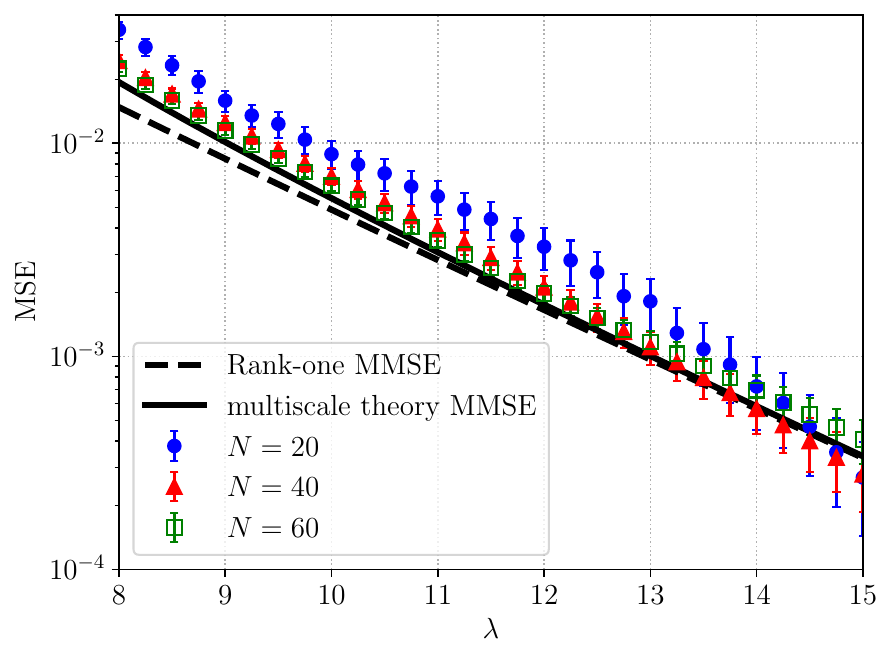}    
    \caption{Monte Carlo results for the MMSE for $N \in\{20, 40,60\}$ with Rademacher prior, with $\alpha=0.5$ and SNR $\lambda$ in the factorisation phase. We also plot the rank-1 MMSE and the MMSE obtained from our multiscale theory (see Section \ref{sec:multiscale_cavity}) for comparison. Error bars represent standard errors of the mean.
    }
    \label{fig:MMSE_log}
\end{figure}

The hard non-universal phase can also be alternatively called the \emph{factorisation phase}. The name comes from the fact that in order to denoise $\bX\bX^\intercal$ optimally in this phase, exploiting all the prior information about the factorised structure of the signal and $P_X$, beyond simply $\rho_{XX^\intercal}$ as the RIE does, becomes necessary and, importantly, factorisation is possible only there, see Section~\ref{sec:findings_overlap}. The universal phase is instead algorithmically easy by optimality of the RIE for denoising there. We thus also call this region the \emph{denoising phase}. 

In the factorisation phase there is a close match between the MMSE and the one in the rank-1 setting where factorisation occurs, see \figurename~\ref{fig:MMSE_log}.  Quantitatively, this translates into an exponential decay of the MMSE as a function of the SNR (similarly to what happens in dictionary learning, see \figurename~5 of \cite{Marc-Kabashima}), whose exponent is the same as in rank-1 matrix estimation, see \figurename~\ref{fig:MMSE_log}. We also stress the even better agreement of our theory presented in Section~\ref{sec:completeansatz_A} with the $N=60$ points in this figure. In particular, despite the gap between the green points and the rank-1 curve is very small, our theory can explain it. In contrast, in the denoising phase the MMSE decay in $\lambda$ is polynomial, see insets of \figurename~\ref{fig:MCMCvsRIE}. This change of polynomial to exponential decay of the MMSE will be captured by the theory, as well as the fact that the MMSE is close to the rank-1 curve in the factorisation phase, see Section~\ref{sec:sp_equations}.


\subsection{Breaking of universality beyond the transition} \label{sec:breakingUniv}

With all these observations in mind we thus agree with the strong universality conjecture of \cite{semerjian2024matrix} when the SNR is sufficiently low, but we can rigorously discard it beyond a threshold using a simple argument based on the definition of the MI. Consider a discrete prior $P_X$ with same mean and variance as $\mathcal{N}(0,1)$. Since $\bX$ is discrete, the MI can always be written as (see \cite[Theorem 3.4]{polyanskiy2025information})
\begin{align}
\frac{I(\bX\bX^\intercal;\bY)}{MN}=\frac{\mathcal{H}(\bX\bX^\intercal)}{MN}-\frac{\mathcal{H}(\bX\bX^\intercal\mid\bY)}{MN}\le \frac{\mathcal{H}(\bX\bX^\intercal)}{MN},\nonumber
\end{align}
where the upper bound follows from the fact that for discrete priors both entropies are non-negative. That is it: the MI cannot grow indefinitely as a function of the SNR for a discrete prior (the upper bound remains finite for any size), while it does so for the Gaussian prior (black dashed curve in \figurename~\ref{fig:ExactMIfiniteSize}). This discards the possibility of universality at least beyond the point where the MI density for the Gaussian case crosses the entropy density associated with the matrix $\bX\bX^\intercal$. 
We are going to show below that $|\mathcal{H}(\bX\bX^\intercal)-\mathcal{H}(\bX)|=o(MN)$. Therefore, we have a bound on the denoising-factorisation transition. Let $X\sim P_X$, a centered discrete prior of unit variance. We have the following upper bound on the domain of possible prior universality:
\begin{align}
    \lambda_c(\alpha)\le\inf\big\{\lambda\ge 0 :\iota^{\rm G}(\alpha,\lambda)\ge \mathcal{H}(X)\big\}.\label{lambdac}
\end{align}
Because $\iota^{\rm G}(\alpha,\lambda)$ is known using the HCIZ integral via \eqref{MMSEspherical}, \eqref{MIspherical} the bound is explicit (we have used that the prior is factorised over the signal matrix entries to get $\mathcal{H}(\bX)/{MN}=\mathcal{H}(X)$ but a similar bound can be written even if there is no such factorisation).
In the case of Rademacher prior its entropy equals $\ln 2$. We indeed see in the left panel of \figurename~\ref{fig:ExactMIfiniteSize} that the bound holds (also for the other prior on right panel) but it does not look sharp.

Interestingly, we can also explain the apparent gap between the value towards which the finite size MI saturates and the entropy of the prior. In fact, the observations $\bY$ are informed about the ground truth $\bX$ only through the combination $\bX\bX^\intercal=\sum_{\mu\le M}\bX_\mu\bX_\mu^\intercal$. Hence for large $\lambda$ we shall have $I(\bX\bX^\intercal;\bY)\approx \mathcal{H}(\bX\bX^\intercal)$ at leading order, which is smaller than $\mathcal{H}(\bX)$ by the data processing inequality. The inequality in this case is strict: $\mathcal{H}(\bX\bX^\intercal)<\mathcal{H}(\bX)$ because $\bX\bX^\intercal$ is a non-invertible function of $\bX$. The degeneracy of the mapping $\bX\mapsto\bX\bX^\intercal$ explains accurately the gap. Specifically, for each $\bX\bX^\intercal$ there can be $M!$ permutations of the columns of $\bX$ that leave it invariant. Furthermore, if the prior is also symmetric like in the Rademacher prior case in the left panel of \figurename~\ref{fig:ExactMIfiniteSize}, we must add possible sign flips of these columns to the invariance group of $\bX\bX^\intercal$. Of these $2^M M!$, some are actually redundant, because columns of $\bX$ can match. The probability that any two columns match with a factorised Rademacher prior, however, is $O(M^22^{-N})$, which is negligible for large enough $N$ with respect to the leading order corrections that we compute. The probability of more columns matching is of even lower order. Hence we can finally say that
\begin{align}
    \mathcal{H}(\bX)=\mathcal{H}(\bX\bX^\intercal)+\ln(2^M M!)-o_N(1) . \label{entropyXX}
\end{align}
Putting these considerations together, we readily get
\begin{align}\label{eq:finite_size_correct_MI}
    \frac{I(\bX\bX^\intercal;\bY)}{MN}\xrightarrow[]{\lambda\to\infty}
    \frac{\mathcal{H}(\bX)}{MN}-\frac{\ln(2^M M!)}{MN}+o_N(1),
\end{align}where the remainder can be estimated to be of order $O( 2^{-N}M\ln M)$. 
It is now manifest that as $M,N\to\infty$, despite the degeneracy induced by signed permutations, the contributions yielded by the second and third terms on the r.h.s.\ vanish. This explains why increasing $N$ and $M$ in \figurename~\ref{fig:ExactMIfiniteSize} is sufficient to reduce the gap with the saturation value $\ln2$. E.g., for $N=20, M=10$ the contribution of the signed permutations is $\approx 0.1102$ which summed to the saturation value of the blue curve $\approx 0.5826$ yields $\approx 0.6928$, which is remarkably close to $\ln2\approx 0.6931$. The same holds for the other gaps, and the argument also explains the discrepancy for the non-symmetric prior in the right panel of \figurename~\ref{fig:ExactMIfiniteSize}, where in this case one does not need to consider sign invariance. Note that we do not claim that beyond the transition the thermodynamic limit of the MI should saturate strictly to the prior entropy. Indeed, \figurename~\ref{fig:MCMCvsRIE} shows that a small gap survives whose size decreases exponentially fast with~$\lambda$.

The r.h.s. of \eqref{lambdac} is a rigorous bound above which universality of the MI and MMSE cannot hold for discrete priors in any possible scenario.
We note that given the extensive set of works on universality in random matrices \cite{anderson2010introduction,potters2020first,bai2010spectral,tao2010random,erdHos2010bulk,erdHos2011universality,10.1214/11-AOP648,10.1214/13-AAP939,ajanki2017universality}, it may come as a surprise that it breaks down. E.g., at the level of the asymptotic spectrum of $\bX\bX^\intercal$ and thus of $\bY$, or even at the level of more local statistics, the previous papers show universality. But the mutual information depends on all the fine statistical properties of the data and $\bX$, which become manifest at high enough SNR.

We also remark that the reasoning leading to the existence of a threshold beyond which universality w.r.t.\ the base measure of the matrix to infer cannot hold can be extended to a much broader set of matrix inference/learning models. In any such model, if the signal $\bX$ is discrete the MI between data and signal cannot exceed the entropy of $\bX$ (as a function of the SNR, number of data, etc). This establishes an upper bound on the region of possible validity of prior universality in models where replacing the signal by a Gaussian equivalent with matching moments yields an unbounded MI. And indeed, this MI saturation has already been shown to have deep implications in neural networks \cite{2025_nn_quadratic_Barbier} .

{ Whether this universality breaking holds for continuous (but non-Gaussian) priors is an unanswered question in this paper. 
We conjecture that, like the case with discrete priors, the non-rotational-invariant nature of the prior should be exploitable to construct a Bayesian estimator with performance well exceeding the RIE once the SNR is high enough. 
However, as the mutual information cannot be bounded from above in such settings, the previous arguments negating universality do not apply. 
Performing numerical experiments under continuous priors is also considerably more difficult due to the inability to exploit the binary nature of $\bx$ in the implementation for fast computations. 
Moreover, the enlargement of the phase space from a discrete one to a continuous one may result in a longer relaxation time for the Monte Carlo chain to converge. 
Confirming whether this conjecture holds or not, either theoretically or numerically, is a challenging problem left open for future investigation. 
}

{
Note that while $N,M$ may seem small, the total number of variables itself is $MN$, which ranges from 50 to 1800 in our experiments. This is comparable with the extensive numerical studies done on the Edwards-Anderson Ising spin glass using parallel tempering \cite{Wang15_EASpinGlassNumerics}, in which the system size ranges from 64 to 1000. 
}

\begin{figure}[t!]
    \centering
    \includegraphics[width=\linewidth]{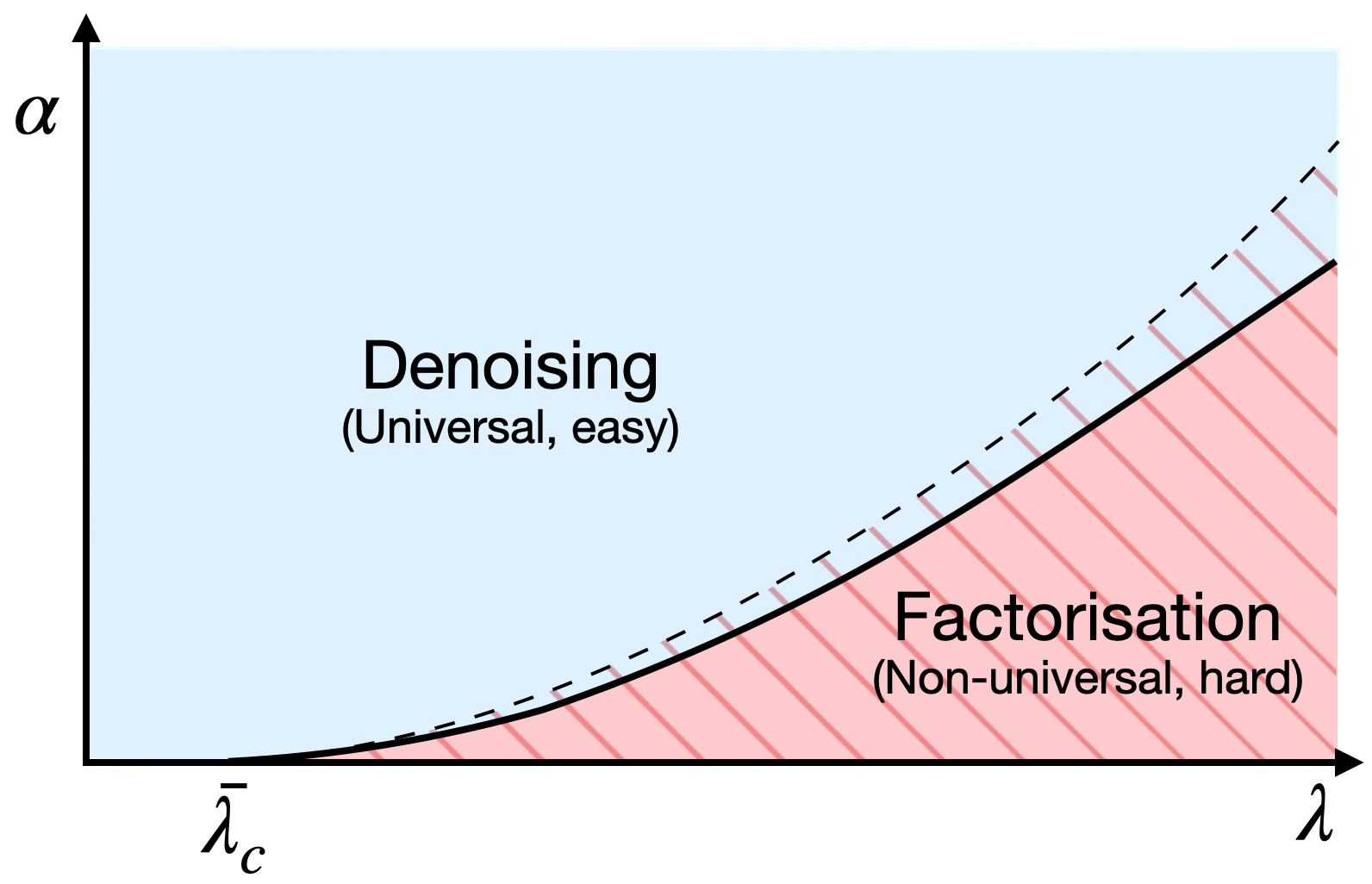}
    \caption{Qualitative phase diagram for matrix denoising with discrete (Rademacher) prior; it generalises to the extensive-rank regime the one of \figurename~\ref{fig:lowRank} for low rank. The thermodynamic phase transition $\lambda_c(\alpha)$ (solid line) is believed to be first order for $\alpha>0$. Increasing $\lambda$, it is preceded by a spinodal transition where the thermodynamic factorisation state appears but is not yet the equilibrium (dashed line). The blue region is the denoising phase where prior universality holds, and optimal denoising is algorithmically easy using the RIE. The shaded region corresponds to the dynamical glass phase (we cannot assess whether it starts precisely at the spinodal transition as displayed). The red region is the factorisation phase where prior universality does not hold. There, it is statistically possible but algorithmically hard to estimate $\bX$ and to denoise $\bX\bX^\intercal$ better than the RIE. The line $\alpha=0$ recovers the results for the rank-1 problem and is thus special: $\bar \lambda_c=\lambda_c(0)$ is second order and there is no hard phase for Rademacher prior. For other priors (possibly discrete), a hard phase can appear in the low-rank case too, see \figurename~\ref{fig:lowRank}.}
    \label{fig:phase_diag}
\end{figure}

\vspace{5pt}
\noindent \textbf{Summary and phase diagram.} Our findings are condensed into a phase diagram in the $(\alpha,\lambda)$ plane, see \figurename~\ref{fig:phase_diag}, which generalises \figurename~\ref{fig:lowRank} to the extensive-rank regime. We have identified two phases for matrix denoising with discrete prior under the hypothesis that the \emph{denoising-factorisation transition} $\lambda_c(\alpha)$ exists and is unique. In the \emph{denoising} or \emph{universal phase} $\lambda\in [0,\lambda_c(\alpha))$ treatable using matrix models techniques and the HCIZ integral, we conjecture that it is statistically not possible to outperform the RIE for denoising $\bX\bX^\intercal$ despite it uses only its asymptotic spectral distribution as prior. It is therefore \emph{algorithmically easy} to reach the MMSE performance (i.e., in polynomial complexity in $N$). Universality in the prior of the MI and MMSE holds, and may be linked to the quasi uniformity of the eigenvectors of $\bX\bX^\intercal$ \cite{f8052592-465a-34f1-b9a1-f4c6eade3b19,erdHos2009local,benaych2011eigenvectors,o2016eigenvectors} and the universality of its spectrum \cite{10.1214/11-AOP648,10.1214/13-AAP939} when $\bX$ is made of i.i.d. entries. We argue in Section~\ref{sec:findings_overlap} that in this region, factorisation, namely estimating $\bX$ up to a unresolvable ambiguities inherent to the mapping $\bX\mapsto\bX\bX^\intercal$, is statistically impossible.

Beyond the transition point we find the \emph{factorisation} or \emph{non-universal phase} $\lambda\in (\lambda_c(\alpha),+\infty)$ where techniques from mean-field spin glasses become correct. In this phase, exploiting the complete prior knowledge about $\bX\bX^\intercal$ beyond its spectrum becomes essential: information theoretically, a Bayesian algorithm, such as Monte Carlo sampling, can outperform the RIE for denoising.  Moreover, factorisation is also possible on top of denoising. Both tasks are however \emph{algorithmically hard} due to the presence of a dynamical glass phase appearing through a spinodal transition.

Finally, we can also interpret the transition as marking the starting point of a \emph{feature learning regime}, where it becomes information theoretically possible to detect and exploit fine features in the data for achieving better denoising and factorisation. In the denoising phase, on the contrary, even knowing in advance the features (factorised and discrete structure) does not yield any improvement compared to assuming no particular structure (i.e., a Gaussian prior).

Notice two major differences compared to the phase diagram \figurename~\ref{fig:lowRank} for $M=1$. Firstly, in the extensive-rank case the diagram is defined on a plane $(\alpha,\lambda)$ rather than parametrized by the SNR $\lambda$ only. The transition is thus now a line $\lambda_c(\alpha)$ and its nature (order) can change along it. And indeed, between $\alpha=0$ or strictly positive, it does change for instance in the Rademacher prior case. Secondly, at fixed $\alpha>0$, no easy regime appears once the transition is crossed, even at very large SNR. Extensive-rank matrix denoising is thus algorithmically much more demanding than in the low-rank case for large $\lambda$.

\section{Multiscale mean-field theory}\label{sec:multiscale_cavity}
We now outline a theory matching the numerical observations discussed earlier. As already mentioned, we cannot employ methods from matrix models such as the HCIZ integral. Instead, we rely on spin glass techniques concatenated in an original way. The proposed approach is called \emph{multiscale mean-field theory}. The reason is that it uses multiple mean-field techniques in a multiscale fashion. {The details of the derivations of this section are deferred to Appendix \ref{app:multiscale_MFT}.}

{To begin with, as proved in App.~\ref{appendix:diagonals}, we observe that the diagonal of the data $(Y_{ii})_{i\le N}$ does not contribute to the MI in the thermodynamic limit. In fact, $Y_{ii}=\sqrt{\lambda/N}\|\bX_i\|^2+Z_{ii} \approx \sqrt{\lambda/N}\|\bX_i\|^2$. But $\|\bX_i\|^2/N$ can be already estimated precisely without knowing $Y_{ii}$: its value is known to a Bayes-optimal statistician, as it concentrates on the second moment of the prior (i.e., $1$) by the law of large numbers. We thus now work with}
\begin{align}\label{eq:channel1}
    Y_{ij}=Y_{ji}=\sqrt{\frac{\lambda}{N}}\bX_i^\intercal\bX_j+Z_{ij},\quad 1\leq i <j\leq N.
\end{align}
The MI for this channel then reads
\begin{align}\label{eq:MI_and_free_entropy}
    \frac{I(\bX\bX^\intercal;\bY)}{MN}&=\frac{\lambda}{4}-\phi_N+O(N^{-1})
    \end{align}
    with \emph{free entropy}
    \begin{align}        
    &\phi_N:=\frac{1}{MN}\EE_\bY\ln  \mathcal{Z}_{N,M}(\bY).  \label{eq:free_entropy} 
\end{align}
Using the language of statistical mechanics, we introduced the partition function of the model
\begin{align*}
    \mathcal{Z}_{N,M}(\bY):=\int_{\mathbb{R}^{N\times M}} dP_X(\bx)e^{-H_N(\bx;\bY)}
\end{align*}
with (minus) Hamiltonian (which is the log-likelihood)
\begin{align}\label{eq:Hamiltonian}
    -H_N(\bx;\bY)\!=\!\sqrt{\frac{\lambda}{N}}\!\sum_{i<j,1}^N\! Y_{ij}\bx_i^\intercal\bx_j-\frac{\lambda}{2N}\!\sum_{i<j,1}^N\!(\bx_i^\intercal\bx_j)^2.
\end{align}The associated Boltzmann-Gibbs average, i.e., the posterior average, will also be denoted by $\langle\,\cdot\,\rangle$ in this section.

The MMSE of the inference problem \eqref{eq:channel1} does not include any error on the reconstruction of the diagonal part of the signal $\bX\bX^\intercal$, and is thus defined (with a little abuse of notation) as
\begin{align}\label{eq:MMSE_Multiscale_starting_point}
    {\rm MMSE}&=\frac{1}{MN^2}\EE\sum_{i\neq j,1}^N\big[(\bX_i^\intercal\bX_j)^2-\langle \bx_i^\intercal\bx_j\rangle^2\big]\nonumber\\
    &=1-\frac{1}{NM}\EE\sum_{i=1}^{N-1}\langle \bx_i^\intercal\bx_N\rangle^2+O(N^{-1}).
\end{align}
One can show that the above is equal to \eqref{MMSEuseful} up to $o_N(1)$, which is why it has been denoted in the same way. This $o_N(1)$ discrepancy is precisely due to the diagonal contributions that would appear in \eqref{MMSEuseful} (see Appendix \ref{appendix:diagonals}). The second line is instead obtained using the exchangeability of the rows, and selecting $j=N$ for convenience.

The free entropy $\phi_N$ can be viewed as that of a non-standard spin system made of $N$ spins (the rows of $\bx$, $(\bx_i)_{i\leq N}$) which are $M$-dimensional (i.e., have extensive dimensionality instead of being scalar or low-dimensional variables) or, instead, $M$ spins (the columns of $\bx$, $(\bx_\mu)_{\mu\leq M}$) which are $N$-dimensional. The two viewpoints are equivalent, {but the first turns out to be more amenable for a cavity computation, which would imply the extraction of a row of $\bx$. The reason is that the part of the Hamiltonian depending on row $\bx_i\in\mathbb{R}^M$ carries a noise variable $\bZ_i\in\mathbb{R}^N$ independent of the rest. 
More precisely, the posterior can be decomposed in terms of rows $\bx_i$ and $\bY_i=(Y_{ij})_{j\le i-1}$, $i=1,\ldots,N$, as
\begin{align}
dP(\bx\mid \bY) &\propto  dP_X(\bx_N) P(\bY_N\mid \bx_N,(\bx_j)_{j\le N-1}) \nonumber\\
&\qquad \times  \prod_{i\le N-1} dP_X(\bx_i) P(\bY_i\mid \bx_i, (\bx_j)_{j\le i-1})\nonumber\\
&\propto dP(\bx_N \mid \bY_N,(\bx_j)_{j\le N-1})\nonumber \\
&\qquad\times dP((\bx_j)_{j\le N-1} \mid (\bY_j)_{j\le N-1})\label{cavVSbulkMeasure}
\end{align}
where we have dropped normalisations. This decomposition following from \eqref{eq:channel1} emphasizes that, conditional on the ``bulk'' $(\bx_j)_{j\le N-1}$, the ``cavity'' $\bx_N$ depends on a single data row $\bY_N$, and the bulk marginal posterior measure is independent of the cavity. Note that the same argument holds when extracting any other cavity row after a proper permutation of row indices. This decomposition is crucial as, firstly, it makes manifest that the posterior of the cavity conditional on the bulk corresponds to a lower dimensional Bayes-optimal inference problem. And secondly, the bulk measure can be simplified without having to modify the cavity measure. Such a decomposition cannot be written if we dig a column cavity. }

\subsection{Computing the limiting MMSE}
The perspective introduced above is particularly useful when combined with the MMSE in \eqref{eq:MMSE_Multiscale_starting_point}. The latter is in fact the so-called YMMSE of a random linear estimation (RLE) problem with noisy design (see \cite{barbier2020mutual,krzakala2013compressed}), where samples $\bar\bx:=(\bx_i)_{i\leq N-1}$ are used for the estimation of the uncertain design matrix $\bar\bX\in\mathbb{R}^{N-1\times M}$, while samples $\bx_N$ from the cavity posterior measure $dP(\bx_N\mid\bY_N,(\bx_j)_{j\leq N-1})$ are used to estimate $\bX_N$.

\vspace{5pt}
\noindent \textbf{First scale reduction.} Mapping onto the RLE task has the advantage of converting a matrix model into a vector inference problem over the cavity (still with an annealed bulk matrix), potentially tractable using known approaches. This has the effect of \enquote{reducing the scale of the problem}: we go from a partition function associated with an Hamiltonian of magnitude $O(MN)$ to one with an $O(M)$ Hamiltonian. This $O(MN)\to O(M)$ reduction performed by means of the cavity method corresponds to the first scale reduction in our scheme. Contrary to standard spin or inference models (like the Hopfield model or the perceptron treated with the cavity method in \cite{Tala_vol1,Tala_vol2}), the RLE problem that will appear remains high-dimensional because $M\to+\infty$ with $N$. 

{Let us now identify the correct RLE problem. We denote the cavity variable $\bx_N$ by $\beeta\in\mathbb{R}^M$, $\bX_N$ by $\bH\in\mathbb{R}^M$ and the bulk (i.e., all but the cavity) variables as $\bar\bx,\bar\bX\in\mathbb{R}^{N-1\times M}$. The RLE under study then corresponds to the two following observation channels:
\begin{align}\label{eq:RLE_channel_cavity}
    &\tilde \bY(\lambda)=\sqrt{\frac{\lambda}{N}}\bar\bX\bH+\tilde\bZ\\
    \label{eq:RLE_channel_bulk}
    &Y_{ij}'(\zeta)=\sqrt{\frac{\zeta}{N}}\bX_i^\intercal\bX_j+Z_{ij},\quad 1\leq i<j\leq N-1.
\end{align}The first one is the linear observation of the cavity ground truth variable $\bH$ through the uncertain design matrix $\bar\bX$, whereas the second is additional information about  $\bar\bX$ only. $\zeta$ is an auxiliary SNR that should be taken equal to $\lambda$, but we consider a more general model where it may differ, in which case we denote the data $\bY'(\zeta)=(Y_{ij}'(\zeta))_{i<j\le N-1}$. Denoting by $\langle\,\cdot\,\rangle_{\circ}$ the posterior average acting on $(\bar\bx,\beeta)$ of this Bayes-optimal inference problem over $(\bar \bX,\bH)$, we define the YMMSE given \eqref{eq:RLE_channel_cavity}, \eqref{eq:RLE_channel_bulk} as
\begin{align*}
    \text{YMMSE}(\lambda,\zeta)&:=\frac{1}{NM}\EE\|\bar\bX\bH-\langle\bar\bx\beeta\rangle_{\circ}\|^2\nonumber\\
    &\ =1-\frac{1}{NM}\EE\|\langle\bar \bx\beeta\rangle_\circ\|^2,
\end{align*}
where the second equality follows from the Nishimori identity Proposition~\ref{prop:Nishi_ID} in App.~\ref{app:multiscale_MFT}. This is the correct RLE problem because for $\zeta=\lambda$ the above is equal to \eqref{eq:MMSE_Multiscale_starting_point} up to vanishing errors (and also $\langle\,\cdot\,\rangle=\langle\,\cdot\,\rangle_\circ$). The YMMSE is obtained by an I-MMSE relation similar to \eqref{eq:I-MMSE} through a $\lambda$-derivative of the following conditional MI:
\begin{align}\label{eq:iotaRLE_def}
    \iota&_N^{\rm RLE}(\lambda,\zeta):=\frac{1}{M}I(({\bar\bX,\bH});\tilde\bY(\lambda)\mid \bY'(\zeta))\\
    &=\frac{1}{M}\mathcal{H}(\tilde \bY(\lambda)\mid \bY'(\zeta))-\frac{1}{M}\mathcal{H}(\tilde \bY(\lambda)\mid \bY'(\zeta),\bar\bX,\bH).\nonumber
\end{align}
Specifically, we have
\begin{align}\label{eq:I-YMMSE_RLE}
    \frac{d}{d\lambda}\iota_N^{\rm RLE}(\lambda,\zeta)=\frac{1}{2}{\rm YMMSE}(\lambda,\zeta).
\end{align}
The reason we introduced $\zeta$ in the first place is to be able to write down this relation. Taking $\zeta=\lambda$ from the beginning would prevent accessing ${\rm MMSE}\approx {\rm YMMSE}(\lambda,\zeta=\lambda)$ through an I-MMSE relation with SNR $\lambda$ influencing a single (cavity) row. This is thus needed to relate the calculation of ${\rm MMSE}$ \eqref{eq:MMSE_Multiscale_starting_point} to the RLE problem. The second entropy in \eqref{eq:iotaRLE_def} is just the one of the noise $\tilde \bZ$ and is thus simple. The first is instead related to the free entropy
\begin{align}
    &\phi^{\rm RLE}_N(\lambda,\zeta)=\phi^{\rm RLE}_N\label{eq:RLE_entropy}\\
    &\ :=\frac{1}{M}\EE_{\tilde \bY(\lambda),\bY'(\zeta)}\ln \int_{\mathbb{R}^M} dP_X(\beeta)\Big\langle e^{-H^{\rm row}_{N,M}(\beeta;\bar\bx,\lambda)}\Big\rangle'\nonumber
\end{align}by a constant shift and a sign. Here, $H^{\rm row}_{N,M}$ given by \eqref{eq:Hrow} in App.~\ref{app:first_scale} is the ``cavity Hamiltonian'', namely, the term capturing all the dependence on the extracted row $\beeta$ (and its interaction with the bulk) in the total Hamiltonian \eqref{eq:Hamiltonian}. What remains of \eqref{eq:Hamiltonian} is the ``bulk Hamiltonian'', see \eqref{eq:Hbulk}. The latter defines the posterior average $\langle\,\cdot\,\rangle'$ given $\bY'(\zeta)$ only, thus acting on the bulk $\bar\bx$ but not the cavity $\beeta$. $\phi^{\rm RLE}_N$ clarifies the connection between our mapping to the RLE setting, and the cavity method. It actually mirrors a term appearing in the Aizenman-Sims-Starr representation of the free energy in standard spin glasses and inference \cite{ASS,lelarge2019fundamental}.}

\vspace{5pt}

\noindent \textbf{Second scale reduction.}
The previous reduction by the cavity method is exact. But the complexity due the inherent matrix nature of the problem has not disappeared: it is now segregated into the uncertainty of the design $\bar\bx=(\bx_i)_{i\leq N-1}$ measured by $\langle\,\cdot\,\rangle'$. In order to integrate this bulk we will resort to a mean-field ansatz, yielding a model that is tractable by means of rigorous replica formulas \cite{BarbierGLM-PNAS}, in which one can identify scalar inference problems with $O(1)$ Hamiltonians. This $O(M)\to O(1)$ magnitude reduction of the Hamiltonian corresponds to the second scale reduction in our approach.

Before doing this let us make a remark. We could have tried to simplify the problem using an ansatz direclty at the level of the original Hamiltonian/problem in order to go from $O(MN)$ degrees of freedom to $O(1)$ effective ones in one step, e.g. using the replica method as in \cite{SK13EPL,Marc-Kabashima,barbier2022DL,perturbative_Maillard21}. But having already reduced the scale down to $O(M)$ by the cavity method \emph{in an exact manner}, our assumptions will only be at the level of the bulk (and thus potentially more interpretable): the ``direct'' interaction between cavity and bulk captured by channel \eqref{eq:RLE_channel_cavity}, which defines the first posterior in \eqref{cavVSbulkMeasure}, is kept untouched, with approximations only made through the simplification of $\langle\,\cdot\,\rangle'$, i.e., the second posterior law in \eqref{cavVSbulkMeasure}. In essence, our mean-field ansatz assumes that the interactions among bulk variables have a lower order effect on the free entropy $\phi_N^{\rm RLE}$ obtained by the cavity method compared to the direct cavity-bulk interaction. 

Besides being physically meaningful, these assumptions allow us to integrate the bulk variables. We achieve this by replacing channel \eqref{eq:RLE_channel_bulk} by an effective, more manageable, one:
\begin{align}\label{eq:aux_channel_bulk}
    \bY'_{\rm eff}(\sigma)=\sqrt{\sigma}\,\bar\bX+\bZ_{\rm eff}',
\end{align}where $\bZ_{\rm eff}'=(Z'_{i\mu})_{i\leq N-1,\mu\leq M}$ is a matrix of i.i.d. standard Gaussian variables, $\bar\bX$ is the ground truth bulk as before, $\sigma$ is a tunable SNR for the effective observations $\bY'_{\rm eff}(\sigma)=(\bY'_i\in\mathbb{R}^{M})_{i\le N-1}=(Y'_{i\mu})_{i\leq N-1,\mu\leq M}$. This is the anticipated mean-field ansatz, i.e., a parametrization of the bulk measure in terms of marginals only. Its potential validity in this extensive-rank setting will be thoroughly discussed later. For now, let us just notice that this form of marginal is asymptotically exact for standard low-rank models \cite{XXT,lelarge2019fundamental,BarbierGLM-PNAS}, provided $\sigma$ is suitably tuned.

Within this effective factorised mean-field ansatz the simplified bulk measure reads
\begin{align*}
\langle\,\cdot\,\rangle_{\rm eff}'=\prod_{i,\mu=1}^{N-1,M}\frac{1}{\mathcal{Z}'(Y'_{i\mu})}\int dP_X(x_{i\mu})e^{\sqrt{\sigma} Y'_{i\mu}x_{i\mu}-\frac{1}{2}\sigma x_{i\mu}^2}(\,\cdot\,) 
\end{align*} 
with proper normalisation $\mathcal{Z}'$. We can now leverage on known results \cite{BarbierGLM-PNAS} in order to derive a variational expression for the limit of $\phi^{\rm RLE}_N$ after trading \eqref{eq:RLE_channel_bulk} for \eqref{eq:aux_channel_bulk} (see Appendix~\ref{app:multiscale_MFT}):
\begin{align}\label{eq:RLE_RS_free_entropy}
    &\phi^{\rm RLE}_{\rm eff}(\lambda,\sigma)={\rm extr}_+\Big\{\EE_{\xi,H}\ln\int \!dP_X(\eta)e^{(\xi\sqrt{r}+Hr)\eta-\frac{1}{2}r\eta^2}\nonumber\\
    &\qquad\quad-\frac{rq}{2}+\frac{\lambda}{2}-\frac{1}{2\alpha}\ln \big(1+\lambda\alpha (1-J(\sigma)q)\big)
    \Big\},
\end{align}where $\xi\sim\mathcal{N}(0,1),H\sim P_X$, and ${\rm extr}_+$ denotes the extremization operation w.r.t.\ $(r,q)$. The $+$ subscript indicates that, in case of multiple extremal points, the one that realizes the largest value of the potential should be considered. Finally, we called this free entropy $\phi^{\rm RLE}_{\rm eff}$ and not $\phi^{\rm RLE}$ to distinguish our ansatz based on \eqref{eq:aux_channel_bulk} from the true asymptotic free entropy. Through our ansatz, $\zeta$ has disappeared in favor of a properly chosen parameter $\sigma$ (we will later see how to chose it), that appears in a function $J(\sigma)$ defined as 
\begin{align}\label{eq:def_J}
    J(\sigma)&:=\EE_{X,Z'} \Big[X\frac{\int dP_X(x) \,x\,e^{(\sqrt{\sigma}Z'+X\sigma)x-\frac{1}{2}\sigma x^2}}{\int dP_X(x)e^{(\sqrt{\sigma}Z'+X\sigma)x-\frac{1}{2}\sigma x^2}}\Big]
\end{align}where $X\sim P_X$ and $Z'\sim \mathcal{N}(0,1)$. This is the overlap of a Gaussian channel of SNR~$\sigma$.
Based on definition \eqref{eq:iotaRLE_def}, we can relate the free entropy $\phi^{\rm RLE}_{\rm eff}$ for RLE with uncertain design with effective observations \eqref{eq:aux_channel_bulk} to the corresponding approximation $\iota^{\rm RLE}_{\rm eff}(\lambda,\sigma)$ of $\lim_{N\to\infty}\iota_N^{\rm RLE}(\lambda,\zeta)$:
\begin{align}\label{eq:MIRLE_vs_psiRLE}
    \iota_{\rm eff}^{\rm RLE}(\lambda,\sigma)=\frac{\lambda}{2}-\phi^{\rm RLE}_{\rm eff}(\lambda,\sigma).
\end{align}
Note that, under the effective bulk observations \eqref{eq:aux_channel_bulk}, this equality is essentially rigorous. This is because in this case the problem can be mapped onto a standard RLE problem with perfect knowledge of the design matrix but with higher measurement noise (as we exploit in Appendix~\ref{app:second_scale}), which is covered by the theorems in \cite{BarbierGLM-PNAS}.

Before moving to the saddle point equations we observe that in the Bayes-optimal setting the extremizer $q_*=q_*(\lambda,\sigma)$ in \eqref{eq:RLE_RS_free_entropy} corresponds to the asymptotic overlap between the ground truth row $\bH$ and a sample $\beeta$ from the posterior associated to the RLE problem with noisy covariates \eqref{eq:RLE_channel_cavity}, \eqref{eq:aux_channel_bulk}, or between two conditionally independent posterior samples (by the Nishimori identity, see \ref{general_NishiId}), which concentrate onto their mean given by
\begin{align*}
    q_*&=\lim_{N\to\infty}\EE \frac1 M \bH \cdot \EE[ \beeta\mid \tilde \bY(\lambda),\bY'_{\rm eff}(\sigma)]\\
    &=\lim_{N\to\infty}\EE \frac1 M  \| \EE[ \beeta\mid \tilde \bY(\lambda),\bY'_{\rm eff}(\sigma)]\|^2, 
\end{align*}
where $\EE[ \beeta\mid \tilde \bY,\bY'_{\rm eff}]$ equals the posterior mean of the cavity given direct observations $\tilde \bY$ about it, and the effective observations for the bulk. Instead, the extremizer $r_*=r_*(\lambda,\sigma)$ is the effective SNR which controls the cavity fields underlying the marginals of the cavity coordinates under model \eqref{eq:RLE_channel_cavity}, \eqref{eq:aux_channel_bulk}: asymptotically, the unormalised random marginal probability measure of any of the $\eta_i$ reads
\begin{align}
dP(\eta_i\in[\eta,\eta+d\eta))\!\propto \!dP_X(\eta) e^{(\xi_i\sqrt{r_*}+H_ir_*)\eta-\frac{1}{2}r_*\eta^2}\label{cavity marginal}
\end{align}
with randomness $\xi_i\sim\mathcal{N}(0,1)$ and $H_i\sim P_X$. 
In other words, in the thermodynamic limit a cavity coordinate has marginal corresponding to a random posterior distribution $\mathbb{P}(H_i\in\,\cdot \mid \sqrt{r_*}H_i+\xi_i)$ associated with the effective scalar observation $\sqrt{r_*}H_i+\xi_i$ with prior $P_X$.

\subsection{Saddle point equations and mutual information}
The saddle point equations are obtained by equating the gradient of $\{\cdots\}$ in \eqref{eq:RLE_RS_free_entropy} to zero:
\begin{align}\label{eq:RLE_saddlepoint_r}
 &q=\EE_{\xi,H} [H\langle \eta \rangle_{r}],\qquad r=\frac{\lambda J(\sigma)}{1+\lambda\alpha(1-J(\sigma)q)},
\end{align}
where
\begin{align*}
    \langle\,\cdot \,\rangle_{r}=\langle\,\cdot \,\rangle_{r}(\xi,H):=\frac{1}{\mathcal{Z}_{r}}\int dP_X(\eta)e^{(\xi\sqrt{r}+ Hr)\eta-\frac{1 }{2}r\eta^2}(\,\cdot\,) ,
\end{align*}and $\mathcal{Z}_{r}(\xi,H)$ is the normalisation. The solutions $(r_*,q_*)$ depend on $\alpha,\lambda$, and $J(\sigma)$.

The main question now is: \emph{how to set the effective bulk SNR $\sigma$ in order to capture the effect of the original bulk observations \eqref{eq:RLE_channel_bulk} when $\zeta=\lambda$?} We argue that it can be determined with a consistency argument. $\sigma$ parametrizes the marginals of the bulk variables $\bar \bx=(x_{i\mu})_{i\le N-1,\mu\le M}$ through \eqref{eq:aux_channel_bulk}. In the original problem, or equivalently the model \eqref{eq:RLE_channel_cavity}, \eqref{eq:RLE_channel_bulk} with  $\zeta=\lambda$, rows of $\bx$ are completely exchangeable under the law $\EE\langle \, \cdot\,\rangle$. Therefore the random cavity marginals \eqref{cavity marginal} corresponding to SNR $\lambda$ should have the same law as the bulk marginals. This is achieved by matching the SNRs in the effective scalar observation channels for the cavity coordinates, $\sqrt{r_*(\lambda,\sigma)}H_i+\xi_i$, and bulk coordinates, $\sqrt{\sigma}X_{i\mu}+Z_{i\mu}'$. This yields the third missing consistency equation needed for $\sigma_*(\lambda)$ in addition to the two we already have for $(q,r)$:
\begin{align}
    \label{eq:consistency_sigma}
    r_*=\frac{\lambda J(\sigma_*)}{1+\lambda\alpha(1-J(\sigma_*)q_*)}=\sigma_*.
\end{align}
By \eqref{eq:def_J} $J$ is the overlap that any bulk row in $\bar\bx$ has with the corresponding ground truth row in $\bar\bX$. As intuition may suggest, matching the laws of the cavity and bulk marginals also equates the two overlaps $J$ and $q_*$. To see this, it is sufficient to realize that the function $F(r)=\EE H\langle \eta\rangle_r$ is monotonic in $r\geq 0$, and thus invertible. Then, by a simple application of it to both sides of \eqref{eq:consistency_sigma} we readily get $q_*=J(\sigma_*)$ as predicted. The two conditions $\sigma_*=r_*$ and $J(\sigma_*)=q_*$ are thus equivalent.

Using this knowledge, we recast \eqref{eq:RLE_saddlepoint_r} as
\begin{align}
    &q=\EE_{\xi,H} [H\langle \eta \rangle_{r}],\qquad r=\frac{\lambda q}{1+\lambda\alpha(1-q^2)}.
    \label{eq:Kabashima_r}
\end{align}
These are solved via fixed point iterations.
Their solutions shall be denoted by $q_*(\alpha,\lambda)$ and $ r_*(\alpha,\lambda)$. Using the relation \eqref{eq:MIRLE_vs_psiRLE}, and that the solution is extremizing $\phi^{\rm RLE}_{\rm eff}$, we finally get an ansatz for the MMSE by using the I-MMSE relation \eqref{eq:I-YMMSE_RLE} but for the model with the mean-field bulk ansatz (here approximate equalities are up to vanishing corrections $O(1/N)$):
\begin{align}\label{eq:MMSE_Matrix_facto}
    \text{MMSE}\approx {\rm YMMSE}(\lambda,\lambda)&\approx 2\frac{\partial}{\partial\lambda }\Big(\frac{\lambda}{2}-\phi^{\rm RLE}_{\rm eff}(\alpha,\lambda,\sigma_*)\Big)\nonumber\\
    &=\frac{1-q_*(\alpha,\lambda)^2}{1+\lambda\alpha(1-q_*(\alpha,\lambda)^2)}.
\end{align}

Note that $q_*=r_*=\sigma_*=0$ is always a solution for any $(\alpha,\lambda)$. We call it the \emph{paramagnetic solution} as it corresponds to no alignment between posterior sample and ground truth. However, the MMSE for this solution,
\begin{align*}
    {\rm MMSE}_{\rm para} = \frac 1{1+\alpha \lambda},
\end{align*}
is a non-trivial decreasing function of the SNR. This is because denoising is still possible, for any $\lambda>0$, when the matrix $\bX$ has no macroscopic overlap with a posterior sample ($q_*=0$ which is both the row or column overlap under ansatz \eqref{eq:aux_channel_bulk}). We have checked that apart from very low SNR, this solution is unstable in the sense that any small positive initialisation for the iterative solution of the saddle point equations will yield another solution.

Given the MMSE, one could be tempted to use the I-MMSE formula \eqref{eq:I-MMSE} to integrate it w.r.t.\ the SNR $\lambda$ in order to obtain the MI of matrix denoising. However, there could be multiple solutions $q_*(\alpha,\lambda)$ associated to given values of $\lambda$ and $\alpha$, which impeaches this naive reasoning. A more careful approach is required. Given equations \eqref{eq:Kabashima_r}, we reconstruct in Appendix~\ref{app:MI_reconstruction_potential} the \emph{mutual information potential} that generates them, and we shall always choose solutions that minimize it. Indeed, there is a unique potential function $\iota(r,q;\alpha,\lambda)$ that yields the system \eqref{eq:Kabashima_r} when extremized over $(r,q)$, and that at the same time satisfies the I-MMSE relation \eqref{eq:I-MMSE} with the previously found solution \eqref{eq:MMSE_Matrix_facto} as MMSE, namely such that
\begin{align*}
    \frac{d}{d\lambda}\iota(r_*,q_*;\alpha,\lambda)=\frac{1}{4}\frac{1-q_*(\alpha,\lambda)^2}{1+\lambda\alpha(1-q_*(\alpha,\lambda)^2)}\,.
\end{align*}
This potential is
\begin{align}
&\iota(r,q;\alpha,\lambda):=
    \frac{rq}{2}+\frac{1}{4\alpha}\ln\big(1+\lambda\alpha(1-q^2)\big)\nonumber\\
    &\ -\EE_{\xi\sim\mathcal{N}(0,1),H\sim P_X}\ln\int dP_X(\eta)e^{(\xi\sqrt{r}+Hr)\eta-\frac{1}{2}r\eta^2}.\label{eq:MI_Matrix_facto_pot}
\end{align}
The mean-field prediction for the MI of the original problem is then
\begin{align}\label{eq:MI_Matrix_facto}
\iota(\alpha,\lambda)&:={\rm extr}_-\big\{ \iota(r,q;\alpha,\lambda)\big\},
\end{align}
where extremization is w.r.t.\ $(r,q)$, and ${\rm extr}_-$ selects the solution of the saddle point equations \eqref{eq:Kabashima_r} minimizing the potential. Let us stress again that this formula will \emph{not} be our final prediction for the MI in the whole phase diagram $(\alpha,\lambda)$, but just for a sub-region of it as explained in the next section.

This formula recovers the known replica symmetric potential for the rank-1 case \cite{XXT,lelarge2019fundamental} (which is always correctly described by mean-field theory) when $\alpha\to 0$:
\begin{align*}
&\iota(r,q;0,\lambda)=
    \frac{rq}{2}+\frac{\lambda}{4}(1-q^2)\nonumber\\
    &\qquad-\EE_{\xi,H}\ln\int dP_X(\eta)e^{(\xi\sqrt{r}+Hr)\eta-\frac{1}{2}r\eta^2}.
\end{align*} 
{We observe that, in the converse limit $\alpha\to\infty$ the formula trivializes and is optimized at $r=q=0$. This is due to our choice of normalizations in channel \eqref{eq:channel1}. We refer to \cite{perturbative_Maillard21,semerjian2024matrix} for a different scaling where this limit is meaningful and thoroughly investigated.} {In addition, Appendix~\ref{app:large-lambda-limit-MI} shows that $\iota(\alpha,\lambda\to\infty)$ approaches the entropy of the prior $\mathcal{H}(X)$ as it should; it is also apparent from \figurename~\ref{fig:MI_vs_SNR_differentalpha} appearing later on.}

From the more general derivation in Appendix \ref{app:multiscale_MFT}, it is clear that our multiscale mean-field theory can handle generic observation channels as long as the observations remain independent conditionally on the signal, namely $Y_{ij}\sim P_{\rm out}(\,\cdot\mid\bX_i\cdot\bX_j/\sqrt{N})$ conditionally independent. The fixed point equation $q=\EE H\langle\eta \rangle_r$ is unchanged, whereas $r=g_{\rm out}(q)$ with a channel-dependent $g_{\rm out}$. The consistency $r_*=\sigma_*$ still holds, as well as $J(\sigma_*)=q_*$.

\vspace{5pt}
\noindent\textbf{Relation to the replica method with Gaussian ansatz of Sakata and Kabashima.}
The same formula for the MI \eqref{eq:MI_Matrix_facto} can be produced using the earlier approach of \cite{SK13EPL}, also used in \cite{SK13ISIT,KMZ_DL-2013,Marc-Kabashima} in a Bayesian context. We have adapted their method to our model in Appendix~\ref{appendix:Kabashima_derivation} for completeness, which is close to what is done in \cite{Marc-Kabashima}. This alternative derivation does not go through a cavity computation but instead uses a non-standard replica method at the scale $O(MN)$. From our understanding, the first \enquote{Gaussian ansatz} they use is to approximate  $(z^a_{ij}=\bx_i^a\cdot\bx_j^a/\sqrt{N})_{i,j=1,\ldots,N, a=0,\ldots,n}$, where $a$ is the replica index, as a jointly Gaussian family. The second ansatz is on their covariance. Somewhat counter intuitively, they consider them independent. Overall, this is equivalent to consider $(z_{ij}^a)_{i,j}$ as a GOE matrix, with a semicircular spectral density. This is hard to justify a-priori, given that this matrix is of Wishart type, whose spectrum follows the Marcenko-Pastur law. This assumption may however be more justifiable as $\alpha\approx M/N$ gets large because in this case, a Wishart matrix resembles more a GOE one with a mean. As discussed in \cite{perturbative_Maillard21}, this approximation is supposedly more accurate in this case (as confirmed by our numerics), and even possibly exact in the $\alpha\to\infty$ limit. Interestingly, \eqref{eq:MI_Matrix_facto_pot} also yields the rank-1 formula when $\alpha\to 0$, which signals that the approach is actually taking into account both $\alpha$ limits.

As mentioned in Section~\ref{sec:related_works}, equations \eqref{eq:Kabashima_r} were deemed incorrect for a long time \cite{comments,perturbative_Maillard21}. In fact, in the past they were tested against the exact MI for Gaussian prior only \cite{perturbative_Maillard21}. But this is precisely the case in which they are not supposed to hold for all $(\alpha,\lambda)$, because the denoising phase never ends for Gaussian prior and so the HCIZ integral is the right tool. With our contribution we redeem the replica equations {found more than a decade ago} in \cite{SK13EPL,SK13ISIT,KMZ_DL-2013,Marc-Kabashima} using an independent, original, mean-field approach, by showing that they actually yield excellent results but \emph{in the factorisation phase only} (and good ones in the denoising phase for large $\alpha$), a phase which exists only for non-Gaussian priors. Despite leading to the same equations, the multiscale theory we have presented is based on different, interpretable, hypotheses. {The fact that the earlier replica method and our cavity-based approach yield the same equations is not evident a-priori given the non-standard Gaussian ansatz used in the former. This is an indication of their validity when the concentration-of-measure effects needed to reduce the high-dimensional integrals entering the free entropy to saddle-point equations over order parameters take place. This, we claim, happens only in the factorisation phase, which was not understood until now.}

\section{Complete ansatz taking into account the prior universality in the denoising phase}\label{sec:completeansatz}

\subsection{Complete ansatz for the whole phase diagram}\label{sec:completeansatz_A}

We have confirmed numerically that the strong prior universality conjecture of \cite{semerjian2024matrix} holds but only before the denoising-factorisation transition. This should be incorporated to correct the inaccuracies of the mean-field approximations in this region. Indeed, the MMSE predicted by the mean-field theory does not match the known exact one for Gaussian prior, nor for discrete priors in the denoising phase. This is somehow expected. The reason is that our ansatz for the bulk measure, corresponding to taking the posterior associated with the channel $\bY'_{\rm eff}=\sqrt{\sigma}\,\bX+\bZ_{\rm eff}'$, breaks an \enquote{effective rotational symmetry} that seems to hold in the denoising phase as we argue in Section~\ref{sec:findings_overlap}. The rotational symmetry breaking comes from the fact that this channel corresponds to a direct observation of $\bX$ instead of $\bX\bX^\intercal$ which ``mixes'' the entries. Therefore there is no reason for the theory to be correct in the denoising phase. 
Yet, in the factorisation phase the mean-field prediction for the MMSE is in agreement with the exponential decay in \figurename~\ref{fig:MMSE_log}, due to the strong alignment between posterior \enquote{patterns} $(\bx_\mu)_{\mu \le M}$ and the planted ones $(\bX_\mu)_{\mu \le M}$ (see Section \ref{sec:findings_overlap}). This implies that in this phase the behavior of a row $\beeta$ is indeed dominated by its direct interaction with the associated planted row $\bH$ (mediated by $H^{\rm row}_{N,M}$) while the effect of the bulk simplifies greatly and is captured by the mean-field ansatz \eqref{eq:aux_channel_bulk}. We conclude that the denoising phase is a \enquote{matrix model regime} which needs to be analyzed by its own set of techniques (i.e., proofs of universality and the HCIZ integral). In contrast, the factorisation phase is a \enquote{mean-field regime} amenable to statistical physics methods. We are thus confident that the mean-field solution is accurate in this region, as backed up by the numerics and its robustness to attempts on improving it, see Section~\ref{sec:improv}, as well as the matching with the replica approach of Sakata and Kabashima \cite{SK13EPL}.

\begin{figure*}[t!!]
    \centering \includegraphics[width=0.41\linewidth,trim={0 0 4.1cm 0},clip]{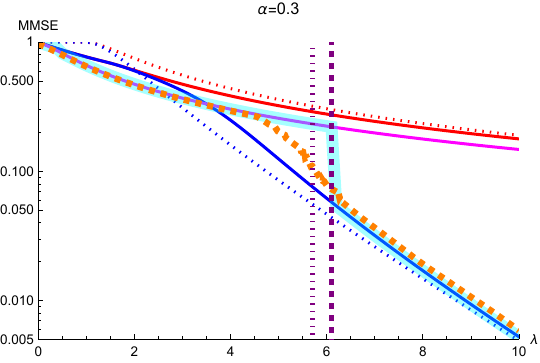}
    \includegraphics[width=0.57\linewidth]{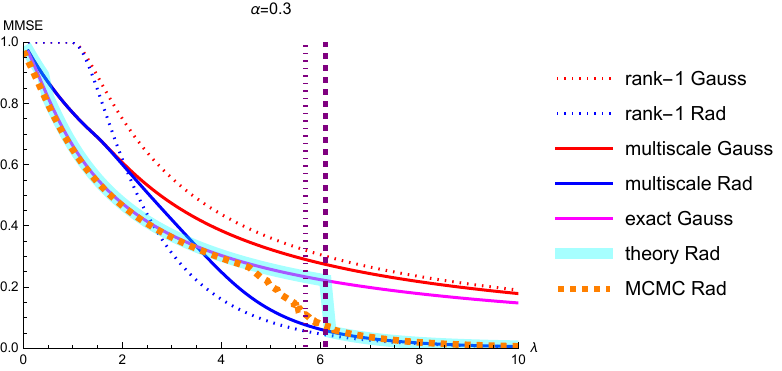}
    
    \includegraphics[width=0.41\linewidth,trim={0 0 4.1cm 0},clip]{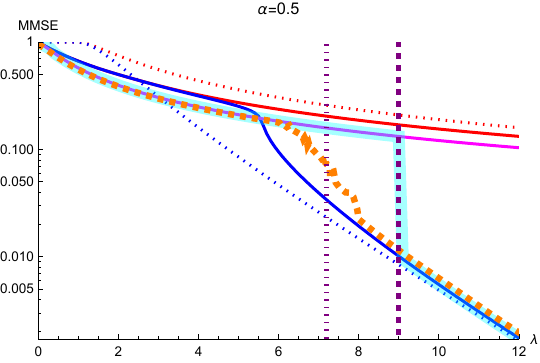}   \includegraphics[width=0.57\linewidth]{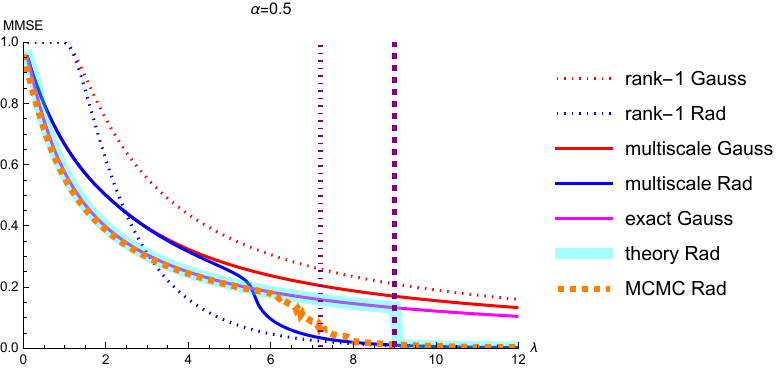}   
    
    \includegraphics[width=0.41\linewidth,trim={0 0 4.1cm 0},clip]{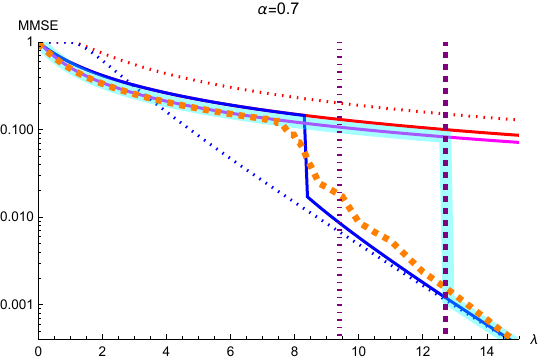}
    \includegraphics[width=0.57\linewidth]{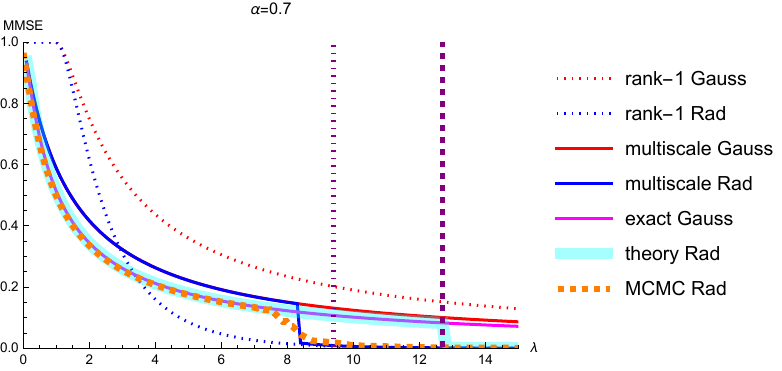}
    
    \caption{Thermodynamic limit of the MMSE as function of the SNR for $\alpha\in\{0.3,0.5,0.7\}$ with Gaussian and Rademacher priors. Logarithmic scale on the left, linear scale on the right. We plot $i)$ the rank-1 MMSEs, $ii)$ the \enquote{multiscale} curves correspond to the mean-field equation \eqref{eq:MMSE_Matrix_facto} with the equilibrium solution selected, $iii)$ the \enquote{exact Gauss} curve is the true Gaussian MMSE obtained via \eqref{MMSEspherical}, which also corresponds to the performance of the RIE for any prior with zero mean and unit variance, $iv)$ the cyan \enquote{theory Rad} is our theoretical prediction for the MMSE with Rademacher prior corresponding to four times the derivative of $\iota^*(\alpha,\lambda)$ \eqref{conjecture} w.r.t. $\lambda$, and $v)$ \enquote{MCMC Rad} are the Monte Carlo results with Rademacher prior (see Tables \ref{tab:exp_data} and \ref{tab:fig_exp_ref} for details). Finally $vi)$ the rightmost vertical line pinpoints the location of the infinite $N$ transition $\lambda_c$ \eqref{transTheory}, while the one on the left is its approximate location once taking into account finite size effects, see the caption of \figurename~\ref{fig:MI_vs_SNR_differentalpha} and the text for an explanation.}
    \label{fig:MMSE_th}
\end{figure*}

\begin{figure*}[t!!]
    \centering  \includegraphics[width=0.39\linewidth,trim={0 0 4.6cm 0},clip]{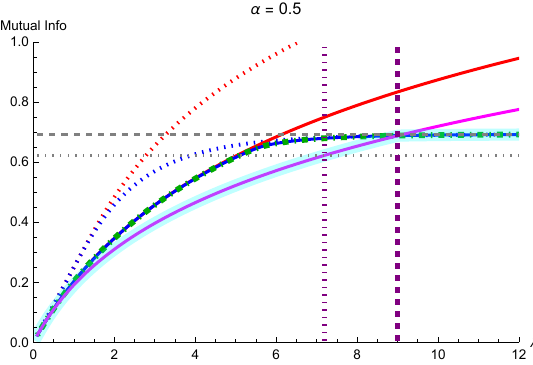}
    \includegraphics[width=0.59\linewidth]{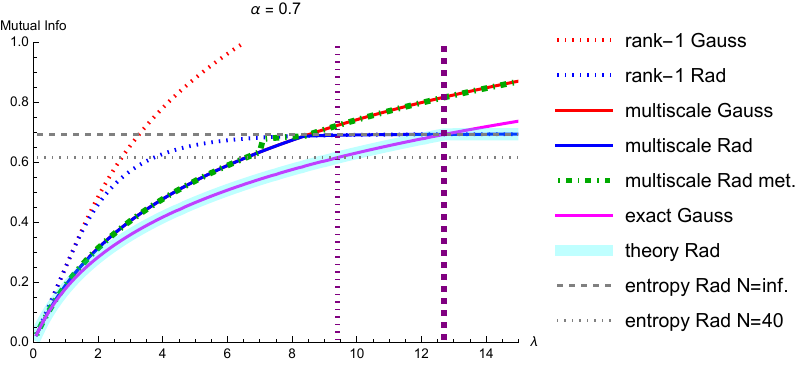}
    
    \caption{Thermodynamic limit of the mutual information as function of the SNR for $\alpha\in\{0.5,0.7\}$ with Gaussian and Rademacher priors. We plot $i)$ the rank-1 MIs, $ii)$ the \enquote{multiscale} curves referring to $\iota(\alpha,\lambda)$ \eqref{eq:MI_Matrix_facto} with the equilibrium solution selected, $iii)$ the \enquote{multiscale met.} curve corresponds to the metastable solution (which only exists for Rademacher prior), $iv)$ the \enquote{exact Gauss} curve is the exact MI for Gaussian prior $\iota^{\rm G}(\alpha,\lambda)$ \eqref{MIspherical}, $v)$ the \enquote{theory Rad} is the complete ansatz $\iota^*(\alpha,\lambda)$ \eqref{conjecture} for Rademacher prior once taking into account universality in the denoising phase, $vi)$ the entropy density $\mathcal{H}(\bX\bX^\intercal)/(NM)$ for $N\to+\infty$, equal to $\ln 2$, and its finite-size approximation \eqref{entropyXX} given by $\ln 2-\ln(2^M M!)/(MN)$ with $N=40,M=\alpha N$. Finally $vii)$ the rightmost vertical line pinpoints the crossing of the exact Gaussian MI with the multiscale mean-field prediction, i.e., the infinite size transition $\lambda_c$ \eqref{transTheory}, while the leftmost vertical line shows the crossing of the Gaussian MI with $\ln 2-\ln(2^M M!)/(MN)$ which approximates the location of the denoising-factorisation transition taking into account finite size effects.
    }    \label{fig:MI_vs_SNR_differentalpha}
\end{figure*}

 \begin{figure}[t!]
    \centering
\includegraphics[width=\linewidth]{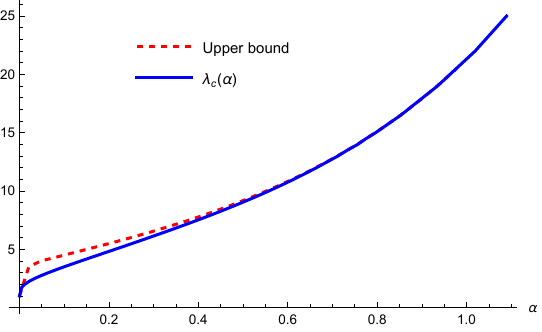}
    \caption{Thermodynamic denoising-factorisation phase transition $\lambda_c(\alpha)$ given by \eqref{transTheory} and its upper bound given by the r.h.s. of \eqref{lambdac}, for Rademacher prior.}
    \label{fig:trans}
\end{figure}

From \figurename~\ref{fig:multiscalevsMCMC} and the related discussion, we deduce that there are only three possible thermodynamic states: one corresponds to the ``universal RIE branch'', another to the ``non-universal factorisation branch'' reached from informative initialisation, and finally the dynamical glass state. Given our numerical evidence as well as the work of Semerjian \cite{semerjian2024matrix}, we select the universal RIE branch as the equilibrium state for low SNRs, which corresponds to the Gaussian MI \eqref{MIspherical}. After the transition yet to be identified, we instead select the factorisation state described by the mean-field formula \eqref{eq:MI_Matrix_facto}.

{The question now is: \emph{how to locate the denoising-factorisation phase transition, assuming there is only one transition?} The single transition scenario is backed up by our numerical experiments. 
We know three key properties of the asymptotic MI that can be used to locate the transition between the universal and mean-field branches:
$i)$ the MI is continuous in $\lambda$, $ii)$ it is also concave in $\lambda$ (its second derivative is the first derivative of the MMSE, which must decrease when $\lambda$ increases), and $iii)$ for a discrete prior, we know from \eqref{eq:finite_size_correct_MI} that the MI saturates to $\mathcal{H}(X)$ when $\lambda\to+\infty$. Therefore, the only possibility to link the Gaussian solution at low SNR to the mean-field one at higher values (which does approach $\mathcal{H}(X)$ exponentially fast in $\lambda$ by Appendix~\ref{app:large-lambda-limit-MI}) is to locate the transition as the point where they cross, namely}
\begin{align}
    \lambda_c(\alpha)=\sup\big\{\lambda\ge 0 :\iota^{\rm G}(\alpha,\lambda)\le \iota(\alpha,\lambda)\big\}.\label{transTheory}
\end{align}
As we show in the next section this is in good agreement with our finite size experiments once we take properly into account finite size effects. Consistently, our final ansatz $\iota^*$ for the MI is then
\begin{align}
    \iota^*(\alpha,\lambda)=\min\big\{
        \iota^{\rm G}(\alpha,\lambda),
        \iota(\alpha,\lambda)\big\}.\label{conjecture}
\end{align}

{The choice of selecting the Gaussian MI in the denoising phase for an informative prior like Rademacher is also well motivated by the following argument. An informative (structured) prior carries more information than the Gaussian one, and thus the associated MMSE cannot be larger than in the Gaussian case given the same data. Therefore, for such a prior, the mean-field theory cannot be correct when it predicts an MMSE \eqref{eq:MMSE_Matrix_facto} larger than the exact one for Gaussian prior \eqref{MMSEspherical}. The Gaussian MMSE is thus at least a better approximation to the true one in such a region. \figurename~\ref{fig:MMSE_th} shows that there is indeed a low SNR region where this happens and therefore the mean-field solution is incorrect. We claim that in this region the Gaussian curve is not just a better approximation to the truth, but is actually exact, in agreement with \cite{semerjian2024matrix} and our numerics. On the other hand, as illustrated by \figurename~\ref{fig:MMSE_log} and further discussed in the following section, the choice of the mean-field solution beyond a point is validated by the excellent agreement with our numerical experiments. Then, from the single transition hypothesis and the aforementioned properties $i),ii),iii)$ of the MI, the transition \eqref{transTheory} emerges.}

{ The form of the final ansatz $\iota^*$ shares features of classical limit formulas in spin glasses such as the Parisi formula \cite{parisi1979infinite, parisi1980sequence}. In the Parisi formula, the limit free entropy $F^*(\beta)$ at inverse temperature $\beta$ is also given by a minimization over different candidate ``branches''
\[
F^*(\beta) = \min \big\{ F_{{\rm RS}}(\beta) , F_{{\rm 1RSB}}(\beta), F_{{\rm 2RSB}}(\beta), \dots \big\},
\]
where $F_{{\rm RS}}(\beta), F_{{\rm 1RSB}}(\beta), \dots$, etc, are the respective replica symmetric, $1$-replica symmetry breaking step$, \dots$, approximations of the free entropy. But there are important differences too. Firstly, our model is replica symmetric as it lives on its Nishimori line. But more importantly, in contrast to our setting, a key difference in spin glasses and usual inference problems is that in all their phases (i.e., regimes of symmetry breaking), these models remain mean-field in nature, and thus the different branches of symmetry breaking can be described by one formalism yielding an overarching formula. This is \emph{not} the case for extensive-rank matrix denoising: the current techniques to compute the MI in the denoising and factorisation phases are very different, which reflects their fundamentally distinct nature. We believe it is unlikely there exists a mean-field theory describing both concurrently, unlike what happens in the Parisi formula. This is the main reason why we had to employ matrix model techniques in the denoising phase.}

\subsection{Theoretical predictions and consistency with the numerical experiments}\label{sec:sp_equations}

We plot in \figurename~\ref{fig:MMSE_th} the mean-field solution for the MMSE \eqref{eq:MMSE_Matrix_facto} obtained by iterating \eqref{eq:Kabashima_r} (dark blue curves), and the associated asymptotic MI \eqref{eq:MI_Matrix_facto} in \figurename~\ref{fig:MI_vs_SNR_differentalpha}. 
When solving \eqref{eq:Kabashima_r} by fixed point iteration, we call the initialisation $q^{t=0}=1$ \emph{informative}. The \emph{uninformative} initialisation is instead $q^{t=0}=\varepsilon\ll 1$ but non-zero. From these two obtained solutions, the \emph{equilibrium} $(q_*,r_*)$ is the one that minimizes \eqref{eq:MI_Matrix_facto_pot}. The other solution (we always observed at most two) is called \emph{metastable}. It matches the equilibrium solution if unique, as for $\alpha=0.5$ in \figurename~\ref{fig:MI_vs_SNR_differentalpha}, while they can differ for $\alpha=0.7$.

Recall however that before the transition \eqref{transTheory} we claim that the correct branch is instead the universal one. For all plots displayed, the complete ansatz taking into account universality at low SNR $\iota^*(\alpha,\lambda)$  \eqref{conjecture} (cyan curves for both MI and MMSE) correctly describes the key features of the problem evidenced by our numerical experiments. In particular we stress the very good agreement between four times the $\lambda$-derivative of $\iota^*$ and the finite size MMSEs in the factorisation phase, see \figurename~\ref{fig:MMSE_log} and \figurename~\ref{fig:MMSE_th}, but also in the denoising phase, see \figurename~\ref{fig:MCMCvsRIE} and \ref{fig:MMSE_th}. Hereby we describe our theoretical findings focusing on the complete ansatz $\iota^*$ if not specified otherwise, comparing them with those of Section \ref{sec:Motivation}. 

\vspace{3pt}
$i)$ \emph{Agreement with the mutual information and transition upper bounds.}\ $\iota^*$ remains upper bounded by the entropy $\mathcal{H}(X)$ of the prior as it should, see App.~\ref{app:large-lambda-limit-MI} and \figurename~\ref{fig:MI_vs_SNR_differentalpha}. This implies that the theoretical transition \eqref{transTheory} occurs before the rigorous upper bound \eqref{lambdac}, which becomes tighter as $\alpha$ increases, see \figurename~\ref{fig:trans}. This non-trivial consistency check entails that for high SNR, the theoretical MMSE decays exponentially with it, as supported by the numerical experiments, see \figurename~\ref{fig:MMSE_log} and \ref{fig:MMSE_th}. This exponential decay in the factorisation phase, which is asymptotically in $\lambda$ the same as for the rank-1 case with Rademacher prior, hints at the fact that in this phase the posterior patterns $(\bx_\mu)$ are weakly coupled, and the problem resembles (but is not equivalent to) $M$ independent rank-1 problems with properly tuned SNRs. 

\vspace{3pt}
$ii)$ \emph{Universality at low SNR.}\   Despite the mean-field solution $\iota$ being only approximate in the denoising phase, at low SNR it captures universality: the mean-field curves $\iota(\alpha,\lambda)$ for Gaussian and Rademacher priors do match, see \figurename~\ref{fig:MI_vs_SNR_differentalpha}; it is also evident at the level of the MMSE in \figurename~\ref{fig:MMSE_th}. This is in agreement with experiments \figurename~\ref{fig:ExactMIfiniteSize} and \ref{fig:MCMCvsRIE}. When $\alpha$ increases, this matching lasts until close to the crossing of $\iota$ and the Rademacher entropy: this signature of universality becomes more pronounced. It shows that $\iota$ becomes more accurate and consistent at larger $\alpha$ (but it also is exact at $\alpha=0$). Obviously, $\iota^*$ is universal before the transition by construction. 

\vspace{3pt}
$iii)$ \emph{Denoising-factorisation transition.}\ $\iota^*$ encodes the denoising-factorisation transition observed in our numerical experiments. Indeed, for all $\alpha$'s tested we see a change of behaviour in the $\lambda$-derivative of $\iota^*(\alpha,\lambda)$ separating the regime where the MI is prior-independent, from the factorisation phase where the MI for discrete prior approaches the entropy of the prior exponentially fast in $\lambda$. The MI for the Gaussian prior instead continues to increase smoothly. The algebraic versus exponential decay of the MMSE with the SNR observed numerically is quantitatively caught by the theory, see \figurename~\ref{fig:MMSE_th}.

For what concerns the transition location, it seems from \figurename~\ref{fig:MMSE_th} that there is a non-negligible gap with what $\iota^*(\alpha,\lambda)$ predicts. We argue here that this is due to the finite size of the simulated system, and can be accounted for. In fact, our arguments in Section~\ref{sec:Motivation} make it clear that for finite sizes the MI for Rademacher prior must not saturate to $\ln 2$ due to the degeneracy induced by column permutations and signs, which explains precisely the gap to $\ln2$ in the left panel of \figurename~\ref{fig:ExactMIfiniteSize}. Since the empirical MI curves closely follow $\iota^{\rm G}(\alpha,\lambda)$ at low SNR, they must start saturating earlier, namely at a smaller $\lambda$ than the asymptotic one $\lambda_c$, and towards a lower value than $\ln 2$. For instance, for $\alpha=0.5$ and $N=40$ saturation occurs at MI $\approx 0.6228$, which is also reported as the lowest dotted horizontal line in the left panel of \figurename~\ref{fig:MI_vs_SNR_differentalpha}. The curve $\iota^{\rm G}$ crosses that value at $\lambda\approx 7.2$, which gives us a rough upper bound on the transition location for that size, in good agreement with what is observed in \figurename~\ref{fig:ExactMIfiniteSize} and \ref{fig:MCMCvsRIE} (even if formally we cannot talk about phase transition at finite size). But we see that the change in the behaviour of the derivative of the finite size MI is not sudden and starts even before $7.2$. This is consistent with \figurename~\ref{fig:MMSE_th} for $\alpha=0.5$, where the MCMC curve detaches from the Gaussian MI around $\lambda \approx 6.5$. Nevertheless, observe that the finite-size-corrected approximation to the transition at $7.2$ falls nicely in the middle of the \enquote{crossover region} between the two phases for the MMSE. The same correcting approach works for the other values of $\alpha$. 

Let us explain in more detail the way we computed corrections to the transition location. The procedure starts by computing the finite size entropy $\mathcal{H}(\bX\bX^\intercal)/(NM)$ using \eqref{entropyXX}. Then we find the SNR where $\iota^{\rm G}$ attains this value: this gives an upper bound on the finite size ``transition'', which we expect to be close to it given the numerical and theoretical evidence. This implies that the larger $\alpha$, the smaller the slope of $\iota^{\rm G}(\alpha,\lambda)$ in $\lambda$, and thus the larger the gap between the asymptotic transition point $\lambda_c$ and its approximate location for finite size, see \figurename~\ref{fig:MI_vs_SNR_differentalpha}.

On top of this gap increasing with $\alpha$, the crossover region for the transition should also get wider. To see it we notice that the finite size corrections to the MI, estimated through \eqref{eq:finite_size_correct_MI} as $\ln (2^M M!)/(MN)$ ($\approx 5\cdot 10^{-2}$ in all our experiments), are large given the size of the system ($MN\approx 1000$ in all our experiments). As comparison, in the rank-1 spiked Wigner model, finite size corrections to the MI instead vanish as $1/N$ which is faster relatively to the system size. For size $1000$, this yields a correction roughly two orders of magnitude smaller than here. Notice also that finite size corrections increase in $\alpha$, since they are directly linked to the number of column permutations. In summary, thanks to the aforementioned correction procedure, the predicted finite size \enquote{transition} does fall in the middle of the crossover region for the MMSE, which gets wider with larger $\alpha$. {Out of these crossover regions, the theory proves to be extremely accurate. Despite the exponential decays and the modest sizes that can be simulated, our predictions for the MMSEs deep in the factorisation phase remain precise even at a $y$-scale of $10^{-3}$.}

\vspace{3pt}
$iv)$ \emph{Spinodal transition and statistical-to-computational gap.}\ According to our theory $\iota^*$, the thermodynamic denoising-factorisation phase transition at $\lambda_c(\alpha)$ is of first order, i.e., with a discontinuity in its $\lambda$-derivative which is the MMSE over four by \eqref{eq:I-MMSE}, for any $\alpha>0$. 
We have already discussed based on \figurename~\ref{fig:multiscalevsMCMC} that at low SNR, the only thermodynamic state present is the universal RIE state. This implies that before the first order transition, there must be a spinodal transition linked to the appearance of the informative factorisation state in addition to the RIE state, a point around which the RIE state seems to shatter into a dynamical glassy state observed in \figurename~\ref{fig:multiscalevsMCMC}. A spinodal transition followed by the coexistence of an informative and glassy states is a typical mechanism at the root of algorithmic hardness in high-dimensional inference \cite{PhysRevLett.65.1683,antenucci2019glassy}. This connection should however be considered with some precaution, as it has been noticed in various settings that the onset of glassiness does not always prevent inference algorithms, including sampling-based, to perform well \cite{antenucci2019glassy,angelini2023limits,decelle2011asymptotic}. $\iota^*$ is a pure equilibrium theory, and is therefore unable to detect the location of the spinodal transition. One would need to modify it in order to take into acount out-of-equilibrium glassy effects linked to replica symmetry breaking \cite{mezard1990spin,antenucci2019glassy}, which is out of the scope of the present paper. Given that we are in the Bayes-optimal setting, the thermodynamic equilibrium should anyway be described by a replica symmetric theory \cite{nishimori2001statistical,barbier2019overlap,barbier2022strong}.

It is worth noting that even though the mean-field MI $\iota(\alpha,\lambda)$ is inexact in the denoising phase, the potential \eqref{eq:MI_Matrix_facto_pot} captures the presence of a spinodal transition for $\alpha\ge 0.7$. Indeed, when $\alpha=0.7$ the mean-field theory exhibits a metastable branch around $\lambda\approx 7$, see the discontinuity in the green curve in \figurename~\ref{fig:MI_vs_SNR_differentalpha}. It detaches from the equilibrium mean-field curve through a spinodal transition where the informative fixed point appears. We thus interpret the mean-field theory as an \emph{approximation} in the denoising phase (which becomes exact after the transition, while $\iota^*$ is deemed to be exact everywhere), in the sense of the critical phenomenology it implies that cannot be probed from the equlibirium theory $\iota^*$ alone.

\section{Probing the phases through the overlap matrix, and factorisation}\label{sec:findings_overlap}
So far, we have focused on the denoising task of recovering the square signal $\bX\bX^\intercal$. We have divided the phase diagram in two thermodynamic phases: the universal denoising phase where the RIE is optimal, and the factorisation phase where it is not, because exploiting the prior and \emph{factorised} structure of the signal beyond its asymptotic spectrum $\rho_{XX^\intercal}$ is needed there to better denoise. 
It is thus natural to wonder if optimal denoising is related to the \emph{factorisation task}, i.e., estimating the single factor $\bX$ non-trivially. Non-trivial estimation means finding an estimator $\hat \bX=\hat \bX(\bY)\in \mathbb{R}^{N\times M}$ verifying 
\begin{align}
\lim_{N\to+\infty} \frac1{NM}\max_{\mathbf{\Pi}\in \bar \pi_M}\Tr\big(\bX^\intercal\hat \bX\mathbf{\Pi}  \big)>0,   \label{optOverlap} 
\end{align}
where $\bar\pi_M$ is the set of $M\times M$ signed permutation matrices. It is not clear that it is possible given the numerous symmetries and invariances of the model when $M\to+\infty$ with $N$. When $M$ is independent of $N$, if denoising is possible then factorisation is, because the number of symmetries to resolve in order to infer $\bX$ given an estimator $\widehat {\bX\bX^\intercal}$ of $\bX\bX^\intercal$ remains generally finite; e.g., for $M=1$, the eigenvector of $\widehat {\bX\bX^\intercal}$ with largest eigenvalue estimates $\bX$ up to a global sign.

We will explore here this question by getting a more detailed description of the phases through the analysis of a richer order parameter than the MMSE, namely the \emph{overlap} matrix:
\begin{align*}
    \bQ:=\frac{\bX^\intercal \bx}{N}=\Big(\frac{\bX_\mu^\intercal \bx_\nu}{N}\Big)_{\mu,\nu\le M}\in\mathbb{R}^{M\times M}, 
\end{align*}
where $\bx\sim P(\,\cdot\mid \bY)$ is a posterior sample, and $(\bx_\mu)$ are the columns/patterns of $\bx$. The MMSE can be written directly in terms of the overlap:
\begin{align}
    {\rm MMSE}&=\frac1{M}{\EE\Big\|\frac{\bX\bX^\intercal}N\Big\|^2}-\frac1M\EE\langle\| \bQ\|^2\rangle\label{Q-MMSE},
\end{align}
which follows from the Nishimori identity. From \eqref{optOverlap}, a posterior sample $\bx$ is a non-trivial estimator for the factorisation task if $\frac1M\Tr (\bQ \mathbf{\Pi}^*)>o_N(1)$ for a $\bx$-dependent permutation $\mathbf{\Pi}^*$. {The overlap can also be defined in terms of two conditionally i.i.d. samples from the posterior rather than using the ground truth: $\bx^{(1)\intercal} \bx^{(2)}/N$. By Bayes-optimality and the Nishimori identities, the conclusions of this section based on $\bQ$ would still hold.}

\subsection{Behavior of the overlap in the two phases: \texorpdfstring{\\}{}can \texorpdfstring{$\bX$}{} be inferred in addition to \texorpdfstring{$\bX\bX^\intercal$}{}?}

It is insightful to distinguish between diagonal and off-diagonal contributions of the overlap to the MMSE. Note that from the independence of the data on the permutation of the columns of $\bX$ (and their signs for symmetric priors), the induced invariance of the posterior implies that the columns of posterior samples $\bx$ obtained from the Monte Carlo sequence must be sorted appropriately to match the ones of the ground truth in order to access a meaningful overlap. In other words, we need an efficient way to approximate the solution of the optimisation problem \eqref{optOverlap}. Numerically, this is achieved via iteration w.r.t. the columns of $\bX$. In the first iteration, one finds the column of $\bx$ showing the largest absolute inner product with $\bX_1$, labeling its index as $\pi(1)$. For iterations $\mu = 2, \ldots, M$, choose $\pi(\mu)$ as 
$\pi(\mu) = {\operatorname{argmax}}_{{\nu \in S_\mu } }|\bX_\mu^\intercal \bx_\nu|$, where $S_\mu$ is the set of indices from $1$ to $M$ excluding $\pi(1), \ldots, \pi(\mu-1)$, see Algorithm~\ref{alg:greedy} below:

\begin{algorithm}[H]
\caption{Greedy column and sign alignment}
\begin{algorithmic}[1]
\Statex \textbf{Input:} Reference matrix $\bX$, ordered matrix $\bx$
\Statex \noindent \textbf{Output:} Permutation $\pi$, signs $\sigma$
\Statex \textbf{Initialize:} $S \gets \{1, 2, \ldots, M\}$ (Set of available indices)
\Statex \textbf{Initialize:} $\pi$ and $\sigma$ as empty arrays of length $M$
\For{$\mu = 1$ to $M$}
        \State $\pi(\mu) \gets \operatorname{argmax}_{\nu \in S} |\mathbf{X}_\mu^\intercal \mathbf{x}_\nu|$
        \State $\sigma(\mu) \gets  {\rm sign} (\mathbf{X}_{\mu}^\intercal \mathbf{x}_{\pi(\mu)}) $
    \State $S \gets S \setminus \{\pi(\mu)\}$
\EndFor
\end{algorithmic}
\label{alg:greedy}
\end{algorithm}

By executing this procedure for \emph{each} posterior sample, we can distinguish between diagonal and off-diagonal contributions to the MMSE, with the diagonal one representing the error made from misalignment between matched patterns, and the off-diagonals representing the contributions from alignment between unmatched patterns. For all the rest of this section, every time a posterior sample $\bx$ appears, it is assumed that the signs and ordering of its columns has already been resolved according to the above greedy procedure as $\bx\leftarrow  (\pi\,\circ\,\sigma)(\bx)$, and the Monte Carlo approximations of expectations $\langle \, \cdot\,\rangle$ are taking into account this per-sample operation.

\begin{figure}[t!]
    {
        \includegraphics[width=\linewidth,trim={16cm 0 0 0},clip]{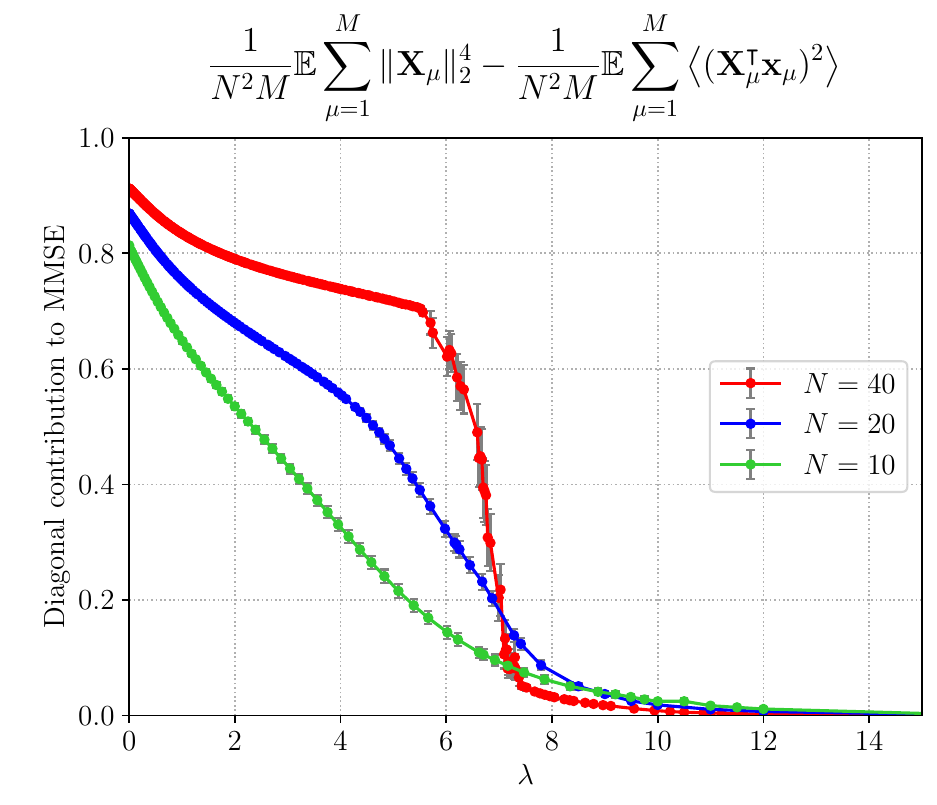}        \includegraphics[width=\linewidth,trim={0 0 16cm 0},clip]{./Plots/offdiag_alpha0.500000_re.pdf}
    }
\vspace{-0.6cm}
    \caption{Monte Carlo results for the diagonal (upper figure) and off-diagonal (lower one) contributions to the MMSE \eqref{Q-MMSE} after reordering and signs selection of the columns of each posterior sample according to Algorithm~\ref{alg:greedy}, with Rademacher prior, for various $N$ with $\alpha=0.5$. Error bars are standard errors of the mean.}
\label{fig:offDiag}
\end{figure}

\begin{figure}[t!]
    {
       \includegraphics[width=\linewidth,trim={0cm 0cm 0 2.5cm},clip]{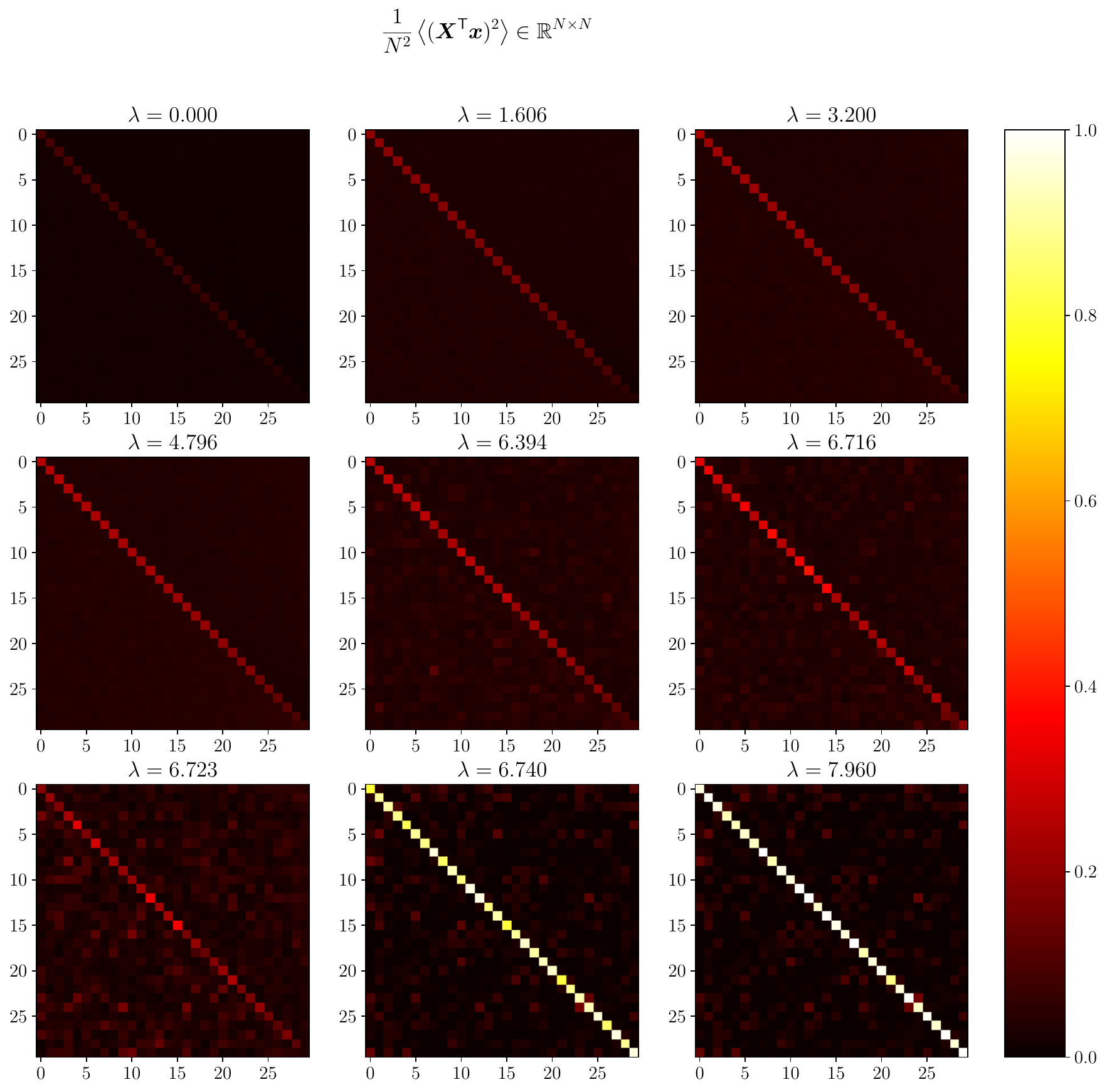}
    }

    \caption{Average squared overlap $N^{-2}\langle(\bX^\intercal\bx)^{\odot 2}\rangle$ (the square applies entrywise) for various SNRs obtained by Monte Carlo after reordering and signs selection of the columns of each posterior sample according to Algorithm~\ref{alg:greedy}, for a single instance $(\bX,\bZ)$ with Rademacher prior and $(N,M) = (60, 30)$.}
\label{fig:transOverlapSquare}
\end{figure}

\figurename~\ref{fig:offDiag} shows the diagonal and off-diagonal overlap contributions to the MMSE \eqref{Q-MMSE} for finite sizes computed by Monte Carlo. First, looking at the diagonal contribution for SNRs in the denoising phase ($\lambda\lesssim6$ for this size), the curve seems to approach $1=\E\sum_{\mu=1}^M \|\bX_\mu\|^4/(N^2M)$ as $N$ increases, indicating that no strong ``pattern matching'' is happening. Looking at the bottom part of \figurename~\ref{fig:offDiag} we observe that in this phase it is instead the off-diagonal of the overlap that contributes the most in decreasing the MMSE, e.g., in \figurename~\ref{fig:MCMCvsRIE}. The strong negative value of $(\EE \sum_{\mu\neq \nu} (\bX_\mu^\intercal\bX_\nu)^2-\EE \sum_{\mu\neq \nu} \langle (\bX_{\mu}^\intercal\bx_\nu)^2\rangle)/(N^2M)$ indicates that $(Q_{\mu\nu})_{\mu\neq \nu}$ fluctuate more around $0$ than in the case of i.i.d. patterns $(\bX_\mu^\intercal \bX_\nu/N)_{\mu\neq \nu}$. This is even more evident when plotting the histogram of the overlap entries in the denoising phase, see \figurename~\ref{fig:histo}. The broader distributions compared to the gray one indicate larger fluctuations than in the random case (where $\bX$ and $\bx$ are i.i.d.). Each $\bx_\mu$ has thus a small projection on most $(\bX_\nu)$ of order $O(1/\sqrt{N})$; yet, as $\lambda$ increases a large fraction of these projections become typically larger than between independent patterns (but of the same order). 

Beyond the transition, however, the diagonal contribution to the MMSE drops to zero, meaning that posterior patterns pair with ground truth patterns which are \emph{retrieved} in a synchronized way (lower part of \figurename~\ref{fig:transOverlapSquare}). Consistently, the off-diagonal contribution approaches $0$, as posterior patterns \enquote{orthogonalize} w.r.t. to the planted patterns, except for matched pairs that contribute to the reduction in the diagonal MMSE contribution. This shows that in the factorisation phase, non-trivial estimation of $\bX$ is possible, up to the unavoidable signs and permutation ambiguities of the columns. 

{Let us discuss a concrete algorithm to do so. Once the posterior samples' column permutation and sign symmetry are broken in the factorisation phase, their average is expected to possess a nontrivial alignment with $\bX$ up to the aforementioned ambiguity. 
The posterior samples can be synchronized efficiently by taking a single posterior sample as reference matrix (instead of $\bX$), and then aligning each remaining sample against this using again Algorithm~\ref{alg:greedy}, see Algorithm~\ref{alg:sync} below:
}
\begin{algorithm}[H]
\caption{Synchronized posterior average}
\begin{algorithmic}[1]
\Statex{\textbf{Input:} Posterior samples $\{\bx_t\}_{t = 1}^{N_{\rm MC}}$.}
\Statex \textbf{Output:} $\langle \bx \rangle_{\rm alg.}$
\Statex{\textbf{Initialize:} $\langle \bx \rangle_{\rm alg.} = \bx_1 / N_{\rm MC}$}
\For{$t = 2$ to $N_{\rm MC}$}
    \State{$\pi, \sigma \gets$ Algorithm 1 with Inputs $\bx_1$ and $\bx_t$}
    \State{$\langle \bx \rangle_{\rm alg.} \gets \langle \bx \rangle_{\rm alg.} + (\pi \circ \sigma)(\bx_t) / N_{\rm MC}$}
\EndFor
\end{algorithmic}
\label{alg:sync}
\end{algorithm}

{This procedure outputs an estimator $\langle \bx \rangle_{\rm alg.}$ of $\bX$ verifying criterion \eqref{optOverlap}. If in addition one is given an oracle permutation $\hat{\pi}(\bX)$ and signs $\hat{\sigma}(\bX)$, the remaining ambiguities can be broken by setting $\langle \bx \rangle_{\rm alg.} \gets (\hat{\pi} \circ \hat{\sigma}) (\langle \bx \rangle_{\rm alg.})$. \figurename~\ref{fig:X_MSE} depicts the recovery performance of $\langle \bx \rangle_{\rm alg.}$ from Algorithm~\ref{alg:sync} (in which the oracle permutation and signs are exploited for the sake of illustration). As can be seen clearly from both the Frobenius norm and trace overlap between $\langle \bx \rangle_{\rm alg.}$ and $\bX$, information on $\bX$ itself can be retrieved in the factorisation phase. Notice that Algorithm~\ref{alg:sync} requires as input posterior samples, which take a time exponential in $N$ to sample by Monte Carlo beforehand. Thus we do not claim that factorisation can be done efficiently, but rather that it is statistically possible.}

\begin{figure}
    \centering
    \includegraphics[width=\linewidth]{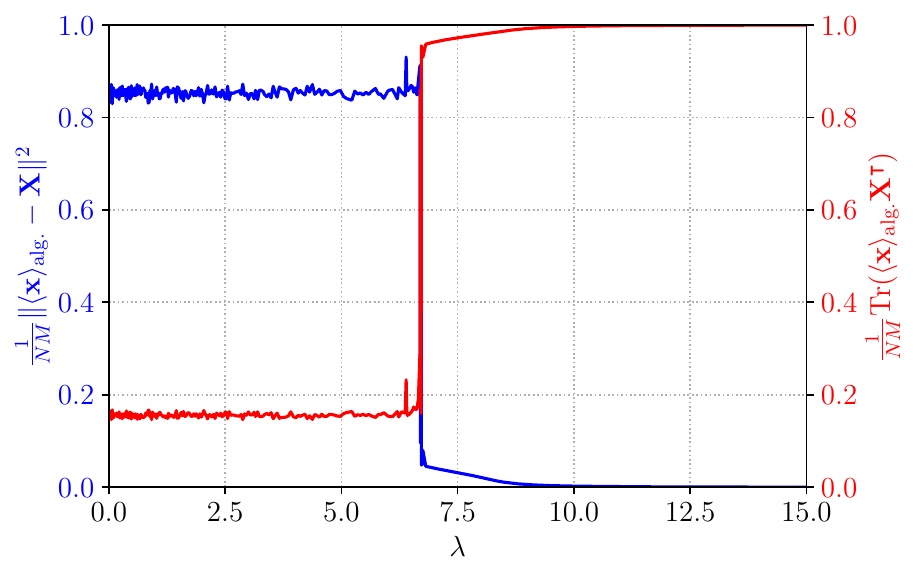}
    \caption{MSE (blue line, left axis) and overlap with the ground truth $\bX$ (red, right) for estimator $\langle \bx \rangle_{\rm alg.}$ from Algorithm~\ref{alg:sync} with $5.10^5$ Monte Carlo samples, for $(N, M) = (60, 30)$ and a single data instance. 
}
    \label{fig:X_MSE}
\end{figure}


Whether a non-trivial estimation of $\bX$ is possible in the denoising phase is less clear. 
The averaged entrywise squared overlap is displayed in \figurename~\ref{fig:transOverlapSquare} for various SNRs, and in particular close to the crossover region (three plots at the bottom). The homogeneous background and weaker diagonal before the transition (upper part) signals that the information about the ground truth patterns is somehow diluted in all the posterior patterns, consistently with the bottom panel of \figurename~\ref{fig:offDiag}. The diagonal then gets suddenly stronger and the background much smaller around the transition, indicating again a strong pattern-to-pattern matching between posterior samples and the signal in the factorisation phase. Overall these figures illustrate a clear change of behavior between the two phases, where the quasi basis of posterior patterns $(\bx_\mu)$ (i.e., a set of $N$-dimensional vectors whose pairwise inner products are approximately $\sqrt{N}$) \enquote{aligns} much more with the quasi basis of planted patterns $(\bX_\mu)$ beyond the transition.

We have made two empirical observations that may suggest at first sight that factorisation is possible in the denoising phase, too: $i)$ the diagonal contribution to the MMSE (upper part of \figurename~\ref{fig:offDiag}) seems bounded away from one, and $ii)$ the diagonals displayed in \figurename~\ref{fig:transOverlapSquare} seem $O(1)$ in the denoising phase. The question we thus aim at elucidating is whether these are artefacts of finite size effects combined with the column reordering procedure applied to the posterior samples (which necessarily amplifies the overlap diagonal) or not. In other words, \emph{is the overlap $|\bx_\mu\cdot \bX_\mu|/N$ between matched pairs vanishing or not in the thermodynamic limit in the denoising phase (still assuming the permutation ambiguity is resolved)?}

In order to answer this question, we can heuristically estimate finite size corrections, helping ourselves with the histograms at the top panel in \figurename~\ref{fig:histo}. Similar histograms were obtained from overlap matrices to which we removed the diagonal for a fair comparison (even if the histogram almost does not change if keeping the diagonal, which is not true anymore beyond the transition as clear form the bottom panel in \figurename~\ref{fig:histo}). The diagonal is removed since it has different statistical properties than the off-diagonal entries, and we need to understand the statistics of the complement of the diagonal in order to assess if such random variables can induce the observed diagonal just through the reordering procedure. From \figurename~\ref{fig:histo} we deduce that the marginal law of off-diagonal elements rescaled by $\sqrt{N}$ behave as centered Gaussian variables with variance $\sigma^2(\lambda)$ tuned by the SNR, evaluated numerically from the histograms. We thus make the simplifying assumption that the correlations among the $(Q_{\mu\nu})_{\mu,\nu\leq M}$ are sufficiently weak to consider them independent. This allows us to estimate the average of the diagonal of the entrywise squared overlap matrix under the null hypothesis that it is just due to the reordering. 

According to the aforementioned independence and Gaussianity hypotheses on the overlap entries, the law of the $\mu$th diagonal squared overlap element in \figurename~\ref{fig:transOverlapSquare} can be modeled as that of
\begin{align*}
    \chi_\mu:=\max_{\mu\leq\nu\leq M} Q_{\mu\nu}^2.
\end{align*} 
The cumulative distribution function of $\chi_\mu$ is explicit:
\begin{align*}
    \mathbb{P}(\chi\leq t)=\mathbb{P}(|Q_{11}|\leq \sqrt{t})^{M_\mu}=\text{erf}
    \Big(\sqrt{\frac{Nt}{2\sigma^2(\lambda)}}\Big)^{M_\mu}
\end{align*}for any $t\geq0$, where $M_\mu:=M-\mu+1$, and $\text{erf}(z):=\frac2{\sqrt{\pi}}\int_0^ze^{-x^2}dx$. Furthermore, $\chi_\mu$ being a positive random variable, the expectation of any function $f$ of it can be computed via tail integration as
\begin{align*}
    \EE f(\chi_\mu)&=f(0)+\int_0^\infty f'(t)\mathbb{P}(\chi_\mu\geq t)\,dt\nonumber\\
    &=f(0)+\int_0^\infty f'(t)\Big(1-\text{erf}
    \Big(\sqrt{\frac{Nt}{2\sigma^2(\lambda)}}\Big)^{M_\mu}\Big)\,dt.
\end{align*}We are interested in evaluating the mean and order of the fluctuations of the average of the diagonal elements
\begin{align*}
    \bar\chi:=\frac{1}{M}\sum_{\mu=1}^M\chi_\mu.
\end{align*}
Following our assumptions, the $\chi_\mu$'s are independent random variables so the variance $\mathbb{V}(\bar\chi)=\sum_{\mu=1}^M\mathbb{V}(\chi_\mu)/{M^2}$. Therefore, if the empirical dominant diagonal in \figurename~\ref{fig:transOverlapSquare} was just an artefact of the reordering, its average value should be compatible with
\begin{align*}
    \bar\chi\approx \EE\,\bar\chi\pm\frac{2}{M}\sqrt{\sum_{\mu=1}^M\mathbb{V}(\chi_\mu)},
\end{align*}
considering two standard deviations as confidence interval.
Recall that both values on the r.h.s. depend on the SNR trough $\sigma^2(\lambda)$.

\begin{figure}[t!]
{
       \includegraphics[width=\linewidth,trim={0cm 0.7cm 0cm 0cm},clip]{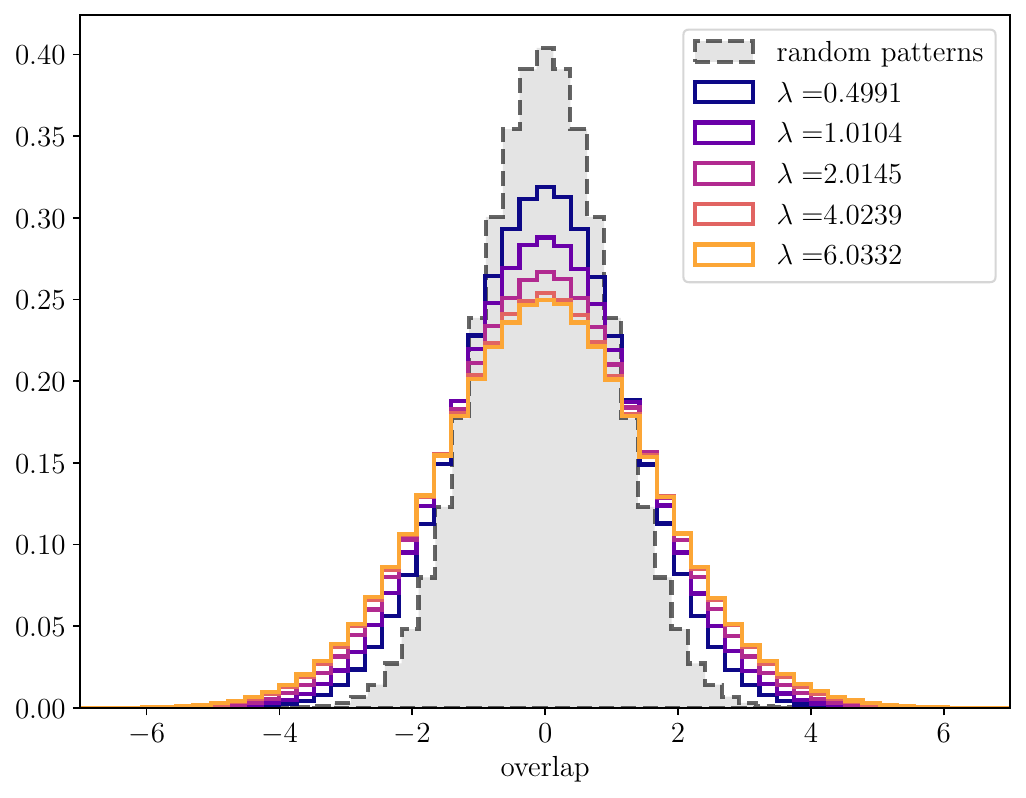}
       \includegraphics[width=\linewidth]{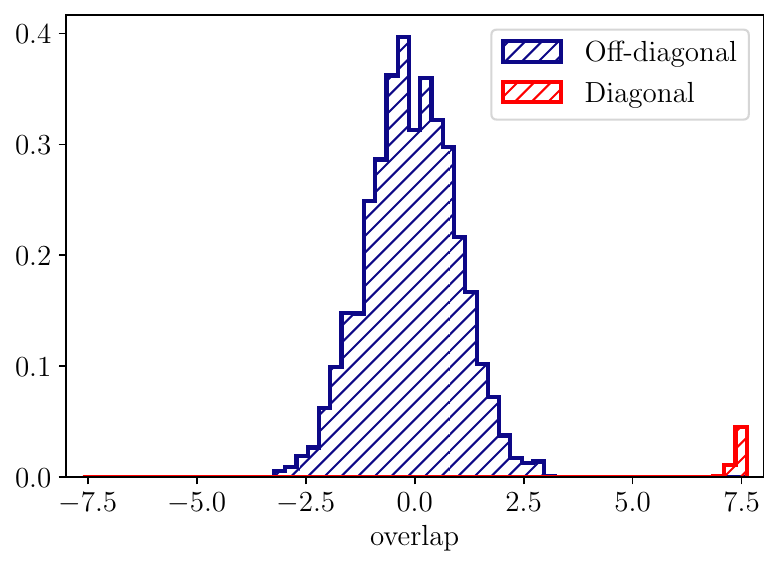}
    }
    \vspace{-0.6cm}
    \caption{Upper figure: Histogram of the entries of $\sqrt{N}\bQ$ in the denoising phase for a single instance of $(\bX,\bZ)$ with size $(N,M) = (60, 30)$ and Rademacher prior, obtained from $5\cdot 10^3$ Monte Carlo samples after reordering of each sample columns. All entries are accumulated into each histogram, resulting in a sample size of $5\cdot 10^3 M^2$ each. The gray histogram indicates the distribution of the overlap between two random Rademacher patterns of size $N$, which approaches a standard normal distribution as $N\to +\infty$. Lower~figure: Histogram of $\sqrt{N}\bQ$ for $\lambda=7.96$ which lies at the start of the factorisation phase given that size, see \figurename~\ref{fig:MMSE_th}. In that case, outliers emerge and so we have separated the diagonal and off-diagonal contributions. The diagonal one is peaked around $\sqrt{N}\approx 7.7$. The off-diagonal histogram resembles again a normal distribution, and its variance is very close to $1$.}
\label{fig:histo}
\end{figure}

\begin{figure}[t!]
    \centering
    \includegraphics[width=\linewidth]{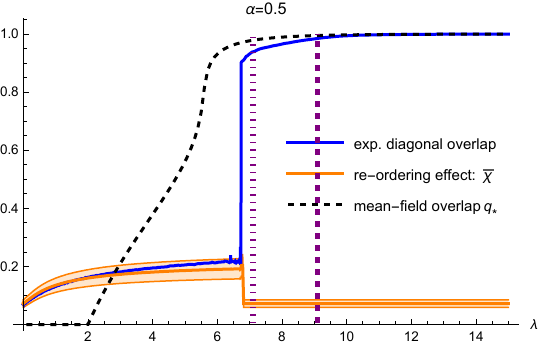}
    \caption{Experimental average diagonal of the squared overlap as a function of $\lambda$ obtained via Monte Carlo after reordering of each sample columns for a single instance $(\bX,\bZ)$ with Rademacher prior and $(N,M) = (60, 30)$ (blue), as in \figurename~\ref{fig:transOverlapSquare}. It is compared to the effect of the re-ordering of the columns of the posterior samples estimated by $\EE \bar\chi$ (orange curve), with two standard deviations confidence interval (band). We also display the equilibrium diagonal overlap $q_*$ given by the mean-field theory \eqref{eq:Kabashima_r}. The complete ansatz \eqref{conjecture} predicts a transition at $\lambda_c(\alpha=0.5)\approx 9.1$ (vertical purple dashed curve), which can be corrected for finite size effects yielding $\approx 7.1$, see caption of \figurename~\ref{fig:MI_vs_SNR_differentalpha}. }
    \label{fig:reordering_finitesize}
\end{figure}

\figurename~\ref{fig:reordering_finitesize} compares the value of $\bar\chi$ with its uncertainty (orange) to the arithmetic mean of the diagonal elements of the averaged squared overlap $N^{-2}\langle(\bX^\intercal\bx)^{\odot 2}\rangle$ (blue). There is an evident agreement in the denoising phase: the curve obtained via MCMC never detaches from $\bar\chi$ by more than two standard deviations. This suggests that our assumptions on the overlap elements $Q_{\mu\nu}$'s are effective, and above all that the diagonals in \figurename~\ref{fig:transOverlapSquare} for $\lambda\leq 6.723$, i.e., in the denoising phase for that size, are compatible with the diagonal-amplifying effect of the column reordering procedure. This entails that the posterior patterns $(\bx_\mu)$ are delocalized in the quasi basis of the ground truth patterns $(\bX_\mu)$. On the contrary, after the transition the values of $\bar\chi$ are incompatible (further than 5 standard deviations) with the empirical curve, signaling that the values in the diagonals of the bottom center and right panels in \figurename~\ref{fig:transOverlapSquare} are not only due to reordering, but a real effect happening in the factorisation phase. Note that the sudden decrease in $\EE \bar\chi$ in \figurename~\ref{fig:reordering_finitesize}  is related to the collapse of the blue histogram at the bottom of \figurename~\ref{fig:histo} onto a variance one normal law right after the transition, while it was more spreaded just before it, i.e., it is a manifestation of the transition through $\sigma^2(\lambda)$. As last comment, we can show that the effect $\EE \bar \chi$ of the reordering vanishes as $O(2\sigma^2(\lambda)\ln(2M)/N)$ (as it is a mean of maxima of i.i.d.\ Gaussians, see \cite{Embrechts2013ExtremeValueTheoryFinanceInsurance}). 

In conclusion, from the evidence we gathered we expect that the seemingly bright diagonals in \figurename~\ref{fig:transOverlapSquare} for $\lambda\leq 6.723$ disappear in the thermodynamic limit. Consequently, in the denoising phase it is more plausible that factorisation is \emph{not} possible, in contrast with the factorisation phase, and only denoising of $\bX\bX^\intercal$ is. This aligns with the fact that one cannot extract an estimator of $\bX$ from the RIE, which we think is Bayes-optimal in that phase (or at least we do not see how to do so).

{Note that \figurename~\ref{fig:offDiag}, \ref{fig:transOverlapSquare}, \ref{fig:X_MSE}, \ref{fig:reordering_finitesize} once more suggest that the phase transition is discontinuous for discrete~priors.}

\begin{figure}[t!]
       \includegraphics[width=\linewidth]{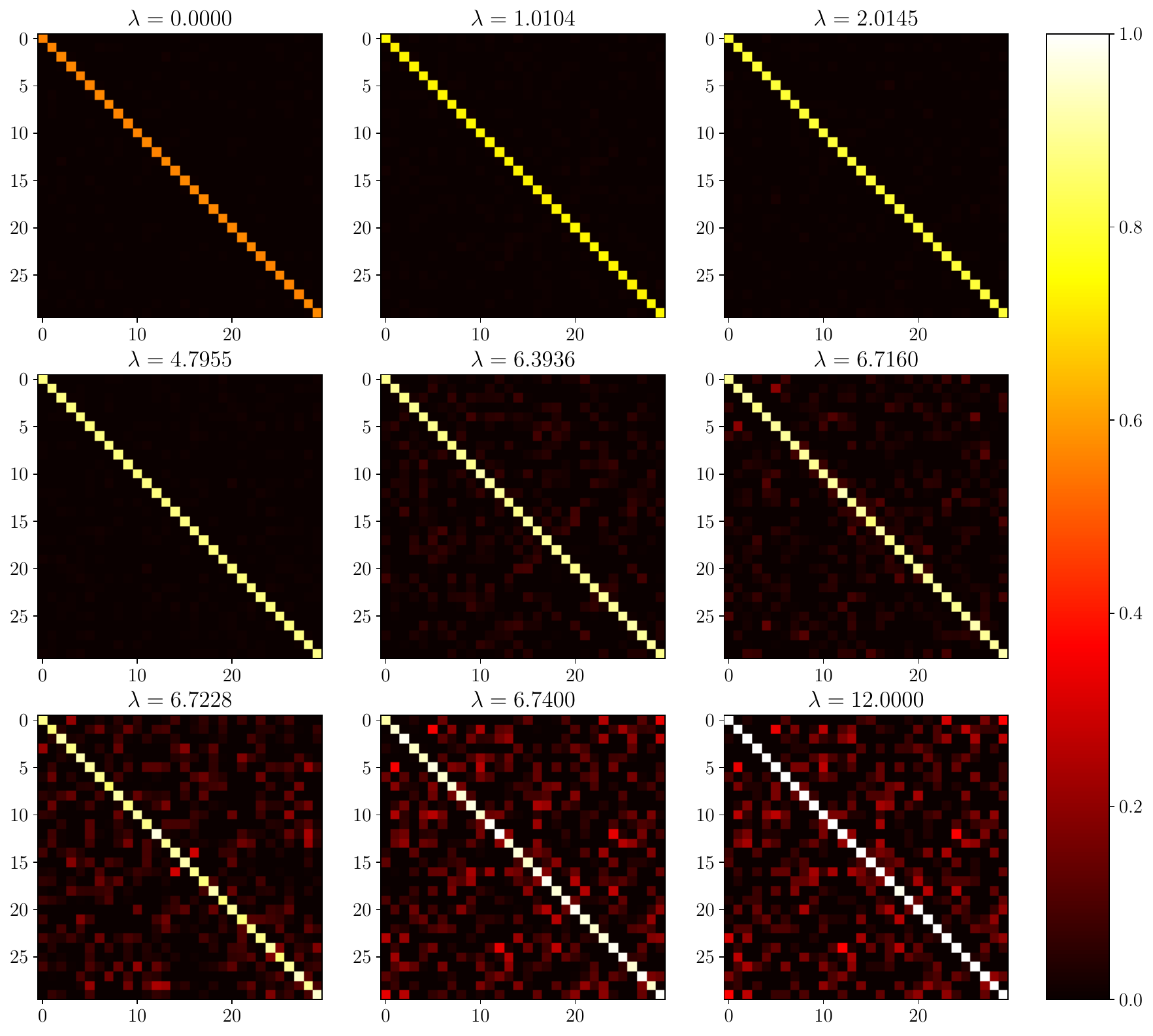}
    \caption{Average overlap $N^{-1}  \bX^\intercal \langle\bx \bO_*(\bx)\rangle$ computed by Monte Carlo after reordering of each sample columns for a single instance of $(\bX,\bZ)$ with Rademacher prior and size $(N,M) = (60, 30)$. }
\label{fig:OptimalRot}
\end{figure}

We have argued that non-trivial estimation of $\bX$ is impossible in the denoising phase. In order to clarify what is inferred about it that allows denoising of $\bX\bX^\intercal$, consider the following rotation for a given posterior sample:
\begin{equation*}
    \bO_*(\bx) := \mathop{\rm argmin}_{\bO \in \R^{M \times M}\,:\,  \bO^\intercal \bO = I_M} \| \bX - \bx \bO\|^2.
\end{equation*}
In \figurename~\ref{fig:OptimalRot} we plot the overlap matrix between \enquote{optimally rotated} posterior samples and the ground truth, averaged over Monte Carlo samples, i.e., $\bX^\intercal \langle \bx \bO_*(\bx)\rangle /N$. Clearly, in the denoising phase the profile of this matrix is diagonal, but with the diagonal now remaining order one as $N$ increases, and whose intensity grows with $\lambda$. Note that there is a finite diagonal overlap even when $\lambda = 0$, as two random $M$-dimensional subspaces of $\R^N$ always have a common subspace of dimension $O(N)$. This demonstrates two points: on one hand the denoising phase is the regime where the ground truth factor is estimated but only up to an unknown rotation of the basis in which the patterns $(\bX_\mu)$ are expressed. On the other hand, this also reveals that it is not sufficient to break the permutation symmetry of the columns in order to see a non-trivial overlap appear (as was done in \figurename~\ref{fig:transOverlapSquare}). One instead needs to break the much larger symmetry group of rotations, which further requires an oracle already knowing the planted matrix. This \enquote{effective rotational invariance} is what we believe to be the source of universality in the denoising phase, and what prevents factorisation. We refer to \cite{semerjian2025some} for a related discussion.

As last comment, solving the saddle point equations \eqref{eq:Kabashima_r}, we get that the mean-field overlap $q_*$ detaches from $0$ around $\lambda=2$ in the denoising phase, see \figurename~\ref{fig:reordering_finitesize}. {We have verified that this has an effect on the multiscale curves in \figurename~\ref{fig:MMSE_th} and \ref{fig:MI_vs_SNR_differentalpha}, but it is just too small to be visible by eye.} This may seem in contradiction with our negative conclusions concerning factorisation in this phase. However, we claim that this is an artefact of the mean-field theory being applied in the \enquote{matrix model/denoising regime} where the necessary concentration-of-measure effects for it to hold are not present. As discussed in Section~\ref{sec:completeansatz}, mean-field equations are reliable only after the transition $\lambda_c(\alpha=0.5)\approx 9.1$, as manifest from the figure. We therefore do not interpret this detachment from $0$ as a true transition nor as a signal of possible factorisation in this phase.

\section{Conclusion and perspectives} 
In this paper we have first conducted numerical experiments to grasp features of the phase diagram of matrix denoising in the challenging regime of extensive-rank beyond the rotational invariant setting. We have clarified the rich phenomenology of the model. In particular, we uncovered a first order phase transition separating a regime akin to a matrix model where strong universality properties hold and the RIE is a Bayes-optimal denoising algorithm, from a fundamentally different phase where it is not anymore due to universality breaking, and which is more amenable to statistical physics techniques. The free entropy \eqref{eq:free_entropy} seen as a function of the SNR $\lambda$ thus interpolates between a matrix model solvable by the HCIZ integral (the denoising phase) and a planted mean-field spin model (the factorisation phase).

We have then proposed a mean-field theory based on a combination of spin glass approaches to compute the mutual information and minimum mean-square error. Interestingly, it yields the same equations as the replica approach of \cite{SK13EPL,SK13ISIT,KMZ_DL-2013,Marc-Kabashima}. The latter was believed to be incorrect for a long time. Through our independent analysis, we have instead clarified in which regime it is reliable --the factorisation phase-- or not --the denoising phase-- and have shown how to correct it when not accurate thanks to universality, as conjectured by Semerjian \cite{semerjian2024matrix} but for the whole phase diagram. Our complete ansatz for the two phases is in good quantitative agreement with the experiments, up to unusually large finite-size corrections which we were able to quantify. 

The numerical evidence suggests that optimal denoising is algorithmically hard only beyond the transition point, due to the presence of the universal RIE state at all signal-to-noise ratios. We have experimentally confirmed that this state prevents Monte Carlo algorithms to outperform the RIE. 

There are many directions to pursue from here. Concerning our present setting, we have proven that beyond a certain SNR universality cannot hold. It would be interesting to rigorously confirm that the universality indeed holds at low SNR and breaks precisely at the predicted transition $\lambda_c(\alpha)$, and that the mean-field theory does become asymptotically exact beyond that point. 

Concerning generalisations, the most natural one is the non-symmetric version of the problem, with more general noises, as considered in \cite{SK13EPL,SK13ISIT,KMZ_DL-2013,Marc-Kabashima}. 

In addition, our setting shares several similarities with the Hopfield model \cite{Hopfield82} of associative memory. The phenomenology of the latter is well understood (at least at the replica symmetric level) thanks to the seminal work \cite{Amit1}. The key idea is that samples from the Boltzmann-Gibbs measure cannot have extensive projections on all memorized patterns, but only on a finite subset and typically $O(1/\sqrt{N})$ projections onto the others. These infinitesimal contributions, being extensively many, give rise to a contribution that can be treated as effective noise called \emph{pattern interference}. Its intensity depends on the load of the memory $M/N\approx \alpha$. There exists a retrieval phase for low $\alpha$ and temperatures, where an algorithm sampling from the Boltzmann-Gibbs distribution with informative initialisation \enquote{recalls} a pattern. Out of this region, the samples have instead $O(1/\sqrt{N})$ overlap with all the patterns.

We observe that for matrix denoising too, in the denoising phase, the overlap matrix presents $O(1/\sqrt{N})$ entries. However, instead of treating these elements as noise as done in Hopfield, here they actually give a fundamental contribution to decreasing the MMSE as illustrated in \figurename~\ref{fig:offDiag}. This is one of the main differences with the Hopfield model: the sum of spurious alignments of the posterior sample columns $\bx_\mu$ onto other patterns of the ground truth is exploitable as information when denoising $\bX\bX^\intercal$, and it is not a disturbance. The term $\Tr(\bx\bx^\intercal)^2$ in the Hamiltonian \eqref{eq:Hamiltonian} contributes to this phenomenology, generating a caging effect that prevents posterior patterns from overlapping and thus helps.
As a consequence, we find that the ``load'' $\alpha$ that the memory can withstand is higher. Our theory supports the fact that it is possible to retrieve \emph{all} the patterns even for $\alpha$'s sensibly larger than the Hopfield critical load $\alpha_c\approx0.138$ for binary patterns, provided an informative initialisation is available and the signal-to-noise ratio is sufficiently high.

Furthermore, our theory and its agreement with experiments in the factorisation phase indicates that this caging effect becomes strong enough to almost completely cleanse the system from pattern interference, yielding an MMSE akin to that of several quasi decoupled rank-1 problems.

Another major difference is that in matrix denoising one attempts to recall \emph{all} patterns (columns of $\bX$) at once (see the recent related work \cite{agliari2024generalized} where multiple, but finitely many, patterns are jointly recovered), differently from the approach of \cite{camilli2023matrix} for instance.

{To conclude this analogy, the study of the memorization properties of this model should be carried out in a non Bayes-optimal setting. One issue is that the universal branch of the matrix denoising problem is known exactly only thanks to the Nishimori identities which allow to simplify the HCIZ integral \cite{perturbative_Maillard21,Jean_Farzad_matrixinference_2024}. Bayes-optimality also prevents replica symmetry breaking \cite{mezard1990spin}. These simplifying symmetries break down in the absence of Bayes-optimality, and solving the model without relying on them could be much more challenging. }

{A further connection that is drawing attention is that with information theoretical limits of neural networks. It was recently pinpointed by some of the present authors that a theory similar to ours can be used to describe learning transitions in shallow neural networks. More specifically, in \cite{maillard2024bayes} the authors exploited a mapping between the learning of an extensive-width neural network with quadratic activation function in a teacher-student setting and a matrix denoising problem, where the hidden matrix follows a Wishart distribution (and thus is \emph{exactly} rotationally invariant). Based on the ideas of the present paper, \cite{2025_nn_quadratic_Barbier} goes beyond by tackling any activation function but also possibly non-rotationally invariant weight matrices. Given the similar nature of the problem treated in \cite{2025_nn_quadratic_Barbier}, the remarkable agreement between theoretical predictions and numerical experiments in the said paper also supports our claims of exactness of the theory proposed here. Hence, our theory (together with the one of \cite{SK13EPL,SK13ISIT,KMZ_DL-2013,Marc-Kabashima}) and in particular the rich phenomenological picture we have identified has the potential to become valuable tools for describing a wide variety of matrix models going much beyond matrix denoising.}

Another stretched, yet potentially fruitful comparison, is that with Anderson localization \cite{anderson1958absence}. The phenomenon is different as being quantum by nature and is not a thermodynamic phase transition. Yet, the similarity is striking due to both transitions being linked to localization, i.e., macroscopic pairwise alignment of two bases of vectors: the eigenvectors of the energy operator of the quantum system at hand localize in lattice basis when tuning the amount of disorder, while in the present problem the quasi basis of posterior patterns $(\bx_\mu)$ localizes at the transition in the quasi basis of planted ones $(\bX_\mu)$, which manifests itself by the abrupt increase of the overlap diagonal in \figurename~\ref{fig:transOverlapSquare}. It would be interesting to explore how far this analogy can be pushed.

A final direction to pursue would be to compare known Bayesian algorithms for matrix inference to our results, such as BiGAMP \cite{PSCvI}, the AMP algorithm of \cite{Marc-Kabashima} or Unitary AMP \cite{yuan2024unitary}. Given our results on the hardness of denoising using Monte Carlo, we think that these algorithms adapted to the present setting may, too, face a hard phase and be outperformed by the RIE without additional information on $\bX$. We plan to look into this in a future work.

\section*{Acknowledgments}

J.B. and F.C. were funded by the European Union (ERC, CHORAL, project number 101039794). Views and opinions expressed are however those of the authors only and do not necessarily reflect those of the European Union or the European Research Council. Neither the European Union nor the granting authority can be held responsible for them. J.K. was funded by the Natural Sciences and Engineering Research Council of Canada (RGPIN-2020-04597). K.O. was funded by JSPS KAKENHI Grant number 22KJ1074, JST CREST Grant number JPMJCR1912, and Grant-in-Aid for Transformative Research Areas (A), “Foundation of Machine Learning Physics” (22H05117). Computational resources of the AI Bridging Cloud Infrastructure (ABCI) provided by the National Institute of Advanced Industrial Science and Technology (AIST) were used. Part of the work of K.O. was done while visiting the International Center for Theoretical Physics in Trieste.

\normalsize
\bibliography{total_biblio}

\appendix

\onecolumngrid
\section{Irrelevance of the diagonal of the data}\label{appendix:diagonals}
\begin{proposition}[Information irrelevance of the data diagonal components]
    The diagonal part of the data $(Y_{ii})_{i\leq N}$ does not contribute to the MI in the high dimensional limit. Specifically, the inference problem
    \begin{align}\label{eq:model-wout-diagonal}
        \tilde Y_{ij}=\sqrt{\frac{\lambda}{N}}\sum_{\mu=1}^MX_{i\mu} X_{j\mu}+\tilde Z_{ij},\quad \tilde Z_{ij}\iid\mathcal{N}(0,1),\quad 1\leq i<j\leq N
    \end{align}has the same asymptotic mutual information density between data and signal as \eqref{eq:channel0}.
\end{proposition}
\begin{proof}
    The proof is based on standard interpolation. The interpolating inference problem is
    \begin{align*}
        Y^t_{ij}=\begin{cases}
            \tilde Y_{ij} &i<j\\
            \sqrt{\frac{t\lambda}{N}}\sum_{\mu=1}^MX_{i\mu}^2 +Z_{ii} &i=j
        \end{cases}
    \end{align*}
    where $Z_{ii}\iid\mathcal{N}(0,2)$.
    The MI density for this problem is
    \begin{align*}
        \frac{1}{MN}I(\bY^t;\bX)&=\frac{\lambda}{4MN^2}\sum_{i\neq j}\sum_{\mu,\nu=1}^M\EE X_{i\mu}X_{j\mu}X_{i\nu}X_{j\nu}+\frac{t\lambda}{4MN^2}\sum_{i=1}^N\sum_{\mu,\nu=1}^M\EE X_{i\mu}^2X^2_{i\nu}\nonumber\\
        &\quad -\frac{1}{MN}\EE\ln \int dP_X(\bx)\exp\sum_{i<j}\Big[\sqrt{\frac{\lambda}{N}}
        \tilde Y_{ij}\bx_i\cdot\bx_j-\frac{\lambda}{2N}(\bx_i\cdot\bx_j)^2
        \Big]\nonumber\\
        &\quad \times\exp\frac{1}{2}\sum_{i=1}^N\Big[\sqrt{\frac{t\lambda}{N}}
         Y^t_{ii}\|\bx_i\|^2-\frac{t\lambda}{2N}\|\bx_i\|^4\Big]  \nonumber  \\       
         &=\frac{\lambda(1+
    \alpha t)}{4}-\frac{1}{MN}\EE\ln \int dP_X(\bx)\exp\sum_{i<j}\Big[\sqrt{\frac{\lambda}{N}}
    \tilde Y_{ij}\bx_i\cdot\bx_j-\frac{\lambda}{2N}(\bx_i\cdot\bx_j)^2
    \Big]\nonumber\\
    &\quad \times\exp\frac{1}{2}\sum_{i=1}^N\Big[\sqrt{\frac{t\lambda}{N}}
     Y^t_{ii}\|\bx_i\|^2-\frac{t\lambda}{2N}\|\bx_i\|^4\Big]+O(N^{-1}).
    \end{align*}
    The last error term is uniform in $t\in[0,1]$.
    The $t$-derivative of this interpolating MI thus reads
    \begin{align*}
    \frac{d}{dt}\frac{1}{MN}I(\bX;\bY^t)&=\frac{\lambda\alpha}{4}-\frac{1}{MN}\sum_{i=1}^N\EE\Big\langle
    \frac{1}{4}\sqrt{\frac{\lambda}{tN}}Z_{ii}\|\bx_i\|^2+\lambda \frac{\|\bX_i\|^2\|\bx_i\|^2}{2N}-\frac{\lambda}{4N}\|\bx_i\|^4
    \Big\rangle_t+O(N^{-1})
\end{align*}Now we can integrate the $Z_{ii}$'s by parts, taking into account that their variance is $2$:
\begin{align*}
    \frac{d}{dt}\frac{1}{MN}I(\bX;\bY^t)&=\frac{\lambda\alpha}{4}-\frac{1}{MN}\sum_{i=1}^N\EE\Big\langle
    \frac{\lambda}{2N}\|\bx_i^{(1)}\|^4-\frac{\lambda}{2N}\|\bx_i^{(1)}\|^2\|\bx^{(2)}_i\|^2+\lambda \frac{\|\bX_i\|^2\|\bx_i^{(1)}\|^2}{2N}-\frac{\lambda}{4N}\|\bx_i^{(1)}\|^4
    \Big\rangle_t +O(N^{-1})\nonumber\\
    &=\frac{\lambda\alpha}{4}-\frac{\lambda}{4MN^2}\sum_{i=1}^N\EE\|\bX\|^4+O(N^{-1})=O(N^{-1}),
\end{align*}
where in the last step we also used the Nishimori identity
\begin{align*}
    \EE\langle\|\bx^{(1)}_i\|^2\|\bx^{(2)}_i\|^2\rangle_t=\EE\langle\|\bX_i\|^2\|\bx^{(1)}_i\|^2\rangle_t.
\end{align*} and superscripts $\bx^{(k)}$ denote replica indices, namely conditionally (on $\bY$) independent samples from the posterior measure. The error term is again uniform in $t$ so the proof is complete by noticing that $\bY^{t=1}$ are the observations in \eqref{eq:channel0} and $\bY^{t=0}$ is \eqref{eq:model-wout-diagonal}.
\end{proof}

The above proposition holds also if $\alpha$ is vanishing in the high-dimensional limit, i.e. for low rank settings.

\section{Details of the Monte Carlo procedure}\label{app:MCMC}

In this appendix, we elaborate on the numerical procedures necessary to evaluate the MI and MMSE for finite size systems. The target distribution is the posterior measure given by 
\begin{equation*}
    P_\lambda(\bx ; \bX, \bZ ) = \frac{P_X(\bx) e^{-H_\lambda(\bx; \bX, \bZ )}}{\mathcal{Z}_\lambda (\bX, \bZ)} , \quad -H_\lambda(\bx ; \bX, \bZ) = \sqrt{\frac{\lambda}{N}}\sum_{i<j,1}^N( \bX_i \cdot \bX_j + Z_{ij} )\bx_i\cdot\bx_j-\frac{\lambda}{2N}\sum_{i<j,1}^N(\bx_i\cdot\bx_j)^2,
\end{equation*}
where $\mathcal{Z}_\lambda$ is the partition function for SNR $\lambda$. Here, all expressions are explicitly rewritten using the quenched random variables $\bX, \bZ$, and the subscript for $N,M$ are dropped for brevity.

Evaluating the MI numerically requires one to calculate $\ln  \mathcal{Z}_\lambda (\bX, \bZ)$ with high precision, which is difficult to perform with standard MCMC (Markov Chain Monte Carlo) samplers. Therefore, we use bridge sampling \cite{ogata1990monte, Jerrum93PolyIsing}, more popularly known as \textit{Annealed Importance Sampling} \cite{neal2001annealed}. Given a strictly increasing sequence of SNRs $0 =\lambda_0 < \lambda_1 < \cdots < \lambda_R$, the logarithm of the partition function at $\lambda_R$ can alternatively be expressed as a telescopic sum:
\begin{equation}\label{eq:telescope}
    \ln  \mathcal{Z}_{\lambda_R}(\bX, \bZ) = \ln  \prod_{r = 0}^{R-1} \frac{\mathcal{Z}_{\lambda_{r+1}}(\bX, \bZ)}{\mathcal{Z}_{\lambda_r}(\bX, \bZ)} = \sum_{r = 0}^{R-1} \ln  \Big\langle  \exp \Big[ H_{\lambda_r}(\bx ; \bX , \bZ) - H_{\lambda_{r+1}}(\bx; \bX, \bZ) \Big] \Big\rangle_{\lambda_{r}}, 
\end{equation}
where $\langle\, \cdot \,\rangle_{\lambda}$ is the expectation with respect to $P_{\lambda}(\bx ; \bX, \bZ)$. By taking $\lambda_{r+1}$ and $\lambda_r$ sufficiently close, and consequently taking $ H_{\lambda_r}(\bx ; \bX , \bZ) - H_{\lambda_{r+1}}(\bx; \bX, \bZ)$ close to 0, the value in the brackets would yield low variance, allowing one to  estimate its mean with high precision using samples generated from $P_{\lambda}(\bx ; \bX, \bZ)$. 
To sample from a sequence of posterior measures with SNRs $\{\lambda_{r}\}_{r=0}^{R}$ in an efficient manner, we employ parallel tempering, or the replica exchange Monte Carlo method \cite{Hukushima96REMC, Marinari92SimulatedTempering}. 
This approach involves simultaneous execution of ``local'' MCMC samplers specifically simulating each of the $R+1$ systems. In addition, configurations between adjacent SNRs are occasionally exchanged while satisfying the detailed balance condition for fast mixing time. 
More explicitly, given an adjacent pair of SNRs $\lambda$ and  $\lambda^\prime$, with its current configuration of their corresponding sampler being given by $\bx$ and $\bx^\prime$, respectively, the exchange probability of the two configurations can be given using Metropolis' method: 
\begin{equation*}
    \min \Big\{ 1, \exp\Big[ H_{\lambda^\prime} (\bx^\prime; \bX, \bZ) - H_{\lambda^\prime}(\bx; \bX, \bZ) + H_\lambda (\bx; \bX, \bZ) - H_\lambda (\bx^\prime; \bX, \bZ) \Big] \Big\}.
\end{equation*}
In our experiments, we use a single spin-flip Metropolis-Hastings algorithm as the local MCMC sampler for each SNR. 
A single MCMC step is defined as $30NM$ local spin flip moves on all $R+1$ systems, followed by a replica exchange move between all adjacent pairs of SNR. 
Before each replica exchange move, the configuration of $\bx \bx^\intercal$ of each MCMC chain is recorded. After the Monte Carlo procedure has finished, the thermal average appearing in the expression for the log-partition function \eqref{eq:telescope} needed for the MI  and in the MMSE \eqref{MMSEuseful} is obtained by the empirical average over these samples. 
This whole Monte Carlo procedure is executed for several instances of $(\bX, \bZ)$ to then evaluate the average over the quenched randomness. 

\begin{table}[t]
\begin{minipage}[t]{.65\textwidth}
    \centering
    \begin{tabular}{|c||c| c | c | c | c | c | }
    \hline
    \begin{tabular}{@{}c@{}}Experiment \\ No. \end{tabular} & $(N,M)$ & $\quad \alpha \quad $ & \begin{tabular}{@{}c@{}}burnin \\ MCMC steps \end{tabular}& \begin{tabular}{@{}c@{}}sampling  \\ MCMC steps \end{tabular} & $R+1$ & \begin{tabular}{@{}c@{}}No. of random  \\ instances\end{tabular}
       \\ \hline \hline
        (1) &  (10, 5) & 0.5 & 2,000 & 8,000 & 72 & 128 \\
        (2) & (20, 10) & 0.5 & 10,000 & 40,000 & 135 & 128 \\
        (3) & (40, 20) & 0.5 & 10,000& 40,000 & 261 & 32 \\
        \hline 
        (4)$^\star$ & (20, 10) & 0.5& 10,000 & 40,000 & 29 & 128 \\
        (5)$^\star$ & (40, 20) &0.5& 10,000 & 40,000 & 29 & 128 \\
        (6)$^\star$ & (60, 30) & 0.5& 40,000 & 160,000 & 29 &   128 \\
        (7)$^\star$ & (40, 28) & 0.7 & 10,000 & 40,000 & 29 & 128 \\
        \hline
        (8) & (60, 30) & 0.5 & 500,000 & 500,000 & 313 & 1 \\
        \hline
        (9) & (60, 18) & 0.3 & 40,000 & 160,000& 256 & 32 \\
        (10) & (40, 28) & 0.7 & 40,000 & 160,000 & 261 & 32 \\
        \hline
    \end{tabular}
    \caption{Specific parameters used to run each Monte Carlo experiment. Experiments (1)-(3) are conducted for both the Rademacher prior case and the non-symmetric prior case ($P_X = \frac14\delta_{\sqrt{3}} + \frac34\delta_{-1/\sqrt{3}}$) with the same parameters. The rest of the experiments are conducted only for the Rademacher prior case.}
    \label{tab:exp_data}
    \end{minipage}%
    \begin{minipage}[t]{.3\textwidth}
    \centering
        \begin{tabular}{|c||c |}
        \hline
    \begin{tabular}{@{}c@{}}Figure \\ reference \end{tabular} & \begin{tabular}{@{}c@{}}Experiment \\ No. \end{tabular} \\
    \hline \hline
      \ref{fig:ExactMIfiniteSize}, \ref{fig:MCMCvsRIE}, \ref{fig:MCMC_NonSym}  &  (1,2,3) \\
      \ref{fig:MMSE_log} & (4,5,6) \\
            \ref{fig:MMSE_th} ($\alpha = 0.3$) & (9) \\
      \ref{fig:MMSE_th} ($\alpha = 0.5$) & \begin{tabular}{@{}c@{}} (3) (for $\lambda < 8$) \\ (6)  (for $\lambda > 8$)  \end{tabular} \\
      \ref{fig:MMSE_th} ($\alpha = 0.7$) & \begin{tabular}{@{}c@{}} (10) (for $\lambda < 12$) \\ (7)  (for $\lambda > 12$)  \end{tabular} \\
      \ref{fig:offDiag} & (1,2,3) \\
      \ref{fig:transOverlapSquare}, \ref{fig:X_MSE}, \ref{fig:histo}, \ref{fig:reordering_finitesize}, \ref{fig:OptimalRot} & (8) \\
      \hline
    \end{tabular}
    \caption{Table of corresponding experiments used to generate each figure in this paper. }
    \label{tab:fig_exp_ref}
    \end{minipage}
\end{table}

The specific number of MCMC steps, the number of replicas $R$, as well as the number of random instances of $(\bX, \bZ)$ used for the simulations are given in 
Table~\ref{tab:exp_data}. Table~\ref{tab:fig_exp_ref} 
specifies which data was used to generate each figure in this paper. 
The precise values of $\{\lambda_r\}_{r=0}^R$ are chosen such that the exchange rate between adjacent SNRs is always finite (typically larger than 0.3) by using the adaptive algorithm given by \cite{Syed22nonrev}. 
This is with the exception of experiments designated with the star in Table~\ref{tab:exp_data}, whose SNR values are given by equispaced values in the range 8 to 15. 

As indicated in Table~\ref{tab:fig_exp_ref}, the MCMC curves in \figurename~\ref{fig:MMSE_th} are obtained by augmenting data from two different types of experiments: one from the usual replica exchange MCMC ((3), (11)), and the other from MCMC specifically run under high SNR, (7) and (8). Note that with only the results from (3) and (11), the mean of the MMSE will become hidden in the statistical error from finite number of random instances. The results from high SNR, (7) and (8), are given to clarify the exponential decay of the MMSE, which have significantly lower statistical error due to the larger number of random instances we can afford to average over. 

\vspace{5pt}
\textbf{Experiments for \figurename~\ref{fig:multiscalevsMCMC}.} \ \ The experimental procedure to obtain \figurename~\ref{fig:multiscalevsMCMC} is very different from that to obtain the other ones. Unlike the other experiments, here we only consider MCMC without replica exchange moves, i.e. independent MCMC chains are simply run for each value of SNR. 
The informative points (black markers) are obtained by initializing a MCMC chain on the ground truth, and performing $50000 \times 30NM$ local spin flip moves for each SNR. This is repeated on 36 random realizations of $(\bX, \bZ)$ to evaluate the disordered average. The uninformative points (red markers) are obtained by initializing each MCMC chain to a random configuration, and repeating the same process as that to obtain the informative points. Last, the averaged uninformative points, indicated by the blue markers, are obtained by running 32 individual MCMC chains with random initialization for each SNR and $(\bX, \bZ)$. The estimator for $\bX \bX^\intercal$ is obtained by averaging over all 32 runs with different initializations. More concretely, for a given realization of $(\bX, \bZ)$, if we let $\bx_i(t) \in \{-1, +1\}^{N\times M}$ be the configuration of the MCMC chain after $30NMt$ local spin flip moves were performed from a random initialization indexed by $i$, the estimate of $\bX \bX^\intercal$ is given by 
\begin{equation*}
    \langle \bx \bx^\intercal \rangle_{\text{ uninfo. averaged}} = \frac{1}{32\times 4000} \sum_{i = 1}^{32} \sum_{t = 1001}^{5000} \bx_i(t) \bx_i(t)^\intercal. 
\end{equation*}
This procedure is then repeated over 36 random realizations of $(\bX, \bZ)$ to evaluate the disordered average. Note that for each initialization, only $5000 \times 30NM$ local spin flip moves were performed in total, since we found after performing the experiments for uninformative points (red markers) that the MCMC chains get stuck to a random spin configuration typically after $\sim 100 \times 30NM$ local spin flip moves for SNR larger than 7. Note that this is different from the procedure to obtain uninformative points (red markers), where for a given realization of $(\bX, \bZ)$, the estimate is given only by a single initialization:
\begin{equation*}
    \langle \bx \bx^\intercal \rangle_{\text{ uninfo.}} = \frac{1}{40000} \sum_{t = 10001}^{50000} \bx_1(t) \bx_1(t)^\intercal. 
\end{equation*}

\textbf{Initialization of the MCMC chains.}\ \ 
In principle, given sufficient amount of Monte Carlo steps, the MCMC simulations are able to sample from the true equilibrium measure, irrespective of the configuration in which each chain was initialised. Here, the ground truth $\bX$ is chosen as the initial configuration, with the obvious exception of experiment (3). This is done for the following reason. As seen in \figurename~\ref{fig:multiscalevsMCMC}, for low SNR the system converges to the RIE irrespective of initialisation, indicating a benign energy landscape explorable by local spin flips. On the other hand, under high SNR, an MCMC chain with random initialisation encounters a possibly glassy phase, where mixing time increases exponentially with respect to $N$. This resembles a ``golf-course"-type of energy landscape \cite{Machta09}, where it consists of a large number of excited glassy states, and a small number of ground states, which correspond to the ground truth $\bX$. An important observation is that while it is difficult to reach the ground truth from the excited states by using replica exchange and local spin flip moves, one can easily reach excited states from a ground state by exchanging configurations with replicas in low SNR. Therefore, given that the replicas exchange with finite probability in our simulations, replicas with a paramagnetic equilibrium will not remain trapped in the initial ground state for exponentially long periods; 
these replicas can access equilibrium states through exchanges with replicas in low SNR, or via local spin flip moves if the landscape is benign (which is the case in the denoising phase).  
Simultaneously, for the replicas within the factorisation phase, initialising on $\bX$ eliminates the effort to search for rare ground states across the entire configuration space. 
This enables efficient investigation of the problem without introducing significant bias towards informative states in the numerical results.


\section{Details of the multiscale mean-field computation}
\label{app:multiscale_MFT}
The subsequent analysis will rely on special identities inherent to the Bayes-optimality of the setting. As a simple consequence of the tower rule for conditional expectations one can prove the Nishimori identities which will induce numerous simplifications along the way. With a little abuse of notation, we denote by $\langle\,\cdot\,\rangle$ also the infinite product measure $\bigotimes_{n=1}^\infty\EE[\,\cdot\mid\bY]$. Hence, \emph{replicas}, namely conditionally (on $\bY$) independent samples from the posterior \eqref{eq:posterior}, are averaged jointly by the measure $\langle\,\cdot\,\rangle$. Now we can state the following:
\begin{proposition}[Nishimori identity]\label{prop:Nishi_ID}
Let $f$ be a bounded function of the observations $\mathbf{Y}$, the ground truth factor $\mathbf{X}$ and of replicas $(\mathbf{x}^{(k)})_{k\le n}$ drawn independently from the posterior \eqref{eq:posterior}. Then
\begin{align}\label{general_NishiId}
\begin{split}
    &\mathbb{E}_{\mathbf{X},\mathbf{Y}}\langle f(\mathbf{Y};\mathbf{X},\mathbf{x}^{(2)},\dots,\mathbf{x}^{(n)})\rangle=\mathbb{E}_{\mathbf{Y}}\langle f(\mathbf{Y};\mathbf{x}^{(1)},\mathbf{x}^{(2)},\dots,\mathbf{x}^{(n)})\rangle.
\end{split}
\end{align}
\end{proposition}
An elementary proof of this fact can be found in \cite{lelarge2019fundamental}.

\subsection{First scale reduction: \texorpdfstring{$O(MN)\to O(M)$}{} reduction by the cavity method}\label{app:first_scale}
Here we detail more precisely how to devise properly the RLE problem mentioned in Section \ref{sec:multiscale_cavity} via the cavity method. Recall the Hamiltonian \eqref{eq:Hamiltonian} with $N$ rows and $M$ columns and where $Y_{ij}$ has been expanded in signal and noise contributions:
\begin{align}
    &-H_{N,M}(\bx;\bZ,\bX)=\sum_{i<j}^{N}\Big[\sqrt{\frac{\lambda}{N}}Z_{ij}\bx_i\cdot \bx_j\nonumber+\frac{\lambda}{N}(\bX_i\cdot \bX_j)\bx_i\cdot \bx_j-\frac{\lambda}{2N}(\bx_i\cdot \bx_j)^2\Big]
\end{align}
with $\bx_i\cdot\bx_j=\sum_{\mu\le M} x_{i\mu}x_{j\mu}$. Recall also the notation for the row cavity $\bx_N\equiv\beeta\in\mathbb{R}^M$, $\bX_N\equiv\bH$ and the bulk (i.e., all but the cavity) variables $\bar\bx,\bar\bX\in\mathbb{R}^{N-1\times M}$. Now we segregate the contribution of $\beeta$ from the Hamiltonian:
\begin{align}
    &-H_{N,M}(\bar\bX,\beeta;\lambda)= -H'_{N,M}(\bar\bX;\lambda)-H^{\rm row}_{N,M}(\beeta;\bar\bX,\lambda),\nonumber
\end{align}with bulk Hamiltonian
\begin{align}\label{eq:Hbulk} 
        &-H'_{N,M}(\bar\bX;\lambda)=\sum_{i<j}^{N-1}\Big[\sqrt{\frac{\lambda}{N}}Z_{ij}\bx_i\cdot \bx_j +\frac{\lambda}{N}(\bX_i\cdot \bX_j)\bx_i\cdot \bx_j-\frac{\lambda}{2N}(\bx_i\cdot \bx_j)^2\Big],    
\end{align}
and cavity Hamiltonian
\begin{align}       \label{eq:Hrow} 
    &-H^{\rm row}_{N,M}(\beeta;\bar\bX,\lambda)=\sum_{i=1}^{N-1}\Big[\sqrt{\frac{\lambda}{N}}Z_{iN}\bx_i\cdot \beeta +\frac{\lambda}{N}(\bX_i\cdot \bH)\bx_i\cdot \beeta-\frac{\lambda}{2N}(\bx_i\cdot \beeta)^2\Big].
\end{align}

In order to lighten the notation we have dropped dependencies on the ground truth and the noise.
From $-H^{\rm row}_{N,M}$ we recognize (up to a constant) the log-likelihood of an RLE with uncertain design inference problem over both $(\bar\bX,\bH)$ with data $\tilde \bY(\lambda) = \bY_N =(\tilde Y_i=Y_{Ni})_{i\le N-1}$ generated as in \eqref{eq:RLE_channel_cavity}, or more generally as
\begin{align}\label{eq:RLE_GLM_def}
    \tilde Y_i\sim P_{\rm out}\Big(\,\cdot\mid\frac{\bX_i\cdot\bH}{\sqrt{N}} \Big), \quad i\le N-1,
\end{align}
with output channel
\begin{align}\label{eq:Gaussian_output_channel}
    P_{\rm out}(y\mid x):=\frac{1}{\sqrt{2\pi}}e^{-\frac{1}{2}(y-\sqrt{\lambda}x)^2}.
\end{align} 
In addition, we also have side information about the bulk variables coming from the observations \eqref{eq:RLE_channel_bulk}. 

We are going to work in this general setting and specialise to the Gaussian linear channel \eqref{eq:Gaussian_output_channel} at the end. Note that with these notations, the mutual information $I(\tilde\bY(\lambda), \bY'(\zeta);\bar\bX,\bH)$ is equal to the original one \eqref{eq:MI_and_free_entropy} when $\zeta=\lambda$. 

Let us also introduce 
\begin{align*}    
\langle\,\cdot\,\rangle'=\frac{1}{\mathcal{Z}'(\bY'(\zeta))}\int_{\mathbb{R}^{N-1\times M}} dP_X(\bar\bx) e^{-H'_{N,M}(\bar\bx;\zeta)}(\,\cdot\,)
\end{align*}with proper normalisation $\mathcal{Z}'$. $\langle\,\cdot\,\rangle'$ is the posterior measure associated to the channel \eqref{eq:RLE_channel_bulk} with SNR $\zeta$, and it is the one averaging over the bulk degrees of freedom.

\subsection{Second scale: \texorpdfstring{$O(M)\to O(1)$}{} reduction by bulk ansatz and replica formulas}\label{app:second_scale}
We recall that within the mean-field ansatz \eqref{eq:aux_channel_bulk} the simplified $\sigma$-dependent bulk measure reads
\begin{align*}
\langle\,\cdot\,\rangle_{\rm eff}'=\prod_{i,\mu=1}^{N-1,M}\frac{1}{\mathcal{Z}'(Y'_{i\mu})}\int dP_X(x_{i\mu})e^{\sqrt{\sigma} Y'_{i\mu}x_{i\mu}-\frac{1}{2}\sigma x_{i\mu}^2}(\,\cdot\,) ,
\end{align*} where $\mathcal{Z}'(Y'_{i\mu})$ is the proper normalisation per bulk entry.
Under this simplified bulk measure, we define the free entropy $\phi^{\rm RLE}_{{\rm eff},N}$ similarly to \eqref{eq:RLE_entropy}:
\begin{align*}
    &\phi^{\rm RLE}_{{\rm eff},N}(\lambda,\sigma)=\phi^{\rm RLE}_{{\rm eff},N} :=-\frac{1}{\alpha}\EE\ln P_{\rm out}(\tilde Y_1\mid 0)+\frac{1}{M}\EE_{\tilde \bY(\lambda),\bY'_{\rm eff}(\sigma)}\ln \int_{\mathbb{R}^M} dP_X(\beeta)\Big\langle \prod_{i=1}^{N-1}P_{\rm out}\Big(\tilde Y_i\mid \frac{\bx_i\cdot\beeta}{\sqrt{N}} \Big)\Big\rangle_{\rm eff}'.
\end{align*}
For the first term we have used the fact that $P_{\rm out}(\tilde\bY\mid 0)$ factorises over the elements of $\tilde\bY$, giving a sum of contributions that are all equal in average. Again, the different symbol for the above quantity stresses the fact that we have used our mean-field ansatz. Notice also that we are using the more generic definition \eqref{eq:RLE_GLM_def}. There is an immediate simplification due the factorisation of the bulk measure:
\begin{align*}
    &\phi^{\rm RLE}_{{\rm eff},N}=-\frac{1}{\alpha}\EE\ln P_{\rm out}(\tilde Y_1\mid 0)+\frac{1}{M}\EE_{\tilde \bY(\lambda),\bY'_{\rm eff}(\sigma)}\ln\int dP_X(\beeta)\prod_{i=1}^{N-1}\Big\langle P_{\rm out}\Big(\tilde Y_i\mid \frac{\bx_i\cdot\beeta}{\sqrt{N}} \Big)\Big\rangle_{\rm eff}'
\end{align*}
where we were able to factorise the $i$-product.

Under the factorised measure $\langle\,\cdot\,\rangle'_{\rm eff}$, the rescaled sums $\bx_i\cdot\beeta/\sqrt{N}$ behave as Gaussians by the central limit theorem, for any fixed $\beeta$. Therefore, it is now equivalent to think of the $x_{i\mu}$'s as independent Gaussians with mean and variance 
\begin{align*}
  m(Y'_{i\mu})&:=\langle x_{i\mu}\rangle'_{\rm eff},\quad V(Y'_{i\mu}):=\langle x^2_{i\mu}\rangle'_{\rm eff}-(\langle x_{i\mu}\rangle'_{\rm eff})^{2}  ,
\end{align*}
or equivalently
\begin{align}\label{eq:Gaussian_equivalent_design}
    x_{i\mu}\distreq m(Y'_{i\mu})+\xi_{i\mu}\sqrt{V(Y'_{i\mu})},
\end{align}
with $\xi_{i\mu}\iid\mathcal{N}(0,1)$ the new annealed bulk degrees of freedom. It is now the moment to define $J(\sigma)$ in \eqref{eq:def_J} as the overlap for the bulk variables
\begin{align*}
    J&=J(\sigma):=\EE \,m(Y'_{i\mu})X_{i\mu}=\EE_{X,Z'} \Big[X\frac{\int dP_X(x) \,x\,e^{(\sqrt{\sigma}Z'+X\sigma)x-\frac{1}{2}\sigma x^2}}{\int dP_X(x)e^{(\sqrt{\sigma}Z'+X\sigma)x-\frac{1}{2}\sigma x^2}}\Big]
\end{align*}
with $Z'\sim\mathcal{N}(0,1), X\sim P_X$. By the Nishimori identity it is also equal to $J=\EE \,m(Y'_{i\mu})^2$. Hence, using the properties of the prior, we have that $(m(Y'_{i\mu}))_{i,\mu}$ are i.i.d. sub-Gaussian random variables with
\begin{align}
\EE\, m(Y'_{i\mu})=0, \quad   \EE\, m(Y'_{i\mu})^2=J.\label{momentsm}
\end{align} 
Rewriting $\phi_{{\rm eff},N}^{\rm RLE}$ by means of \eqref{eq:Gaussian_equivalent_design} we get 
\begin{align*}
   \phi_{{\rm eff},N}^{\rm RLE}=\frac{1}{M}\EE_{\tilde \bY,\bY'_{\rm eff}}\ln\int dP_X(\beeta)\prod_{i=1}^{N-1}\EE_{\bxi_i} P_{\rm out}\Big(\tilde Y_i\mid\frac{m(\bY_i')\cdot\beeta}{\sqrt{N}}+\frac{(\sqrt{V(\bY'_i)}\circ\bxi_i)\cdot\beeta}{\sqrt{N}}\Big)-\frac{1}{\alpha}\EE\ln P_{\rm out}(\tilde Y_1\mid 0)+o_N(1)
\end{align*}
with $\bxi_i=(\xi_{i\mu})_{\mu\leq M}$, $m(\bY'_i)=(m(Y'_{i\mu}))_\mu$ and similarly for $V(\bY'_i)$, and $\circ$ denotes component-wise product. The remainder is due to the fact that $i$ runs to $N-1$ and not to $N$. It will be dropped in the following.

Being in a Bayes-optimal setting,  strong concentration effects allow us to simplify $((\sqrt{V(\bY'_i)/N}\circ\bxi_i)\cdot \beeta)_i$. Indeed, conditionally on $\beeta$ these are independent centered Gaussian random variables with variance
\begin{align*}
  \frac1N \sum_{\mu=1}^M V(Y'_{i\mu}) \eta^2_\mu= \alpha \EE V(\sqrt{\sigma}X+Z')+o_N(1)
\end{align*}where the equality holds with overwhelming probability when $N,M\to\infty$. We stress that we have used $\|\beeta\|^2/M= 1+o_N(1)$. $\beeta$ is a sample from a posterior, so this replacement is a priori not guaranteed, but thanks to Bayes-optimality $\|\beeta\|^2/M$ concentrates onto its mean which is the same as for the ground truth by the Nishimori identity. Still using the Nishimori identity one gets directly that $\EE V(Y'_{i\mu})=1-J$. Therefore, by defining the noisier channel
\begin{align*}
    \tilde P_{\rm out}^{(\sigma,\alpha)}(y\mid x):=\EE_{\xi\sim\mathcal{N}(0,1)}P_{\rm out}\Big(y\mid x+\xi\sqrt{\alpha(1-J(\sigma))}\Big)
\end{align*}and collecting all the above observations, up to negligible corrections the free entropy rewrites as
\begin{align*}
    &\phi_{{\rm eff},N}^{\rm RLE}=-\frac{1}{\alpha}\EE\ln P_{\rm out}(\tilde Y_1\mid 0)+\frac{1}{M}\EE_{\tilde \bY,\bY'_{\rm eff}}\ln\int dP_X(\beeta)\prod_{i=1}^{N-1}\tilde P_{\rm out}^{(\sigma,\alpha)}\Big(\tilde Y_i\mid\frac{m(\bY'_i)\cdot \beeta}{\sqrt{N}}\Big).
\end{align*}
The first contribution is simple to compute since $\tilde Y_1\sim P_{\rm out}(\,\cdot\mid\bX_1\cdot\bH/\sqrt{N})$ is asymptotically equal in law to $\tilde Y_1\sim P_{\rm out}(\,\cdot\mid\sqrt{\alpha}\,Z')$ with $Z'\sim\mathcal{N}(0,1)$. A final simplification is to replace the i.i.d. centered sub-Gaussian variables $(m(Y'_{i\mu}))$ playing the role of quenched covariates by Gaussians $(\sqrt{J}\tau_{i\mu})$ with $\tau_{i\mu}\iid \mathcal{N}(0,1)$, using \eqref{momentsm}. This leaves the free entropy unchanged at leading order because the replica computation for the GLM and its rigorous proof simply leverage on the fact that $(m(\bY_i')\cdot \beeta/\sqrt{N})_{i\le N}$ behave as Gaussian random variables under the quenched expectation w.r.t.\ $\bY'$, and these have the same law in the large system limit as $(\btau_i\cdot \beeta\sqrt{J/N})_i$. Therefore, we have
\begin{align*}
    &\phi_{{\rm eff},N}^{\rm RLE}=-\frac{1}{\alpha}\EE\ln P_{\rm out}(\tilde Y_1\mid 0)
    +\frac{1}{M}\EE_{\tilde \bY,\btau}\ln\int dP_X(\beeta)\prod_{i=1}^{N-1}\tilde P_{\rm out}^{(\sigma,\alpha)}\Big(\tilde Y_i\mid\sqrt{J}\frac{\btau_i\cdot \beeta}{\sqrt{N}}\Big)
\end{align*}
at leading order. It is finally in the familiar form of free entropy of a GLM with quenched Gaussian covariates matrix, whose rigorous limiting value is found in \cite{BarbierGLM-PNAS}:
\begin{align*}
    &\phi_{{\rm eff}}^{\rm RLE}=\lim_{N\to\infty}\phi_{{\rm eff},N}^{\rm RLE}={\rm extr}_+\Big\{-\frac{rq}{2}+\EE_{\xi,H}\ln\int  dP_X(\eta) e^{(\xi\sqrt{r}+Hr)\eta-\frac{1}{2}r\eta^2}-\frac{1}{\alpha}\EE_{\xi}\int dy  P_{\rm out}(y\mid \sqrt{\alpha}\,\xi)\ln  P_{\rm out}(y\mid 0)\nonumber\\
    & +\frac{1}{\alpha}\EE_{\xi,u_0}\int dy \tilde P_{\rm out}^{(\sigma,\alpha)}(y\mid \sqrt{\alpha J(\sigma)(1-q)}\,u_0+\sqrt{\alpha J(\sigma)q}\,\xi)\ln \EE_u\tilde P_{\rm out}^{(\sigma,\alpha)}(y\mid \sqrt{\alpha J(\sigma)(1-q)}\,u+\sqrt{\alpha J(\sigma)q}\,\xi)\Big\}
\end{align*}    
where extremization is intended w.r.t.\ $(r,q)$, and $(\xi,u_0,u)$ are i.i.d. standard Gaussian variables, $H\sim P_X$. ${\rm extr}_+$ selects the solution of the saddle point equations, obtained by equating to zero the gradient of the replica potential $\{\cdots\}$, which maximizes it. If we specify the above to the Gaussian output channel \eqref{eq:Gaussian_output_channel} we readily obtain \eqref{eq:RLE_RS_free_entropy}.

\section{Properties of the mean-field mutual information}
\label{app:MI_reconstruction_potential}
\subsection{Reconstruction of \texorpdfstring{\eqref{eq:MI_Matrix_facto_pot}}{} from the saddle point equations}
We are looking for an MI potential, denoted in the following as $\iota (r,q;\alpha,\lambda)$, generating the saddle point equations \eqref{eq:Kabashima_r}. We start by imposing the I-MMSE relation \eqref{eq:I-MMSE}:
\begin{align}
    \frac{d}{d\lambda}\iota (r_*,q_*;\alpha,\lambda)=\frac14\frac{1-q_*(\alpha,\lambda)^2}{1+\lambda\alpha(1-q_*(\alpha,\lambda)^2)}.\label{IMMSEapprox}
\end{align}
We observe that \eqref{eq:Kabashima_r} defines a vector field
\begin{align*}
    v_r(r,q):=q-\EE H\langle \eta \rangle_{r},\quad 
    v_q(r,q):=r-\frac{\lambda q}{1+\lambda\alpha(1-q^2)}.
\end{align*}
It is conservative, because it is irrotational on a simply connected domain:
\begin{align*}
    \frac{\partial}{\partial q}v_r(r,q)=1=\frac{\partial}{\partial r}v_q(r,q).
\end{align*}
Therefore, any line integral of this vector field starting from $(0,0)$ and ending at an arbitrary point $(r,q)$ will yield the potential, up to a constant shift and an arbitrary common rescaling of the vector field components. We choose the integration path $(0,0)\to (r,0)\to (r,q)$:  
\begin{align*}
    \iota(r,q;\alpha,\lambda)K+C&=\int_0^r v_r(s,0) ds+\int_0^q v_q(r,p) dp\nonumber\\
    &=-2\,\EE_{\xi,H}\ln\int dP_X(\eta)e^{(\xi\sqrt{r}+Hr)\eta-\frac{1}{2}r\eta^2}+rq+\frac{1}{2\alpha}\ln\big(1+\lambda\alpha(1-q^2)\big)
\end{align*}where we have introduced two constants $K$ and $C$ to be fixed later. The term $\EE H\langle \eta \rangle_{s}=1-{\rm mmse}(H\mid \sqrt{s}H+\xi)$ has been integrated w.r.t. $s$ using the I-MMSE relation for a scalar Gaussian channel \cite{Verdu_I-MMSE}. It is easy to check that this potential generates the correct saddle point equations. Now we evaluate it at their solutions. This will yield a MI whose basic properties will help us to fix $K$ and $C$. In particular, at $\lambda=0$ one has $q_*=r_*=0$, and $\iota(r_*,q_*;\alpha,0)=0$, which directly entails $C=0$. In order to fix $K$ we compute the derivative
\begin{align*}
    \frac{d}{d\lambda}\iota(r_*,q_*;\alpha,\lambda)=\frac{1}{2K}\frac{1-q_*(\alpha,\lambda)^2}{1+\lambda\alpha(1-q_*(\alpha,\lambda)^2)}.
\end{align*}The I-MMSE relation \eqref{IMMSEapprox} then requires $K=2$. This determines \eqref{eq:MI_Matrix_facto_pot}.

\subsection{Large SNR limit of \texorpdfstring{\eqref{eq:MI_Matrix_facto}}{} for discrete priors}\label{app:large-lambda-limit-MI}
With a discrete prior one can see that $1-q^2= o_\lambda(\lambda^{-1})=1-q$. This must be verified in order for $r$ in \eqref{eq:Kabashima_r} to diverge as $\sim\lambda$, so that $q$ indeed saturates to $1$ in the infinite SNR limit. Furthermore, the $\EE_{\xi,H}\dots$ term of \eqref{eq:MI_Matrix_facto_pot} is related to a Bayes-optimal scalar Gaussian channel with SNR $r$. Hence, when $r\sim\lambda\to\infty$, the integral over the prior can be done by saddle point and must select $\eta=H$:
\begin{align*}
    \EE_{\xi,H}\ln\int dP_X(\eta)&e^{(\xi\sqrt{r}+Hr)\eta-\frac{1}{2}r\eta^2}=\frac{r}{2}+\EE_H\ln\int dP_X(\eta)\mathbbm{1}(\eta=H)+o_\lambda(1).
\end{align*}
Here $\mathbbm{1}(\eta=H)$ is the indicator function of the event $\eta=H$ (recall that $H$ is discrete in this derivation). The integral on the second line is thus precisely $-\mathcal{H}(X)$. Plugging this into \eqref{eq:MI_Matrix_facto_pot} we get
\begin{align*}
    &\iota(\alpha,\lambda)=
    \frac{r(q-1)}{2}+o_\lambda(1)+\frac{1}{4\alpha}\ln\big(1+o_\lambda(1)\big)+
    \mathcal{H}(X).
\end{align*}
When $\lambda\to\infty$, $r\sim\lambda$ at leading order and both $\lambda(1-q^2),\lambda(1-q)$ are $o_\lambda(1)$. Hence
$$\lim_{\lambda\to\infty}\iota(\alpha,\lambda)=\mathcal{H}(X).$$

\section{Mean-square error of the rotational invariant estimator when \texorpdfstring{$\alpha\to 0$}{}}\label{appendix:alpha-to-0-RIE-MSE}
Consider the usual problem,
\begin{align*}
    \frac{\bY}{\sqrt{N}}=\frac{\sqrt{\lambda}}{N}\bX\bX^\intercal +\frac{\bZ}{\sqrt{N}}
\end{align*}where we rescaled the observations so that they have $O(1)$ eigenvalues and a well defined, finitely supported, asymptotic spectral density $\rho_Y$. Standard random matrix theory \cite{potters2020first} shows that  $\rho_Y=\rho_{\rm s.c.}\boxplus\rho_{\sqrt{\lambda} {\rm MP}}(\,\cdot\,;\alpha)$ is the free convolution between the semicircular density and the spectral density of a Wishart matrix (Marchenko-Pastur distribution of parameter $1/\alpha$) multiplied by $\sqrt{\lambda}$. We recall once more the shrinking procedure of \cite{BABP} needed to get the eigenvalues of the RIE:
\begin{align}\label{eq:BABP_shirnkage_appendix}    \xi_i=\xi_i(\bY):=\frac{\gamma_{\bY,i}-2\pi\mathsf{H}[\rho_Y](\gamma_{\bY,i})}{\sqrt{\lambda}}
\end{align}
where $(\gamma_{\bY,i})$ are the eigenvalues of $\bY/\sqrt{N}$, and
\begin{align*}
    \pi \mathsf{H}[\rho_Y](x)=\text{P.V.}\int dy\frac{\rho_Y(y)}{x-y}=\lim_{\epsilon\to 0^+} \text{Re} \,g_Y(x+i\epsilon).
\end{align*}
$g_Y$ is the Stieltjes transform of the spectral density $\rho_Y$ \cite{potters2020first}:
\begin{align*}
    g_Y(z):=\int dy\frac{\rho_Y(y)}{z-y},\quad z\in\mathbb{C}\setminus\text{Supp}(\rho_Y).
\end{align*}
We consider the low-rank limit $\alpha\to 0$. For simplicity we set $M=1$, but the argument below can be extended to any $M$ finite, and we claim it can be extended to sub-linear rank regimes $M=\alpha N^\gamma$ with $\gamma<1$. Under this hypothesis, $g_Y=g_{Z}$, since low-rank perturbations are not enough to modify the spectral density of the noise. $\bY$ then turns into a low-rank perturbation of a full-rank Wigner matrix. Hence, some eigenvalues (one for $M=1$) will pop out of the bulk spectrum of the noise. 

The eigenvalue(s) that pops out, for $\lambda >1$, is located at $z^*=\sqrt{\lambda}+\frac{1}{\sqrt{\lambda}}$ according to our notation \cite{baik2005phase}.  The Stietljes transform of the observations evaluated at $z^*$ is thus
\begin{align*}
    g_Y(z^*)\xrightarrow[]{\alpha\to 0}g_{Z}(z^*)=\frac{1}{2}\Big(z^*-z^*\sqrt{1-4/z^{*2}}\Big)=\frac{1}{2}\Big(\sqrt{\lambda}+\frac{1}{\sqrt{\lambda}}-\sqrt{\lambda-2+\frac{1}{\lambda}}\Big)=\frac{1}{\sqrt{\lambda}},
\end{align*}
and hence
\begin{align*}
    \pi \mathsf{H}[\rho_Y](z^*)\xrightarrow[]{\alpha\to 0}\frac{1}{\sqrt{\lambda}}.
\end{align*}
From the above, we can infer the shrinkage of the leading eigenvalue:
\begin{align*}
    \xi_1=\frac{1}{\sqrt{\lambda}}\Big(\sqrt{\lambda}+\frac{1}{\sqrt{\lambda}}-\frac{2}{\sqrt{\lambda}}\Big)=1-\frac{1}{\lambda}.
\end{align*}

All other eigenvalues are flushed to $0$. In fact, we can safely assume that under low-rank perturbations of the noise bulk, all eigenvalues expect the leading one lie inside the bulk. In that case, thanks to the identity
\begin{align*}
    \text{P.V.}\int dy\frac{\rho_Z(y)}{x-y}=\frac{x}{2},\quad x\in\R,
\end{align*}
it is straightforward to see that \eqref{eq:BABP_shirnkage_appendix} yields $0$ for all such eigenvalues.

Therefore the estimator given by the RIE in the low rank limit (recall $M=1$) would read
\begin{align*}
    \Big\langle \frac{\bX\bX^\intercal}{N}\Big\rangle_\bY= \bu_{\bY}\bu_{\bY}^\intercal \Big(1-\frac{1}{\lambda}\Big),
\end{align*}where $\bu_{\bY}$ is the unit eigenvector corresponding to the leading eigenvalue $z^*$ of $\bY$. This estimator is also Bayes-optimal when the prior on $\bX$ is rotational invariant, and the MSE attained is 
\begin{align*}
    \text{MSE}=1-\Big(1-\frac{1}{\lambda}\Big)^2.
\end{align*}
This MSE is attained by the RIE regardless of the prior, and is equal to that of PCA \cite{lelarge2019fundamental}, which is sub-optimal for general prior other than Gaussian. This argument thus supports the sub-optimality of the RIE in the limit $\alpha\to 0$, unless there is a discontinuous behavior in $\alpha$.

\section{The Sakata and Kabashima replica approach to \texorpdfstring{\eqref{eq:MI_Matrix_facto}}{}}\label{appendix:Kabashima_derivation}
Here we adapt to our model the method of \cite{SK13EPL} developed originally for dictionary learning in an optimisation setting, then concurrently extended in \cite{SK13ISIT,KMZ_DL-2013,Marc-Kabashima} to the Bayesian case. In order to match notations, we consider the generic output channel
\begin{align*}
    Y_{ij}\sim P_{\rm out}\Big(\cdot\mid\frac{\bX_i\cdot\bX_j}{\sqrt{N}}\Big),\quad 1\leq i<j\leq N.
\end{align*}
The free entropy we are interested in computing, to be later connected to the MI, is
\begin{align*}
    \tilde\phi_N:=\frac{1}{MN}\EE\ln \tilde{\mathcal{Z}}(\bY)-\frac{1}{MN}\EE\ln \prod_{i<j} P_{\rm out}(Y_{ij}\mid 0)
\end{align*}
where
\begin{align*}
    \tilde{\mathcal{Z}}(\bY):=\int dP_X(\bx) \prod_{i<j,1}^N P_{\rm out}\Big(Y_{ij}\mid\frac{\bx_i\cdot\bx_j}{\sqrt{N}}\Big).
\end{align*}
We proceed to the computation of the log partition function with the replica method:
\begin{align*}
    \EE\tilde{\mathcal{Z}}^n(\bY)=\int \prod_{i<j}^N d Y_{ij}\int \prod_{a=0}^n dP_X(\bx^a) \prod_{i<j,1}^N P_{\rm out}\Big(Y_{ij}\mid\frac{\bx^a_i\cdot\bx^a_j}{\sqrt{N}}\Big),
\end{align*}where we set $\bx^0=\bX$.
We now introduce a constraint on the overlaps
\begin{align*}
    \EE \tilde{\mathcal{Z}}^n(\bY)\propto\int\prod_{i<j}^N d Y_{ij} \int \prod_{a=0}^n dP_X(\bx^a) \int\prod_{{a\leq b=0}}^n dQ_{ab}\delta\big(NMQ_{ab}-\Tr[\bx^a\bx^{b\intercal}]\big) \prod_{a=0}^n\prod_{i<j,1}^N P_{\rm out}\Big(Y_{ij}\mid\frac{\bx^a_i\cdot\bx^a_j}{\sqrt{N}}\Big).
\end{align*}
Now we aim at identifying the joint law of the variables
\begin{align*}
    z_{ij}^a=\frac{\bx^a_i\cdot\bx^a_j}{\sqrt{N}},\quad 1\leq i<j\leq N.
\end{align*}
More precisely, let us introduce a delta function to isolate these variables
\begin{align*}
    \int d\bY\int\prod_{{a\leq b=0}}^n dQ_{ab}\int \prod_{a=0}^n dP_X(\bx^a)\prod_{{a\leq b=0}}^n\delta\big(NMQ_{ab}-\Tr[\bx^a\bx^{b\intercal}]\big)\prod_{i<j,1}^N\prod_{a=0}^n\delta\Big(z_{ij}^a-\frac{\bx^a_i\cdot\bx^a_j}{\sqrt{N}}\Big)   P_{\rm out}\big(Y_{ij}\mid z_{ij}^a\big)
\end{align*}
and let us isolate their non normalised probability measure for a given value of the overlaps $Q_{ab}$:
\begin{align*}
    \int \prod_{a=0}^n dP_X(\bx^a)\prod_{{a\leq b=0}}^n\delta\big(NMQ_{ab}-\Tr[\bx^a\bx^{b\intercal}]\big)\prod_{i<j,1}^N\prod_{a=0}^n\delta\Big(z_{ij}^a-\frac{\bx^a_i\cdot\bx^a_j}{\sqrt{N}}\Big).
\end{align*}
The key assumption is then to consider the $z$'s as a Gaussian family with zero mean and a covariance equal to
\begin{align*}
    \EE_{z} z_{ij}^az_{kl}^b=\alpha Q^2_{ab}\delta_{ik}\delta_{jl}.
\end{align*}
This \enquote{Gaussian ansatz} allows to move forward with the replica approach. With the above covariance structure, coupled only in replica space, we can write
\begin{align*}
    P_z((z^a_{ij})_{a\leq n} \mid (Q_{ab})_{a,b=0}^n)=\frac{1}{\sqrt{(2\pi)^{n+1}\det \mathcal{T}}}\exp\Big[
    -\frac{1}{2}\sum_{a,b=0}^nz^a_{ij}z^b_{ij}(\mathcal{T}^{-1})_{ab}
    \Big]
\end{align*}with $\mathcal{T}:= (\alpha Q_{ab}^2)_{a,b=0}^n$. Hence we can replace
\begin{align*}
    \int \prod_{a=0}^n dP_X(\bx^a)\prod_{{a\leq b=0}}^n\delta\big(NMQ_{ab}-\Tr[\bx^a\bx^{b\intercal}]\big)\prod_{i<j,1}^N\prod_{a=0}^n\delta\Big(z_{ij}^a-\frac{\bx^a_i\cdot\bx^a_j}{\sqrt{N}}\Big) \approx P_z((z^a_{ij})^{a\leq n}\mid  (Q_{ab})_{a,b=0}^n)e^{\ln V(Q)}
\end{align*}where
\begin{align*}
    V(Q):=\int \prod_{a=0}^ndP_X(\bx^a)\prod_{{a\leq b=0}}^n\delta\big(NMQ_{ab}-\Tr[\bx^a\bx^{b\intercal}]\big),
\end{align*}is a normalisation.
Therefore the replicated partition function reads
\begin{align*}
    \EE\tilde{\mathcal{Z}}^n(\bY)\propto \int\prod_{{a\leq b=0}}^n dQ_{ab} \exp\big(\ln V(Q)\big)\Big[\EE_{\bz\sim\mathcal{N}(0;\mathcal{T})}\int dY \prod_{a=0}^n P_{\rm out}\big(Y\mid z^a\big)\Big]^{N^2/2}.
\end{align*}Letting $N$ and $M$ diverge we get the following expression by saddle point w.r.t.\ the integration variables $(Q_{ab})_{0\leq a,b\leq n}$:
\begin{align*}
    \frac{1}{MN}\ln \EE\tilde{\mathcal{Z}}^n(\bY)={\rm extr}_+\Big\{\frac{1}{MN}\ln V(Q)+\frac{1}{2\alpha}\ln \EE_{\bz\sim\mathcal{N}(0;\mathcal{T})}  \int dY \prod_{a=0}^n P_{\rm out}\big(Y\mid z^a\big)\Big\}+o_N(1).
\end{align*}
Now we specify our saddle point to the replica symmetric ansatz: $\mathcal{T}_{RS}=\alpha (1-q^2)\delta_{ab}+\alpha q^2$, and we then send $n\to 0$. Here we have further imposed that $Q_{a0}=q=Q_{ab}$ for any $a\neq b$, and $v=Q_{aa}$ for all $a=0,\dots,n$, which amounts to impose the Nishimori identities on the solution of the saddle point. This is justified by the the Bayes-optimal setting. The computation of $V(Q)$ also requires the introduction of Fourier delta representations. These steps are standard, and we defer the reader to \cite{SK13ISIT} for this computation, since for what $V(Q)$ is concerned, it is identical. In the end one gets:
\begin{align*}
    \frac{1}{MN}\ln V(Q)=-n\,{\rm extr}_+\Big\{\frac{r  q}{2}+\EE\ln\int dP_X(x)\exp\Big[(\sqrt{ r }Z+ r  X)x-\frac{ r }{2}x^2\Big]\Big\}+O(n^2)
\end{align*}
where extremization is w.r.t.\ $r$. The other term is easily dealt with:
\begin{align*}
    \frac{1}{\alpha}&\ln \EE_{\bz\sim\mathcal{N}(0;\mathcal{T})}  \int dY \prod_{a=0}^n P_{\rm out}\big(Y\mid z^a\big)=\nonumber\\
    &=\frac{n}{\alpha}\EE_{u_0,\zeta} \int dY P_{\rm out}\big(Y\mid \sqrt{\alpha(v^2-q^2)}u_0+ \sqrt{\alpha q^2}\zeta\big)\ln \EE_{u}P_{\rm out}\big(Y\mid \sqrt{\alpha(v^2-q^2)}u+ \sqrt{\alpha q^2}\zeta\big)+O(n^2).
\end{align*}
We reach
\begin{align*}
    &\frac{1}{MNn}\ln \EE\tilde{\mathcal{Z}}^n(\bY)={\rm extr}_+\Big\{-\frac{r q}{2}+\EE\ln\int dP_X(x)\exp\Big[(\sqrt{ r }Z+ r  X)x-\frac{ r }{2}x^2\Big]\nonumber\\
    &+\frac{1}{2\alpha}\EE_{u_0,\zeta} \int dY P_{\rm out}\big(Y\mid \sqrt{\alpha(1-q^2)}u_0+ \sqrt{\alpha q^2}\zeta\big)\ln \EE_{u}P_{\rm out}\big(Y\mid \sqrt{\alpha(1-q^2)}u+ \sqrt{\alpha q^2}\zeta\big)\Big\}+o_N(1)+O(n),
\end{align*}where extremization is always intended w.r.t.\ $r,q$.
Specifying to Gaussian channel
\begin{align*}
    P_{\rm out}(y\mid x)=\frac{1}{\sqrt{2\pi}}\exp\Big(-\frac{1}{2}(y-\sqrt{\lambda} x)^2\Big).
\end{align*}one gets
\begin{align*}
    {\rm extr}_+\Big\{-\frac{ r  q}{2}+\EE\ln\int dP_X(x)\exp\Big[(\sqrt{ r }Z+ r  X)x-\frac{ r }{2}x^2\Big]-\frac{1}{4\alpha}\ln(1+\lambda\alpha (1-q^2))-\frac{1}{4\alpha}\log 2\pi e\Big\}+O(n)
\end{align*}
and we can now subtract the term
\begin{align*}
    \frac{1}{MN}\EE\ln \prod_{i<j}^NP_{\rm out}(Y_{ij}\mid 0)=\frac{1}{2\alpha}\EE\ln P_{\rm out}(Y\mid 0)+o_N(1)=-\frac{1}{4\alpha}\log 2\pi e-\frac{\lambda }{4}.
\end{align*}
Putting altogether we get the free entropy
\begin{align*}
    \tilde\phi=\lim_{N\to\infty}\tilde \phi_N={\rm extr}_+\Big\{\frac{\lambda}{4}-\frac{r q}{2}+\EE\ln\int dP_X(x)\exp\Big[(\sqrt{ r }Z+ r  X)x-\frac{ r }{2}x^2\Big]-\frac{1}{4\alpha}\ln(1+\lambda\alpha (1-q^2))\Big\}.
\end{align*}
For Gaussian channel, in order to obtain the MI it is sufficient to subtract $\lambda/4$ from the above and flip the sign:
\begin{align*}
    \iota(\alpha,\lambda)={\rm extr}_-\Big\{\frac{rq}{2}+\frac{1}{4\alpha}\ln(1+\alpha\lambda (1-q^2))-\EE\ln\int dP_X(x)\exp\Big[(\sqrt{ r }Z+ r  X)x-\frac{ r }{2}x^2\Big]\Big\},
\end{align*}
where $r$ and $q$ have to be extremized. This formula matches \eqref{eq:MI_Matrix_facto}.

\section{Can the mean-field theory \texorpdfstring{$\iota(\alpha,\lambda)$}{} be improved?}\label{sec:improv}

We have shown that the mean-field equations \eqref{eq:MI_Matrix_facto} based on the bulk ansatz \eqref{eq:aux_channel_bulk} produce excellent results in the factorisation phase, where an MMSE lower than the Gaussian one is possible; they are however only approximative in the denoising phase where $\iota^*$ correct them. One may wonder if it possible to do better by generalizing ansatz \eqref{eq:aux_channel_bulk}. We have performed several attempts keeping the bulk measure factorised over the indices $(i,\mu)$ as needed to justify the computation of the RLE free entropy.

As first trial, one may think that each entry $(i,\mu)$ comes with a different SNR $\sigma_{i\mu}$, whose values are i.i.d. drawn from a certain distribution which becomes the new order parameter. In that case, equations \eqref{eq:RLE_saddlepoint_r} simply contain $\EE_\sigma J(\sigma)$ in place of $J(\sigma)$. In particular, the equation for $r$ forces the distribution of $\sigma$ to be a $\delta_{\sigma_*}$ where again $\sigma_*=r_*$, since the r.h.s.\ is deterministic. Hence this approach recovers the proposed solution.

Another approach is to introduce an SNR profile $(\sigma_\mu)$ such that $Y_{i\mu}'=\sqrt{\sigma_\mu}X_{i\mu}+Z'_{i\mu}$, only labeled by $\mu$ to detect possible inhomogeneities at equilibrium in the \enquote{column index space}. The $(\sigma_\mu)$ are then assumed to be divided in an arbitrary number $C$ of blocks to be sent to $+\infty$ after the thermodynamic limit, in a similar fashion to spatial coupling, see \cite{krzakala2012probabilistic,barbier2017approximate} for similar computations. Specifically $\sigma_{\mu}=\sigma_c$ if $\mu\in\Lambda_c\subset\{1,\dots,M\} $. The computation in Section~\ref{sec:multiscale_cavity} (and in Appendix \ref{app:multiscale_MFT}) still works, except that this time we need to introduce more order parameters $(r_c, q_c)_{c=1,\dots,C}$, where $q_c$ corresponds to the overlap $\frac{C}{M}\sum_{\mu\in\Lambda_c}\eta_\mu H_\mu$. These order parameters then satisfy the equations
\begin{align*}
&q_c=\EE H\langle \eta\rangle_{r_c},\qquad r_c=\frac{\lambda J(\sigma_c)}{1+\lambda\alpha\big(1-\frac{1}{C}\sum_{c=1}^CJ(\sigma_c)q_c\big)},
\end{align*}together with the cavity-bulk consistency equation $\sigma_c=r_c$. This system always admits the paramagnetic solution $\sigma_c=r_c=q_c=0$ for any given set of $c$'s, and the homogeneous solution $(\sigma_c,r_c,q_c)=(\sigma_*,r_*,q_*)$ for all $c$'s, which also coincides with the solution given by the ansatz \eqref{eq:aux_channel_bulk}. No other solution seems to exist. Interestingly, the latter seems to be always the most convenient in terms of MMSE, since we see numerically that in that case all the SNR values $\sigma_c$ get attracted to a higher SNR $\sigma_*$ w.r.t. the case when some $\sigma_c$ are set to $0$. Furthermore, if one interprets the above system of equations as a fixed point equation for $\sigma_c$,
\begin{align*}
    \sigma_c=\frac{\lambda J(\sigma_c)}{1+\lambda\alpha\big(1-\frac{1}{C}\sum_{c=1}^CJ(\sigma_c)^2\big)},
\end{align*}the solutions where some $\sigma_c=0$ are unstable under such iteration which discards them. One can also run a population dynamics \cite{mezard2009information} using the above equation interpreted as a distributional fixed point equation, which thus looks for a solution where the law of the $(\sigma_\mu)$ is the same for all $\mu$, rather than to look for solutions for the ordered profile $(\sigma_c$) shared by all rows. And, again, the final histogram of $(\sigma_\mu)$ ends up peaked around the value of the homogeneous solution $\sigma_*$.

A final attempt we made consists in admitting a block structure also in the $i$-index, i.e. the row index: $\sigma_{i\mu}=\sigma_{\ell c}$ for $i\in\Lambda_\ell,\mu\in\Lambda_c$, with $\ell=1,\dots,L$ and $c$ as before. With this block structure, one can recast the fixed points equations in terms of the same order parameters as above with some slight differences:
\begin{align*}
 &q_{\ell' c}=\EE H\langle \eta\rangle_{r_{\ell' c}},\qquad r_{\ell'c}=\frac{1}{L}\sum_{\ell=1}^L\frac{\lambda J(\sigma_{\ell c})}{1+\lambda\alpha\big(1-\frac{1}{C}\sum_{c=1}^CJ(\sigma_{\ell c})q_c\big)},
\end{align*}
where $\ell'$ corresponds to the index of the row-block from which the cavity was extracted. They are completed with the cavity-bulk consistency $\sigma_{\ell' c}=r_{\ell'c}$. All this must be true for all $\ell'\in \{1,\ldots,L\}$ as the cavity can be extracted from any block. Because the r.h.s. of the $r_{\ell'c}=\cdots$ equation is $\ell'$-independent, solutions must be homogeneous in the row-block indices and we are back to the previous case.

In conclusion, none of these possible generalizations of \eqref{eq:aux_channel_bulk} yields an improvement w.r.t.\ our current ansatz. This strengthen our belief that the mean-field theory is exact in the factorisation phase, as it is robust to richer ansatz as long as the bulk measure is assumed to be factorised.

\end{document}